\documentclass[12pt,a4paper,oneside]{book}

\usepackage{fancyhdr}

\usepackage{hyperref}

\usepackage{graphicx}
\usepackage{ifpdf}

\ifpdf
	\DeclareGraphicsExtensions{.pdf,.png,.jpg,.mps}
	\usepackage{hyperref} 
\else
\fi

\usepackage{myStyleMasters}

\usepackage{etoolbox}
\makeatletter
\patchcmd{\@makechapterhead}{50\p@}{-60pt}{}{}
\patchcmd{\@makeschapterhead}{50\p@}{-60pt}{}{}
\makeatother

\allsectionsfont{\nohang\centering}

\cftpagenumbersoff{part} 

\begin{document}

	\begin{titlepage}
	    \centering
	    \vfill
	    {\bfseries\Huge
	        Comparing Different Mathematical Definitions of 2D CFT\\ }  
	         \vfill
	    \includegraphics[width=0.4\textwidth]{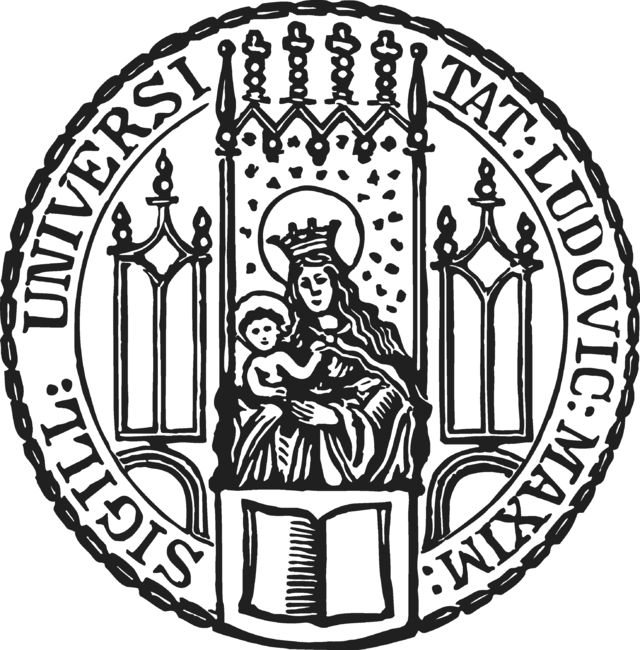} \\
	    \vfill
	    {\bfseries\large{Master's Thesis by} \\}
	    \vspace{5mm}
	    {\bfseries\huge Gytis Kulaitis }
	    \vfill
	    \vfill
	    {\bfseries\large Supervisor: Prof. Martin Schottenloher\\}
	    \vspace{2mm}
	    {\bfseries\large{LMU University of Munich}}\\
	    \vspace{2mm}
	    {\bfseries May 10, 2016}
	    	    
	\end{titlepage}

		\cleardoublepage
	
	\thispagestyle{empty}
	\vspace*{\fill}
	\textit{I certify that this project report has been written by me, is a record of work carried out by me, and is essentially different from work undertaken for any other purpose or assessment.}
	\vspace*{\fill}
	\cleardoublepage

	\frontmatter 

	\chapter*{Acknowledgements}

First of all, I would like to thank my advisor Prof. Schottenloher for the suggestion of the topic and writing such an eye-opening introductory book to the mathematical CFT without which this project would have not been possible, or at least would have been significantly harder.

I thank Dr. Dybalski for agreeing to be the second marker on such a short notice, and for teaching me AQFT and answering all my questions about the course.

My sincere gratitude goes to the Stack Exchange community, in particular, to Physics Stack Exchange, Mathematics Stack Exchange and MathOverflow. I think Stack Exchange is the best thing which happened to the Internet after Wikipedia. The users Marcel Bischoff and Anshuman deserve a separate mention because their answers came when I was very stuck and without their input this work would have been of way much poorer quality.

I would also like to thank Profs. Kac and Nikolov for the clarifications of their work which were very insightful.

For fruitful discussions and great googling skills I am grateful to Abhiram.

For feeling welcome in Germany I am very thankful to the DAAD, which has supported me during my first year, and to the Elite Netzwerk Bayern for funding such a great master's program.

I thank Robert for all the little things which have made my studies in Munich smoother and more fun.

For hospitality during my defense I am indebted to Andreas, Marin and Nina.

Last but not least, I would like to thank my family for its continuous support throughout the years.

			\chapter*{Abstract}


We give introductions into the representation theory of the Virasoro algebra, Wightman axioms and vertex algebras in the first part.

In the second part, we compare the above definitions. We give a proof of L\"uscher and Mack \cite{Luescher1976} that a dilation invariant 2D QFT with an energy-momentum tensor gives rise to two commuting unitary representations of the Virasoro algebra.

We give a proof of Schottenloher \cite[p. 193]{Schottenloher2008} that associated to a Verma module $M(c,0)$ of highest weight zero, there exists a vertex operator algebra of CFT type. This result was firstly proved by Frenkel and Zhu \cite{Frenkel1992}. We then recall another result of \cite{Frenkel1992} that related to $M(c,0)$ there exists a Virasoro vertex operator algebra $L(c,0)$. We follow \cite{Dong2014} and show that $L(c,0)$ is a unitary vertex operator algebra. The converse is a tautology---each conformal vertex algebra has at least one representation of the Virasoro algebra. Moreover, if we have a unitary vertex algebra, then this representation is unitary as well.

Finally we compare Wightman QFTs to vertex algebras. We present Kac's \cite[Sec. 1.2]{Kac1998} proof that every Wightman M\"obius CFT (a 2D Wightman QFT containing quasiprimary fields) gives rise to two commuting strongly-generated positive-energy M\"obius conformal vertex algebras. If the number of generating fields of each conformal weight is finite, then these vertex algebras are unitary quasi-vertex operator algebras. As a corollary using L\"uscher--Mack's Theorem we obtain that a Wightman CFT (a Wightman M\"obius CFT with an energy-momentum tensor) gives rise to a conformal vertex algebra which furthermore becomes a unitary vertex operator algebra, if the number of generating fields of each conformal weight is finite. We reverse Kac's arguments and get a converse proof that two unitary (quasi)-vertex operator algebras can be combined to give a Wightman (M\"obius) CFT.

	\tableofcontents
	\thispagestyle{plain} 

\chapter*{Introduction}


One could argue that modern physics is the study of symmetries. Indeed, Noether's theorem states that symmetries correspond to conservation laws and this observation underlies most of the current physics. One such commonly arising symmetry is the conformal symmetry. Loosely speaking conformal symmetry means that our physical system under consideration is invariant under angle preserving maps. Such a symmetry may seem to be rather restrictive and indeed it is. However, there is an abundance of physical systems that can be treated as conformally invariant at least up to a very good approximation. More precisely, one of the most notable applications of a 2-dimensional conformal field theory (2D CFT), a field theory invariant under conformal transformations in 2 dimensions, is to statistical mechanics and string theory \cite{Francesco1999}, \cite{Blumenhagen2013}. Among the newer developments one could mention AdS/CFT correspondence which was first formulated in high energy physics \cite{Maldacena1999} and now is also applied in condensed matter physics \cite{Pires2014}.

2D CFTs are special among other CFTs because their Lie algebra contains the Virasoro algebra which is infinite dimensional. Thus, 2D CFTs are restricted even more than their higher-dimensional counterparts. This restrictiveness have led to many different mathematical axiomatizations of CFTs. We will present and explore the relationship between two of them: 2D CFT in Wightman framework and vertex algebras. Because of the importance of the Virasoro algebra to 2D CFTs, we also add its representation theory for completeness.

Wightman axioms \cite{Wightman1964} are the first attempt to define QFT rigorously. As such, they try to encompass the whole of QFT. Under some modifications they also describe 2D CFTs. The language of Wightman framework is functional analysis.

Vertex algebras \cite{Borcherds1986,Frenkel1988} on the other hand are algebraic and describe only the chiral half of a 2D CFT. A 2D field is called chiral if it depends only on a single coordinate. Therefore, without a background in physics the fact that there should be a relationship between 2D Wightman CFT and vertex algebras is not obvious and even armed with such knowledge providing a detailed proof still requires some work. The master's thesis aims to fill in these gaps.


As far as we know, the first mathematically rigorous proof that from a 2D Wightman M\"obius CFT one can construct two M\"obius conformal vertex algebras was given by Kac in \cite{Kac1998}. In the same reference Kac also wrote that: ``Under certain assumptions and with certain additional data one may reconstruct the whole QFT from these chiral algebras, but we shall not discuss this problem here". We were unable to find any references containing a proof of this plausible claim. The users of MathOverflow were not aware of any references either \cite{Gytis2015}, although the general idea was rather clear (see Marcel Bischoff's comment in \cite{Gytis2015}). Since we found Kac's proof clear and natural, it was an obvious choice to base the thesis on it and give a converse proof, namely, that two vertex algebras can be combined into a 2D Wightman M\"obius CFT by reversing the arguments of \cite{Kac1998}. Along the way we also managed to extend Kac's proof to conformal vertex algebras using L\"uscher--Mack Theorem. For the converse proof we started with vertex operator algebras because there is a wealth of mathematical literature about them and the recent work \cite{Carpi2015} includes a lot of useful results. However, there should be a more general proof which would also include (M\"obius) conformal vertex algebras which are not (quasi-)vertex operator algebras.

\thispagestyle{plain}
This work is divided into two parts. The first part gives the necessary background, whereas the second part explores the relationships. Experts in the field are encouraged to skip the first part altogether and use it just to refresh their memory for the well-known definitions, if needed.

We have chosen to present the material as follows:
\begin{itemize}
\item In \Cref{chap:conformal group} we show that the conformal group of the Minkowski plane $\reals^{1,1}$ is $\Diff_+(\reals)\times\Diff_+(\reals)$ or $\Diff_+(\OneTorus)\times\Diff_+(\OneTorus)$ and its relation to $\sorthgp^+(2,2)/\{\pm 1\}$ and $\pslinear(2,\reals)$.
\item \Cref{chap:Vir algebra} is concerned with the Virasoro algebra. We define the Virasoro algebra as the unique non-trivial universal central extension of the Witt algebra---a dense subalgebra of the vector fields on a circle. Moreover, we give an introduction to the representation theory of the Virasoro algebra.
\item In \Cref{chap:vertex} we start with the basics and carefully define vertex algebras and related notions of (M\"obius) conformal vertex algebras and (quasi)-vertex operator algebras. No prior knowledge is assumed. We give full proofs of all the fundamentals and start relying on other sources for proofs only in the last section for which readily accessible sources are available, e.g. \cite{Carpi2015}.
\item In \Cref{chap:Wightman} we present the Wightman axioms for a scalar field. We prove the existence of Wightman distributions, which according to Wightman's Reconstruction \cref{thm:Wightman reconstruction} provide an equivalent description of the theory. We also define a Wightman (M\"obius) CFT.
\item \Cref{chap:Wightman and Vir} starts the second part. We prove the L\"uscher--Mack Theorem which shows that 2D dilation invariant Wightman QFT gives rise to two commuting Virasoro algebras.
\item \Cref{chap:Vir and vertex} is rather trivial. We construct a vertex operator algebra from a Verma module of weight zero and note that the converse is a tautology.
\item \Cref{chap:Wightman and vertex} is the highlight of this work. It contains Kac's Theorem that a Wightman (M\"obius) CFT gives rise to two commuting (M\"obius) conformal strongly-generated vertex algebras and a converse that two unitary (quasi)-vertex operator algebras give rise to a Wightman (M\"obius) CFT.
\end{itemize}\thispagestyle{plain}

Throughout the master's thesis we consider bosonic QFTs on the plane because we also wanted to make this work accessible and thus not cluttered with minor details. However, the generalization to superspaces including fermions is quite trivial and can be found in our main references: for Wightman axioms in \cite{Wightman1964,Bogolubov1989} and for vertex algebras in \cite{Kac1998}. We are shy of examples because constructing them for general QFTs is rather hard and there is even a Millenium Prize for constructing a non-trivial QFT in $\reals^4$ \cite{yangmills}. However, we provide full details for the transformations of scalar fields and the energy-momentum tensor from one framework to another in the proofs themselves.

 It should be noted that Wightman axioms and vertex algebras are not the only mathematical definitions of 2D CFTs. Other mathematical definitions include Segal's axioms \cite{Segal1988} and conformal nets, see, e.g., \cite{Carpi2015}. Conformal nets describe chiral CFTs in the framework of algebraic QFT, whereas Segal's axioms describe full 2D CFTs on arbitrary genera, i.e. not only on $\reals^2$ or the open disk, as considered in this work. Thus, Segal's axioms seem to be superior to other approaches. However, many different approaches to the same problem are often beneficial in providing more tools to tackle it and to gain familiarity with the problem in the simpler cases before embarking on the most general form of the problem.
 

\thispagestyle{plain}



		
	\mainmatter 
	\part{Background}

\chapter{Conformal Group}\label{chap:conformal group}\setcounter{page}{1}

\section{General Case}
We start with some basic definitions as given in \cite{Schottenloher2008}. Chapters 1 and 2 of \cite{Schottenloher2008} are our main references for this chapter.

\begin{defn}[(Semi-)Riemannian manifold]
A \textit{semi-Riemannian manifold} is a smooth manifold $M$ equipped with a non-degenerate, smooth, symmetric metric tensor $g$. A \textit{Riemannian manifold} is a semi-Riemannian manifold whose metric tensor is also positive-definite.
\end{defn}

\begin{defn}[Conformal transformation]
Let $(M, g)$ and $(M', g')$ be two semi-Riemannian manifolds of dimension $n$. Let $U\subset M$, $V\subset M'$ be open. A smooth mapping $f : U\to V$ of maximal rank is called a \textit{conformal transformation} or \textit{conformal map} if there exists a smooth function $\Omega : U \to \reals_{>0}$ such that
\[
	f^* g' = \Omega ^2 g,
\]
where $f^* g'_p (X,Y) := g'_{f(p)} (D_p f(X), D_p f(Y))$ is the pullback of $g'$ by $f$ evaluated at a point $p\in U$ and $D_p f : T_p U\to T_p V$ is the derivative of $f$ at the point $p\in U$. The function $\Omega$ is called the \textit{conformal factor} of $f$.
\end{defn}

Some authors also require a conformal transformation to be bijective and/or orientation preserving.\\

Locally in a chart $(U, \phi)$ of $M$ we have
\[
	(f^* g')_{\mu\nu} (p) = g'_{ij} (f(p)) \partial_\mu f^i \partial_\nu f^j\quad \forall p\in U.
\]
Hence $f$ is conformal if and only if
\begin{equation}\label{eq:locally conformal}
	\Omega^2 g_{\mu\nu} = (g'_{ij} \circ f) \partial_\mu f^i \partial_\nu f^j
\end{equation}
in every coordinate patch.

\begin{rmk}
Since we have required a conformal map to be of maximal rank, conformal maps are local diffeomorphisms.
\end{rmk}

Even though the definition of a conformal transformation is straightforward, it turns out that it is not trivial to sensibly define the conformal group. We state the general definition as given in \cite{Schottenloher2008}.

\begin{defn}[Conformal group] \label{def:conformal group}
The \textit{conformal group} $\confgp(\reals^{p,q})$ is the connected component containing the identity in the group of conformal diffeomorphisms of the conformal compactification of $\reals^{p,q}$.
\end{defn}

In \Cref{sec:conformal group of Euclidean plane} we will see that this definition has to be modified for the Euclidean plane. Moreover, the Minkowski plane is also special, since it does not need a conformal compactification to make sense. We will show this in \Cref{sec:conformal group Minkowski}.
Thus, the general definition of the conformal group boils down to cases $\reals^{1,1}$, $\reals^{2,0}$ and $\reals^{p,q}$ with $p+q\ge 3$.

\begin{thm}[Conformal group] The conformal group $\confgp (\reals^{p,q})$ of $\reals^{p,q}$ is:

	1) $(p,\,q)\ne (1,1),\; p,q\ge 1$
		\begin{fleqn}[19pt] 
		\begin{equation*}
 		\confgp(\reals^{p,q}) = \begin{cases}
 		\sorthgp^+(p+1,q+1) &\mbox{if } -\id\notin \sorthgp^+(p+1,q+1) \\ 
 		\sorthgp^+(p+1,q+1)/ \{\pm \id\} &\mbox{if  } -\id\in  \sorthgp^+(p+1,q+1);
	\end{cases}
		\end{equation*}
		\end{fleqn}
		
	2) $(p,\,q)=(1,1)$
		\begin{fleqn}[19pt]
		\begin{equation*}
		\confgp(\reals^{1,1})= \Diff_+(\OneTorus)\times \Diff_+(\OneTorus).
		\end{equation*}
		\end{fleqn}
\end{thm}

By the above, the groups $\sorthgp^+(p,q)$ are the most important for CFT. We state here their generators before specializing to the 2-dimensional case. For a proof check \cite[Thms. 2.9 and 2.11]{Schottenloher2008}.
\begin{thm}The group $\sorthgp^+(p+1,q+1)$,  with $p,q\ge1$, $p+q=n$, is isomorphic to the group generated by
\begin{itemize}
	\item translations 
	\[
		x\mapsto x + c,
	\]
	\item special orthogonal transformations 
	\[
	x\mapsto \Lambda x,
	\]
	\item dilations 
	\[
	x\mapsto e^{\lambda} x,
	\]
	\item special conformal transformations
	\[
	x\mapsto \frac{x+|x|^2 b}{1+2 \langle x, b \rangle + |x|^2 |b|^2}. 
	\]
	\end{itemize}
Here $x,b,c\in\reals^n,\, \Lambda\in \sorthgp^+(p,q),\,  \lambda\in\reals.$
\end{thm}

\section{Conformal group of $\reals^{1,1}$}\label{sec:conformal group Minkowski}

In this section we will prove that $\confgp(\reals^{1,1})\cong \Diff_+(\OneTorus)\times\Diff_+(\OneTorus)$.
\begin{prop}\label{prop: characterization of conformality in Minkowski}
A smooth map $f = (u, v) : U \to \reals^{1,1}$ is conformal if and only if 
\begin{equation}\label{eq:when Minkowski maps conformal}
u_x^2 > v_x^2\quad \text{and} \quad u_x =v_y,\; u_y =v_x\quad \text{or}\quad u_x = -v_y,\; u_y = -v_x.
\end{equation}
Here $U\subset\reals^{1,1}$ is connected and open.
\end{prop}
\begin{proof}
The condition of being conformal $f^*g=\Omega^2 g$ for $g=g^{1,1}$ with $\Omega^2>0$ is equivalent to the equations
\begin{equation}\label{eq:conformal locally in Minkowski}
u_x^2-v_x^2=\Omega^2,\quad u_x u_y-v_x v_y=0,\quad u_y^2-v_y^2=-\Omega^2,\quad \Omega^2>0.
\end{equation}

First assume that the map $f$ is conformal. Then the equations \eqref{eq:conformal locally in Minkowski} imply that $u_x^2=\Omega^2+v_x^2>v_x^2$ and adding the first three of them we get
\[
0=u_x^2-v_x^2+u_y^2-v_y^2+2u_x u_y-2v_x v_y= (u_x+u_y)^2-(v_x+v_y)^2.
\]
Hence, 
\begin{equation}\label{eq:random equation}
u_x+u_y= \pm (v_x+v_y).
\end{equation}
Taking the positive root and using the second equation of \eqref{eq:conformal locally in Minkowski} we get
\begin{align*}
0&=-u_x u_y + v_x v_y = u_x^2 - u_x^2 -u_x u_y + v_x v_y = u_x^2 - u_x(u_x+u_y)+v_x v_y\\
&= u_x^2 - u_x(v_x+v_y)+v_x v_y=(u_x-v_x)(u_x-v_y),
\end{align*}
i.e. $u_x=v_x$ or $u_x=v_y$. The solution $u_x=v_x$ contradicts $u_x^2- v_x^2 = \Omega^2>0$. Thus, we have $u_x=v_y$ and by \cref{eq:random equation} $u_y=v_x$ as required. Similarly taking the negative root in \cref{eq:random equation} yields $u_x = -v_y$ and $u_y=-v_x$.

Now assume that the equations \eqref{eq:when Minkowski maps conformal} are fulfilled. Setting $\Omega^2:=u_x^2-v_x^2>0$ and substituting $u_y=\pm v_x$, $v_y=\pm u_x$ yields
\[u_y^2-v_y^2=v_x^2-u_x^2=-\Omega^2 \quad \text{and} \quad u_x u_y- v_x v_y =0,\]
i.e. $f$ is conformal. If $u_x=v_y$, $u_y=v_x$, then 
\[\det Df=u_x v_y-u_y v_x = u_x^2 -v_x^2>0.\]
So $f$ is orientation preserving in this case. Similarly, if $u_x=-v_y$ and $u_y=-v_x$, then $f$ is orientation reversing.
\end{proof}

The next lemma shows that in the case of $U=\reals^{1,1}$, the global orientation-preserving conformal transformations can be conveniently described using light-cone coordinates $x^{\pm}=x\pm y$.
\begin{lem}\label{lem:conformal transformations in Minkowski}
Given $f\in C^{\infty}(\reals)$, define $f_{\pm}\in C^{\infty}(\reals^2,\reals)$ by $f_{\pm}:=f(x\pm y)$. The map
\begin{align*}
\phi: C^{\infty} (\reals)\times C^{\infty}(\reals) &\to C^{\infty}(\reals^2 ,\reals^2)\\
(f\, ,\, g)\quad\quad\;\, &\mapsto \frac{1}{2}(f_+ + g_-\, ,\, f_+ - g_-)
\end{align*}
has the following properties:
\begin{enumerate} 
\item $\im \phi = \{(u,v)\in C^{\infty}(\reals^2,\reals^2)\; |\; u_x=v_y, \; u_y=v_x \}$,\label{item:phi1}
\item $\phi(f,g)$ is conformal $\iff f'>0,\; g'>0$ or $f'<0,\; g'<0$,\label{item:phi2}
\item $\phi(f,g)$ is bijective $\iff$ $f$ and $g$ are bijective,\label{item:phi3}
\item $\phi(f\circ h, g\circ k) = \phi(f,g)\circ \phi(h,k)$ for $f,g,h,k\in C^{\infty}(\reals)$.\label{item:phi4}
\end{enumerate}
\end{lem}

\begin{proof}
\ref{item:phi1} Let $(u,v)\in\im \phi$, i.e. $(u,v)=\phi(f,g)$ for some $f,g\in C^{\infty}(\reals)$. From
\begin{alignat*}{2}
u_x&=\frac{1}{2}(f'_+ + g'_-),\quad && u_y=\frac{1}{2}(f'_+ - g'_-),\\
v_x&=\frac{1}{2}(f'_+ - g'_-), \quad && v_y=\frac{1}{2}(f'_+ + g'_-),
\end{alignat*}
it follows that $u_x=v_y$ and $u_y=v_x$.

Conversely, let $(u,v)\in C^{\infty}(\reals^2,\reals^2)$ be such that $u_x=v_y$ and $u_y=v_x$. Then $u_{xx}=v_{yx}=v_{xy}=u_{yy}$. But this is just the one dimensional wave equation and it has solutions $u(x,y)=\frac{1}{2}(f_+(x,y)+g_-(x,y))$ with suitable $f,g\in C^{\infty}(\reals)$. Because of $v_x=u_y=\frac{1}{2}(f'_+ - g'_-)$ and $v_y=u_x=\frac{1}{2}(f'_+ +g'_-)$, we have $v=\frac{1}{2}(f_+-g_-)$ where $f$ and $g$ might have to be translated by a constant.

\ref{item:phi2} If $(u,v)=\phi(f,g)$, then $u_x^2-v_x^2=f'_+ g'_-$. Thus,
\[u_x^2-v_x^2>0 \iff f'_+ g'_- >0\iff f'g'>0.\]

\ref{item:phi3} Let $\vphi=\phi(f,g)$. To prove the equivalence of injectivities note that
\begin{gather*}
\vphi(x,y)=\vphi(x',y')\iff 
\begin{cases} 
f(x+y)+g(x-y)=f(x'+y')+g(x'-y') \\ 
f(x+y)-g(x-y)=f(x'+y')-g(x'-y')
\end{cases}
\\
\iff\begin{cases} 
f(x+y)=f(x'+y') \\ 
g(x-y)=g(x'-y')
\end{cases}
\impliedby
\begin{cases} 
x+y=x'+y' \\ 
x-y=x'-y'
\end{cases}
\iff
\begin{cases} 
x=x' \\ 
y=y'.
\end{cases}
\end{gather*}
Now if both $f$ and $g$ are injective, we also get the forward implication in step 3 and the equivalence diagram above shows that $\phi$ is also injective. On the other hand, if $\phi$ is injective, i.e. $\vphi(x,y)=\vphi(x',y')\iff x=x'\; \text{and}\; y=y'$, then the diagram above shows that $f$ and $g$ are injective as well.

Let $(x',y')\in\reals^2$ be arbitrary. If $f$ and $g$ are surjective, then $\exists\, s,t\in\reals$ with $f(s)=x'+y'$, $g(t)=x'-y'$. Moreover, $\vphi(x,y)=(x',y')$ with $x:=\frac{1}{2}(s+t)$, $ y:=\frac{1}{2} (s-t)$.

Conversely, fix $x'\in\reals$ and assume that $\phi$ is surjective. Then $ \exists (x,y)\in\reals^2 $ such that $ \vphi(x,y) = (x',0) $. This implies that $ f(x+y) = x'=g(x-y)$ and hence $f$ and $g$ are surjective.

\ref{item:phi4} Set $ \phi:=\phi(f,g) $ and $ \psi:=\phi(h,k) $. We have 
\[   \phi\circ\psi=\frac{1}{2} \left(f_+\circ\psi+g_-\circ\psi\,,\, f_+\circ\psi-g_-\circ\psi\right), \]
where $ f_+\circ\psi=f\left(1/2(h_+ +k_-)+1/2(h_+ -k_-)\right)=f\circ h_+ = (f\circ h)_+$ and other terms evaluate similarly. Thus,
\[ \phi\circ\psi=\frac{1}{2} \left((f\circ h)_+ + (g\circ k)_-\, , \, (f\circ h)_+ - (g\circ k)_-\right) = \phi(f\circ h\, , \, g\circ k) \]
as required.
\end{proof}

\begin{prop}\label{prop:all conformal transformations in Minkowski}
The group of orientation-preserving conformal diffeomorphisms
\[\vphi:\reals^{1,1}\to\reals^{1,1}\]
is isomorphic to the group
\[(\Diff_+(\reals)\times \Diff_+(\reals))\cup (\Diff_-(\reals)\times\Diff_-(\reals)).\]
\end{prop}

\begin{proof}
The result follows from \cref{lem:conformal transformations in Minkowski} and \cref{prop: characterization of conformality in Minkowski}.
\end{proof}

From \cref{prop:all conformal transformations in Minkowski} we see that the conformal compactifications mentioned in the general \cref{def:conformal group} are not necessary in Minkowski plane $ \reals^{1,1} $. Hence, it would make sense simply to define the conformal group $ \confgp(\reals^{1,1}) $ as the identity component of the group of conformal transformations $ \reals^{1,1}\to\reals^{1,1} $ which is isomorphic to $ \Diff_+(\reals) \times\Diff_+(\reals) $ by \cref{lem:conformal transformations in Minkowski}. However, usually researchers want to work with a group of transformations on a compact manifold. Thus, $ \reals$ is replaced by the circle $S$:
\[ \reals^{1,1}\to S^{1,1}=\OneTorus\times\OneTorus\subset\reals^{2,0}\times\reals^{0,2}\cong\reals^{2,2}. \]
From such reasoning it follows that a sensible definition for the conformal group $ \confgp(\reals^{1,1}) $ is the identity component of the group of all conformal diffeomorphisms $ S^{1,1}\to S^{1,1} $. Analogously to Theorem \ref{prop:all conformal transformations in Minkowski} the group of orientation-preserving conformal diffeomorphisms $ S^{1,1} $ turns out to be isomorphic to
\[ (\Diff_+(S)\times \Diff_+(S))\cup (\Diff_-(S)\times\Diff_-(S)).\] 
To prove this result, one simply has to consider $2\pi$-periodic functions in the proof of Lemma \ref{lem:conformal transformations in Minkowski}. Therefore, we have:

\begin{thm}\label{thm:conformal group of Minkowski}
$ \confgp(\reals^{1,1}) \cong \Diff_+(\OneTorus)\times\Diff_+(\OneTorus) $.
\end{thm}

By \cref{thm:conformal group of Minkowski} we see that the conformal group of Minkowski plane $ \confgp(\reals^{1,1}) $ is infinite dimensional. However, there is also a finite dimensional counterpart $ \sorthgp^+(2,2)/\{\pm 1\}\subset \confgp(\reals^{1,1}) $.

\begin{defn}
The \textit{restricted conformal group} of the (compactified) Minkowski plane $ \reals^{1,1} $ is $ \sorthgp^+(2,2)/\{\pm 1\} $.
\end{defn}

The group $\sorthgp^+(2,2)/\{\pm 1\}$ consists of translations, Lorentz transformations, dilations and special conformal transformations \cite[Thm. 2.9]{Schottenloher2008}. If we introduce the light-cone coordinates
\[
	x_+ = x+y,\quad x_-= x-y,
\]
then the restricted conformal group acts as
\[
	(A_+,\, A_-) (x_+, x_-) =  \left(\frac{a_+ x_+ +b_+}{c_+ x_+ + d_+},\; \frac{a_- x_- +b_-}{c_- x_- + d_-}\right).		
\]
Thus,
\[
	\sorthgp^+(2,2)/ \{\pm 1\} \cong \pslinear(2,\reals) \times \pslinear(2,\reals).
\]
Because of this decoupling translations and special conformal transformations can be chosen as the generators of the group as the following proposition shows. It is a special case of Dickson's Theorem.

\begin{prop}\label{prop:generators of PSL2} 
	The group $\pslinear_2(\field)$ with $\field=\complex$ or $\field=\reals$ is generated by translations and special conformal transformations.
\end{prop}

\begin{proof}
	We use the action of $\pslinear_2(F)$ on one of the light-cone coordinates to identify translations with upper triangular matrices and special conformal transformations with lower triangular matrices. So we need to prove that
	\[
		S = \left\langle 
			\begin{pmatrix}
			1 & a\\
			0 & 1
			\end{pmatrix},
			\begin{pmatrix}
				1 & 0\\
				b & 1
			\end{pmatrix}
			\right\rangle
	\]
	is the whole of $\pslinear_2(F)$, i.e. that $S$ contains dilations and Lorentz transformations.
	
	First of all we note that a reflection (rotation by $\pi$) is equal to
	\[
		\begin{pmatrix*}[r]
			0 & 1\\
			-1& 0
		\end{pmatrix*} =
		\begin{pmatrix}
				1 & 1\\
			0 & 1
		\end{pmatrix}
		\begin{pmatrix*}[r]
			1 & 0\\
			-1 & 1
		\end{pmatrix*}
		\begin{pmatrix}
			1 & 1\\
			0 & 1
		\end{pmatrix}
	\]
	and hence is in $S$. Furthermore, we have
	\[
		\begin{pmatrix*}[r]
			0 & 1\\
			-1& 0
		\end{pmatrix*}
		\begin{pmatrix*}[c]
			1 & -e^{-x}\\
			0 & 1
		\end{pmatrix*}
		\begin{pmatrix*}
			1 & 0\\
			e^x & 1
		\end{pmatrix*}
		\begin{pmatrix}
			1 & -e^{-x}\\
			0 & 1
		\end{pmatrix}=
		\begin{pmatrix}
			e^x & 0\\
			0	& e^{-x}
		\end{pmatrix}
	\]
	and hence dilations are in $S$. Moreover, every Lorentz transformation can be diagonalized
	\[
		\begin{pmatrix}
		\cosh(\xi) & \sinh(\xi)\\
		\sinh(\xi)	& \cosh(\xi)
		\end{pmatrix} =
		\begin{pmatrix*}[r]
			1/2  & -1\\
			1/2  & 1
		\end{pmatrix*}
		\begin{pmatrix}
			e^{\xi} & 0\\
			0		& e^{-\xi}
		\end{pmatrix}
		\begin{pmatrix*}[c]
			1 	 & 1\\
			-1/2 & 1/2
		\end{pmatrix*}
	\]
	and
	\[
		\begin{pmatrix*}[r]
			1/2 & -1\\
			1/2 &  1
		\end{pmatrix*}=
		\begin{pmatrix*}[r]
			1 & -1 \\
			0 &  1
		\end{pmatrix*}
		\begin{pmatrix*}
			1   & 0 \\
		   1/2  & 1 
		\end{pmatrix*},\;
		\begin{pmatrix*}[c]
			1 & 1\\
			-1/2 & 1/2 
		\end{pmatrix*} = 
		\begin{pmatrix*}[c]
			1    & 0\\
			-1/2 & 1
		\end{pmatrix*}
		\begin{pmatrix}
			1 & 1\\
			0 & 1
		\end{pmatrix}.
	\]
\end{proof}


\section{Conformal group of $\reals^{2,0}$}\label{sec:conformal group of Euclidean plane}

The next lemma shows why the general definition of conformal group \ref{def:conformal group} fails for the Euclidean plane.

\begin{lem}\label{lem:conformal transformations in the plane}
Conformal transformations $f:U\to \complex$ are the locally invertible holomorphic or antiholomorphic functions with conformal factor $|\det Df\,|$. Here $U$ is an open and connected subset of $\complex$.
\end{lem}

\begin{proof}
A smooth map $f: U \to \complex$ on a connected open subset $U\subset \complex$ is conformal according to Equation \ref{eq:locally conformal} with conformal factor $\Omega: U \to \reals_{>0}$ if and only if for $u= \re f$ and $v= \im f$ we have
\begin{equation}\label{eq:complex conformal locally}
u_x ^2 + v_x ^2= \Omega^2 = u_y^2+v_y^2 >0 \quad\text{and}\quad u_x u_y + v_x v_y =0.
\end{equation}

These equations are satisfied by holomorphic and antiholomorphic functions with the property $u_x ^2 + v_x^2 >0$, since by Cauchy-Riemann equations $u_x=v_y$, $u_y = -v_x$ for holomorphic functions and $u_x = -v_y$, $u_y = v_x$ for antiholomorphic. For holomorphic (antiholomorphic) functions $u_x^2+v_x^2 >0$ is equivalent to $\det Df >0$ ($\det Df <0$).

Conversely, given a conformal transformation $f= (u,v)$ the equations \eqref{eq:complex conformal locally} imply that $(u_x, v_x)$ and $(u_y, v_y)$ are perpendicular vectors in $\reals^{2,0}$ of equal length $\Omega >0$. Hence, $(u_x,v_x)=(-v_y,u_y)$ or $(u_x,v_x)=(v_y,-u_y)$. These correspond to $f$ being holomorphic or antiholomorphic with $\det Df>0$ or $\det Df<0$, respectively.
\end{proof}

\begin{cor}
The holomorphic maps $f : U\to \complex$ with $f' \ne 0$ are in one-to-one correspondence with conformal orientation preserving maps $h: U\to \complex$. Here $U\subset \complex$ is open and connected.
\end{cor}
 Hence, the conformal compactification \cite[Rmk. 2.2 and Def. 2.7]{Schottenloher2008} does not exist: there are many noninjective conformal transformations. For example,
\[
	\complex\backslash \{0\}\to\complex,\quad z\mapsto z^k,\quad \text{with}\; k\in\integ\backslash \{-1,0,1\}.
\]
Therefore, one is lead to a different definition for the Euclidean plane.

\begin{defn}\label{def:global conformal transformation} 
A \textit{global} conformal transformation of $\reals^{2,0}$ is an injective holomorphic function, defined on the whole of $\complex$ with at most one exceptional point.
\end{defn}

It follows that the group $\confgp(\reals^{2,0})$ is isomorphic to the \textbf{M\"obius group} which is the group of all holomorphic maps $f:\complex\to\complex$ such that
\[
	f(z)=\frac{az+b}{cz+d},\quad cz+d\ne 0\quad\quad\text{and}\quad\quad \begin{pmatrix}
		a & b \\
		c & d
	\end{pmatrix}\in\slinear(2,\complex).
\]
Even though the matrices are in $\slinear(2,\complex)$, the transformations are invariant under multiplication by $-1$. Hence,
\[
\confgp(\reals^{2,0})\cong \slinear(2,\complex)/\{\pm 1\}=\pslinear(2,\complex).
\]
Moreover, there exist other well-known isomorphisms of $\pslinear(2,\complex)$ and so we have
\[
\confgp(\reals^{2,0})\cong\pslinear(2,\complex)\cong \sorthgp^+(3,1)\cong \Aut( \widehat{\mathbf{C}}),
\]
where $\widehat{\mathbf{C}}$ is the Riemann sphere.


\chapter{Virasoro Algebra}\label{chap:Vir algebra}

We present a short introduction to the Virasoro algebra which arises as a complexification of a restriction of the Lie algebra of $\confgp(\reals^{1,1})$.

The rest of this chapter is arranged as follows:
\begin{itemize}
\item In \Cref{sec: extensions} we present the general theory of central extensions of Lie algebras.
\item The Witt algebra is defined and it is shown that the Virasoro algebra is the unique nontrivial universal central extension of it in \Cref{sec: witt}.
\item In \Cref{sec: representations} main tools of the representation theory of the Virasoro algebra are provided. The highlight of the section is the proof that the Virasoro algebra admits unitary representations for all $c>1,\, h>0$.
\end{itemize}

\section{Central Extensions of Lie Algebras}\label{sec: extensions}

Our main references for this section are \cite{Iohara2011} and \cite{Schottenloher2008}.

\vspace{0.2cm}
\begin{flushleft}
Throughout this section let $\field$ be a field of characteristic zero (usually $\reals$ or $\complex$).
\end{flushleft}

\begin{defn}[Lie algebra]
A \textit{Lie algebra} is a vector space $\gLie$ over some field $F$ together with a binary operation $[\cdot,\cdot]:\gLie \times \gLie\to \gLie$ called the \textit{Lie bracket}, which satisfies the following axioms $\forall a,b\in F$ and $\forall X,Y,Z\in\gLie$:  
\begin{itemize}
	\item bilinearity\\
		\hspace*{0.5cm}$[aX+bY,Z] = a[X,Z]+b[Y,Z], \quad [Z, aX+bY] = a[Z,X] + b[Z,Y]$,
	\item alternating property\\
		\hspace*{0.5cm}$[X,X]=0$,
	\item the Jacobi identity \\
		\hspace*{0.5cm}$[X,[Y,Z]]+[Y,[Z,X]]+[Z,[X,Y]]=0$.
\end{itemize}
\end{defn}

\begin{defn}[Abelian Lie algebra]
A Lie algebra $\aLie$ is called $abelian$ if ${[X,Y]=0}\;\, \forall X,Y \in \aLie$.
\end{defn}


\begin{defn}[Central extension of Lie algebra]
Let $\aLie$ be an abelian Lie algebra over a field $\field$ and $\gLie$ a Lie algebra over $\field$. An \textit{exact sequence} of Lie algebra homomorphisms 
\[
0\longrightarrow\aLie\xlongrightarrow{\iota}\hLie\xlongrightarrow{\pi}\gLie\longrightarrow 0
\]
is called a \textit{central extension} of $\gLie$ by $\aLie$, if $[\iota({\aLie}),\hLie]=0$. Then $\aLie$ is called the \textit{kernel} of the central extension.

A sequence of maps is called $exact$ if the kernel of each map is equal to the image of the previous map. In particular, here we have that $\iota$ is injective, $\pi$ is surjective and $\ker\pi=\ima\iota\cong\aLie$. Moreover, $[\iota({\aLie}),\hLie]=0$ implies that $\aLie$ corresponds to an ideal in $\hLie$ and hence $\gLie\cong\mathfrak{h/a}$ via $\pi$. 
\end{defn}

\begin{defn}[Universal central extension of Lie algebra]\label{universal central}
A central extension
\[
0\longrightarrow\aLie\xlongrightarrow{\iota}\hLie\xlongrightarrow{\pi}\gLie\longrightarrow 0
\]
of $\gLie$ is called a \textit{universal central extension} if 
	\begin{itemize}
		\item $\hLie = [\hLie, \hLie]$ i.e. $\hLie$ is \textit{perfect},
		\item for all central extensions $\pi' : \hLie' \to \gLie$ there exists a Lie algebra homomorphism $\gamma : \hLie \to \hLie'$ such that the following diagram commutes:
		\[
		\begin{tikzcd}
			\hLie \arrow{d}{\gamma} \arrow{r}{\pi}
			& \gLie \arrow{d}{id} \\
			\hLie' \arrow{r}{\pi'}
			& \gLie
		\end{tikzcd}
		\]
	\end{itemize}
\end{defn}

\begin{lem}
The Lie algebra homomorphism $\gamma$ from \cref{universal central} is unique.
\end{lem}

\begin{proof}
		Let $\gamma' \colon \hLie \to \hLie'$ be another homomorphism such that $\pi = \pi' \circ \gamma'$. Then 
		$\forall X,Y\in \hLie$ we have
		\begin{align*}
		    (\gamma - \gamma')([X,Y]) &= [\gamma(X),\gamma(Y)] - [\gamma'(X),\gamma'(Y)]  \\
		        &= [\gamma(X) - \gamma'(X),\gamma(Y)] + [\gamma'(X),\gamma(Y) - \gamma'(Y)]
		\end{align*}
		Now $\pi' \circ (\gamma(Z)-\gamma'(Z)) = \pi(Z) - \pi(Z) = 0$ 
		i.e. $\gamma(Z) - \gamma'(Z) \in \ker\pi' = \iota'(\aLie')$  $\forall Z\in \hLie$. Since the extension 
		 $\pi' \colon \hLie' \to \gLie$ is central, $[\iota'(\aLie'), \hLie'] = 0$.  Hence $ (\gamma - \gamma')([X,Y]) =  [\gamma(X) - \gamma'(X),\gamma(Y)] + [\gamma'(X),\gamma(Y) - \gamma'(Y)] = 0$\quad $\forall X,Y\in \hLie$. Since 				$\hLie$ is perfect, we get that $\gamma = \gamma'$.	
	\end{proof}
	
\begin{cor}
	A universal central extension, if it exists, is unique up to Lie algebra isomorphism.
\end{cor}

\begin{defn}[Second cohomology group]
%
By $H^2(\gLie, \aLie) := Z^2(\gLie, \aLie)/B^2(\gLie, \aLie)$ the \textit{second cohomology group} is defined where $\aLie$ is regarded as a trivial $\gLie$-module. Here $Z^2(\gLie, \aLie)$ (respectively $B^2(\gLie, \aLie)$) is called the \textit{space of 2-cocycles} (respectively \textit{2-coboundaries}) of $\gLie$ with coefficients in $\aLie$:
	\begin{flalign*}
		Z^2(\gLie, \aLie) &:=\left\{\Theta\colon \gLie \times \gLie \to \aLie \middle |
		\begin{aligned}
				 &\forall X,Y,Z \in \gLie : \\
				 &\text{1. } \Theta \text{ is bilinear, } \\
				 &\text{2. } \Theta(X,Y) = -\Theta(Y,X), \\
				 &\text{3. } \Theta(X, [Y,Z]) + \Theta(Y, [Z,X]) + \Theta(Z, [X,Y]) = 0 \\
		\end{aligned}
		\right\}\\
		B^2(\gLie, \aLie) &:=\left\{\Theta\colon \gLie \times \gLie \to \aLie \mid \exists \mu\colon\gLie\to\aLie \text{ linear,  such that } \Theta(X,Y) = \mu([X,Y]) \right\}.
	\end{flalign*}
\end{defn}

\begin{defn}[Equivalent central extensions]
Two central extensions of Lie algebra $\gLie$ by $\aLie$ are \textit{equivalent} if there exists a Lie algebra isomorphism $\psi \colon \hLie' \to \hLie$ such that the diagram

\[
\begin{tikzcd}
	{} 
	0 \arrow{rr} &&
	\aLie \arrow{rr}{\iota'} \arrow{d}{\id}  &&
	\hLie' \arrow{rr}{\pi'}   \arrow{d}{\psi} &&
	\gLie \arrow{rr}           \arrow{d}{\id}&&
	0
	\\
	0 \arrow{rr} &&
	\aLie \arrow{rr}{\iota} &&
	\hLie \arrow{rr}{\pi}  &&
	\gLie \arrow{rr} &&
	0
\end{tikzcd}
\]
commutes.  
\end{defn}

\begin{lem}\label{extensionsToCocycles}
There is a correspondence between 2-cocycles of $\gLie$ with values in $\aLie$ and central extensions of $\gLie$ by $\aLie$.
\end{lem}

\begin{proof}
Given $\Theta\in Z^2(\gLie, \aLie)$, define $\hLie := \gLie\oplus \aLie$. Then define a bracket
\[
[(X,V), (Y,W)]_\hLie :=([X,Y]_\gLie , \Theta(X,Y)) \quad \forall X,Y\in\gLie, \forall\, V,W\in\aLie.
\]
It follows that this is a Lie bracket by definition of $\Theta$. Thus $(\hLie, [\cdot,\cdot]_\hLie)$ is a Lie algebra. Therefore the exact sequence
\[
0\longrightarrow\aLie\xlongrightarrow{\iota}\hLie\xlongrightarrow{\pr_1}\gLie\longrightarrow 0
\]
where $\iota$ is the inclusion and $\pr_1$ is the projection onto the first variable, is a central extension of $\gLie$.

Conversely, given a central extension
\[
0\longrightarrow\aLie\xlongrightarrow{\iota}\hLie\xlongrightarrow{\pi}\gLie\longrightarrow 0
\]
there is a linear map $\beta \colon \gLie \to \hLie$ with $\pi \circ \beta = \id_\gLie$ (it is not a Lie algebra homomorphism in general). Let
\begin{equation}\label{eq:extToBracket11}
\Theta_\beta (X,Y) :=[\beta(X),\beta(Y)] - \beta([X,Y]) \quad \forall X,Y\in\gLie.
\end{equation}
Since $\pi$ is a Lie algebra homomorphism, 
\[
\pi\circ\Theta_\beta (X,Y) = \pi([\beta(X),\beta(Y)])-[X,Y] =0 \quad \forall X,Y\in \gLie
\]
 i.e. $\ima (\Theta_\beta)\subset \ker\pi = \ima(\iota)\cong \aLie$. So we can interpret $\Theta_\beta$ as $\Theta_\beta \colon\gLie \times \gLie \to \aLie$. Clearly $\Theta_{\beta}$ is bilinear and alternating. The Jacobi identity is proved by noticing that by linearity of $\beta$ and the Jacobi identity on $\hLie$ we have 
 \[
 \beta ([X,[Y,Z]])+\beta ([Y,[Z,X]]+\beta ([Z,[X,Y]]) = 0,
 \]
so that
\begin{align*}
	\Theta_\beta (&X, [Y,Z]) + \Theta_\beta (Y, [Z,X]) + \Theta_\beta (Z, [X,Y])= \\
			&= \bm{[}\beta (X), \beta([Y,Z])\bm{]}  +  \bm{[}\beta (Y), \beta([Z,X])\bm{]} + \bm{[}\beta (Z), \beta([X,Y])\bm{]}=   \\
			&=\bm{[}\beta (X),\bm{(} [\beta(Y),\beta(Z)]- \Theta_\beta (Y,Z)\bm{)]} +\bm{[}\beta (Y),\bm{(} [\beta(Z),\beta(X)]- \Theta_\beta (Z,X)\bm{)]}+ \\
			&+\bm{[}\beta (Z),\bm{(}[\beta(X),\beta(Y)]- \Theta_\beta (X,Y)\bm{)]} =0.
\end{align*}
Here  we have used again the Jacobi identity on $\hLie$ and the fact that $\ima (\Theta_\beta) \in \aLie$ with $[\aLie, \hLie]=0$.
Thus, $\Theta_\beta \in Z^2(\gLie, \aLie)$.
Moreover, $\hLie \cong \gLie \oplus \aLie$ as vector spaces via the linear isomorphism
\[
	\psi \colon\gLie \times \aLie \to \hLie, \quad (X,W) = X\oplus W \mapsto \beta(X) + W.
\]
If we define the Lie bracket on $\gLie \oplus \aLie$ by
\begin{equation}\label{eq:extToBracket12}
	[X\oplus W, Y\oplus V]_{\gLie \oplus \aLie} :=\beta( [X, Y]_\gLie) + \Theta_{\beta} (X,Y) \quad \forall X,Y\in \gLie, \, \forall\, W,V\in\aLie,
\end{equation}
then the map $\psi$ becomes a Lie algebra isomorphism.
\end{proof}

\begin{defn}[Split exact sequence, trivial central extension]
	An exact sequence of Lie algebra homomorphisms
	\[
	0\longrightarrow\aLie\xlongrightarrow{\iota}\hLie\xlongrightarrow{\pi}\gLie\longrightarrow 0
	\]
		\textit{splits} if there is a Lie algebra homomorphim $\beta \colon \gLie \to \hLie$ with $\pi\circ\beta = \id_\gLie$. The homomorphism $\beta$ is called a \textit{splitting map}. A central extension which splits is called a \textit{trivial extension}, since from the proof of \cref{extensionsToCocycles} it is equivalent to the exact sequence
	\[
	0\longrightarrow\aLie\xlongrightarrow{\iota}\gLie\oplus\aLie\xlongrightarrow{\pi}\gLie\longrightarrow 0
	\] 
	where $\gLie\oplus\aLie$ has the Lie bracket $[X\oplus W, Y\oplus V]_{\gLie \oplus \aLie} =\beta( [X, Y]_\gLie)$.
\end{defn}

The following proposition ``mods out'' the trivial cases.
\begin{prop}
There exists a bijection between $H^2(\gLie, \aLie)$ and the set of equivalence classes of central extensions of $\gLie$ by $\aLie$. 
\end{prop}

\begin{proof}
	Using \cref{extensionsToCocycles} all that is left to prove is that two elements $\Theta,\Omega\in Z^2(\gLie, \aLie)$ are such that $\Theta-\Omega \in B^2(\gLie, \aLie)$ if and only if the central extensions defined by $\Theta$ and $\Omega$ are equivalent. Equivalently, we must show that every trivial extension is an extension defined by a coboundary and vice versa.
	
	So let $\Theta \in B^2(\gLie, \aLie)$, i.e. $\Theta(X,Y) = \mu([X,Y])$ for some $\mu\in \text{Hom}_\field (\gLie, \aLie)$. Define a linear map $\beta \colon\gLie\to\hLie(\cong\gLie\oplus\aLie)$ by $\beta(X) := X + \mu(X),\,\, \forall X\in\gLie$. Then 
	\begin{align*}
	\beta([X,Y]) &= [X,Y]_\gLie + \mu([X,Y]) = [X,Y]_\gLie + \Theta(X,Y)\\
	&= [X+\mu(X), Y+\mu(Y)]_\hLie = [\beta(X), \beta(Y)]_\hLie ,
	\end{align*}
	i.e. $\beta$ is a Lie algebra homomorphism.
	Hence, $\beta$ is a splitting map.
	
	Conversely, given a splitting map $\beta\colon\gLie\to\hLie (\cong \gLie\oplus\aLie)$, then $\beta$ has to be of the form $\beta(X) = X + \mu(X),\,\, \forall X\in\gLie$, for some suitable $\mu \in \text{Hom}_\field(\gLie,\aLie)$ since  $\pi\circ\beta = \id_\gLie$. By definition of the bracket on $\hLie$, $[\beta(X),\beta(Y)] = [X,Y] + \Theta(X,Y)$ for all $X,Y\in\gLie$. Moreover, since $\beta$ is a Lie algebra homomorphism we have $[\beta(X), \beta(Y)] = \beta([X,Y])=[X,Y] + \mu([X,Y])$. Hence $\Theta(X,Y) = \mu([X,Y])$.
\end{proof}

\begin{prop}\label{Gar}
	 A Lie algebra $\gLie$ admits a universal central extension if and only if $\gLie$ is perfect.
\end{prop}

\begin{proof}
	First suppose that $\pi \colon \hLie \to \gLie$ is the universal central extension. By definition, $\hLie$ is perfect. Hence,
	\[
		\gLie = \pi(\hLie) = \pi( [\hLie, \hLie] ) = [\pi(\hLie), \pi(\hLie) ] = [\gLie, \gLie].
	\]
	
Now assume that $\gLie$ is perfect. We set
	\begin{alignat*}{2}
	&W'&&:={\bigwedge }^2 \gLie = (\gLie\otimes\gLie)/ \langle X\otimes Y + Y\otimes X\mid X,Y\in \gLie \rangle_\field \, ,\\
	&B_2(\gLie, \field) &&:= \left\{ X\wedge [Y,Z]+ Y\wedge [Z,X] + Z\wedge [X,Y] \mid X,Y,Z \in \gLie \right\},
	\end{alignat*}
	and $W:=W'/B_2(\gLie, \field)$. Let $\Omega \colon W'\to W$ be the canonical projection. By definition, $\Omega\in Z^2(\gLie, W)$. Let
\[
	0\longrightarrow W\xlongrightarrow{\iota}\hLie_{\Omega}\xlongrightarrow{\pi_\Omega}\gLie\longrightarrow 0
\]
	be the central extension defined by $\Omega$. Using this central extension, we construct the universal central extension of $\gLie$.
	
	Let $\aLie$ be an arbitrary $\field$-vector space (a Lie algebra with trivial bracket) and $\Theta\in Z^2(\gLie, \aLie)$. Since $\Theta(X,Y) = - \Theta(Y,X)$ , we have a $\field$-linear map
	\[
	\psi \colon W\to \aLie \quad \text{such that} \quad \Omega(X,Y)\mapsto \Theta(X,Y).
	\]
	We define $\phi'\colon\hLie_\Omega \to \hLie_\Theta$ by
	\[
	\phi'((X,U)):=(X,\psi(U)).
	\]
	Then, it is clear that the diagram
	\[
		\begin{tikzcd}
			\hLie_\Omega \arrow{d}{\phi'} \arrow{r}{\pi_\Omega}
			& \gLie \arrow{d}{\id} \\
			\hLie_\Theta \arrow{r}{\pi_\Theta}
			& \gLie
		\end{tikzcd}
		\]
commutes. Now set 
\[
\hat{\hLie} := [\hLie_\Omega, \hLie_\Omega].
\]
Since $\gLie$ is perfect, it follows that $\hat{\hLie} \oplus W = \hLie_\Omega$. This implies that
\[
\hat{\hLie} = [\hat{\hLie}\oplus W,\hat{\hLie}\oplus W]=[\hat{\hLie},\hat{\hLie}],
\]
i.e. $\hat{\hLie}$ is perfect.
Moreover, if we set 
\[
\cLie:=W\cap\hat{\hLie},
\]
then we obtain a central extension
\[
0\longrightarrow\cLie\xlongrightarrow{}\hat{\hLie}\xlongrightarrow{}\gLie\longrightarrow 0
\]
such that  $\hat{\hLie}$ is perfect. Defining $\phi:=\phi'|_{\hat{\hLie}}$ we get a  commutative diagram
\[
		\begin{tikzcd}
			\hat{\hLie} \arrow{d}{\phi} \arrow{r}{\pi_\Omega |_{\hat{\hLie}} }
			& \gLie \arrow{d}{\id} \\
			\hLie_\Theta \arrow{r}{\pi_\Theta}
			& \gLie\rlap{\ .}
		\end{tikzcd}
		\]
Therefore, $0\longrightarrow\cLie\xlongrightarrow{}\hat{\hLie}\xlongrightarrow{}\gLie\longrightarrow 0$ is the universal central extension.
\end{proof}

%

\vspace{0.2cm}
\section{Witt Algebra} \label{sec: witt}
Our main references for this section are \cite{Schottenloher2008} and \cite{Kac1987}.

\vspace{0.2cm}
The goal of this section is to prove that the Virasoro algebra is the unique universal nontrivial central extension of the Witt algebra.
\begin{defn}[Lie algebra of smooth vector fields]
	Let $M$ be a smooth compact manifold. The space $\vfields(M)$ is the \textit{space of smooth vector fields} on $M$. Here we consider $X\in\vfields(M)$ as a derivation $X\colon C^\infty(M)\to C^\infty(M)$, i.e. as an $\reals$-linear map with
	\[
	X(fg)=X(f)g+fX(g) \quad \forall f,g\in C^\infty(M).
	\]
	The Lie bracket of $X,Y\in \vfields(M)$ is the \textit{commutator}
	\[
	[X,Y]:=X\circ Y-Y\circ X
	\]
	which is also a derivation. Consequently, $(\vfields(M), [\cdot,\cdot])$ is an infinite dimensional Lie algebra over $\reals$.
\end{defn}

We will be interested in the case when $M=\OneTorus$. In this case, the space $C^\infty(\OneTorus)$ can be described as the vector space $C^\infty _{2\pi}(\reals)$ of 2$\pi$-periodic functions $\reals\to\reals$. Then  $\vfields(\OneTorus)=\{f\frac{d}{d\theta}|f\in C^\infty _{2\pi} (\reals)\}$  and $\OneTorus = \{e^{i\theta}|\theta\in\reals\}$. For $X=f\frac{d}{d\theta}$ and $Y=g\frac{d}{d\theta}$ we get
\[
[X,Y]=(fg'-f'g)\frac{d}{d\theta}.
\]
Since $f$ is smooth and periodic, it can be represented by a convergent Fourier series
\[
	f(\theta) = a_0 + \sum_{n=1} ^\infty (a_n \cos(n\theta)+b_n\sin(n\theta)).
\]
This leads to a natural (topological) generating system for $\vfields(\OneTorus)$:
\[
	\frac{d}{d\theta},\quad \cos(n\theta)\frac{d}{d\theta},\quad \sin(n\theta)\frac{d}{d\theta}.
\]
\\
Complexifying $\vfields(\OneTorus)$, i.e. by defining $\vfields^\complex (\OneTorus) : = \vfields(\OneTorus) \otimes \complex$, we finally arrive at:
\begin{defn}[Witt algebra]
	The \textit{Witt algebra} $\W$ is the linear span of $L_n$'s over $\complex$:
	\[
	\W:=\bigoplus_{n\in\integ} \complex L_n,
	\]
	where $L_n:=z^{1-n}\frac{d}{dz}=-i z^{-n} \frac{d}{d\theta}=-ie^{-in\theta}\frac{d}{d\theta}\in \vfields^\complex (\OneTorus)$ with $z=e^{i\theta}$ and $n\in\integ$.
\end{defn}

We note that $L_n \colon C^\infty(\OneTorus, \complex)\to C^\infty(\OneTorus, \complex)$, $\, f\mapsto z^{1-n}f'$, so that to prove that $\W$ with the Lie bracket in $\vfields^\complex(\OneTorus)$ is actually a Lie algebra over $\complex$ we need to show that $[\W,\W]\subset \W$.

 For $m,n\in\integ$ and $f\in C^\infty(\OneTorus, \complex)$
 \[
 L_m L_n f = z^{1-m}\frac{d}{dz}\left( z^{1-n}\frac{d}{dz}f \right)=(1-n)z^{1-m-n}\frac{d}{dz}f+z^{1-m}z^{1-n}\frac{d^2}{dz^2}f.
\] 
Therefore
\[
\begin{aligned}
	[L_m, L_n] f &= L_m L_n f - L_n L_m f = ((1-n)-(1-m))z^{1-m-n}\frac{d}{dz}f\\
		          &= (m-n) L_{m+n} f
\end{aligned}
\]
as required. Note that this actually implies that $[\W, \W] = \W$, i.e. that $\W$ is perfect.

\begin{thm}\label{thm:H2dim1}
	$\dim H^2(\W, \complex) = 1.$
\end{thm}

\begin{proof}
	Given an $\omega \in Z^2(\W, \complex)$, define $g_\omega:\W\to\complex$ by
	\[
	g_\omega(L_n):=
	\begin{cases}
	    \omega(L_0,L_n)/n & \text{if } n\neq 0,\\
	    0              & \text{if } n=0.
	\end{cases}
	\]
	Then $\hat{\omega}(x,y) :=\omega(x,y)+g_\omega([x,y])$ is such that
	\[
	\hat{\omega}\in Z^2(\W,\complex),\quad \omega-\hat{\omega}\in B^2(\W,\complex)\quad \text{and}\quad \hat{\omega}(L_0,x)=0\quad \forall x\in\W.
	\]
	Hence for any $\omega + B^2(\W, \complex)\in H^2(\W, \complex)$ we can take its representative $\omega$ such that $\omega(L_0, x) = 0\,\, \forall x\in\W$. Moreover, since $\omega\in Z^2(\W,\complex)$, we have
	\[
	\omega(L_m,[L_n,L_k])+\omega(L_n,[L_k,L_m])+\omega(L_k,[L_m,L_n])=0,
	\]
	and therefore
	\begin{equation}\label{eq:VirCentral}
	(n-k)\omega(L_m,L_{n+k})+(k-m)\omega(L_n,L_{k+m})+(m-n)\omega(L_k,L_{m+n})=0.
	\end{equation}
	First set $k=0$ in \eqref{eq:VirCentral} to get
	\[
	(n+m)\omega(L_m,L_n) = 0,
	\]
	since $\omega(L_0,x) =0$ and $\omega(x,y) = -\omega(y,x)$. This implies that 
	\[
	\omega(L_m,L_n)=\delta_{m+n,0}f(m)
	\]
	for some $f\colon\integ\to\complex$ such that $-f(-m)=f(m)$ since $\omega(L_m, L_n)=-\omega(L_n, L_m)$.
	Plugging this into \eqref{eq:VirCentral} with $m+n+k=0$ gives
	\begin{equation}\label{eq:VirCentral2}
	(2n+m)f(m)-(n+2m)f(n)+(n-m)f(m+n)=0.
	\end{equation}
	Setting $m=1$ in \eqref{eq:VirCentral2} gives us a linear recursion relation:
	\begin{equation}\label{eq:recurrence}
	(n-1)f(n+1)=(n+2)f(n)-(2n+1)f(1).
	\end{equation}
	Since $f(-n)=-f(n)$ we have $f(0)=0$ and thus we have to solve \eqref{eq:recurrence} only for $n>0$. The space of solutions of \cref{eq:recurrence} is at most 2-dimensional because if we know $f(1)$ and $f(2)$ we can calculate all $f(n)$'s using \eqref{eq:recurrence}. Note that $f(n)=n$ and $f(n)=n^3$ are solutions. Hence the general solution is $f(n)=\alpha n+ \beta n^3$, where $\alpha,\beta\in\complex$. However $\omega \in B^2(\W, \complex)$ if and only if  $f(n)=\alpha n$ (otherwise $f(n)$ is non-linear and thus $\omega \not\in B^2(\W,\complex)$). Hence, to get a nontrivial central extension, we must set $\beta \neq 0$ and $\alpha$ can be arbitrary, so following the usual convention we set $\alpha := -\beta$. Hence, $f(n) = \beta(n^3-n)$ and 
	\begin{equation}\label{eq:VirCocycle}
	\omega(L_m, L_n)= \delta_{m+n,0}\beta(n^3-n).
	\end{equation}
	Therefore, $\dim H^2(\W, \complex)$ = 1.
\end{proof}

Combining the above theorem with the previous observation that the Witt algebra is perfect, we can define the Virasoro algebra as the unique universal nontrivial central extension of $\W$ by $\complex$.

\begin{defn}[Virasoro algebra]
	The \textit{Virasoro algebra}
		\[
		\Vir :=\bigoplus_{n\in\integ} \complex L_n \oplus \complex C
		\]
	is the Lie algebra which satisfies the following commutation relations:
	\begin{alignat*}{2}
	&[L_m,L_n]&&=(m-n)L_{m+n}+\frac{C}{12}(m^3-m)\delta_{m+n,0}\, , \\
	&[\Vir,C]&&=0.
	\end{alignat*}
\end{defn}

\begin{rmk}
	The Virasoro algebra is defined by the nontrivial cocycle ${\omega\in H^2(\W,\complex)}$ from \cref{thm:H2dim1} by setting $\beta = C/12$ in \cref{eq:VirCocycle}. Cf. \cref{eq:extToBracket11} and \cref{eq:extToBracket12} from the proof of \cref{extensionsToCocycles}.
\end{rmk}

\vspace{0.2cm}
\section{Representation Theory of Virasoro Algebra}\label{sec: representations}


Our main references for this section are \cite{Schottenloher2008} and \cite{Kac1987}.

\vspace{0.2cm}
 Let $V$ be a vector space over $\complex$. 

\begin{defn}[Hermitian form, inner product]
A map \[\langle \cdot, \cdot\rangle \colon V\times V\to\complex\] is called a \textit{Hermitian form} if it is complex antilinear in the first variable, complex linear in the second and satisfies
\[
\langle v, w \rangle = \overline{\langle w, v \rangle} \quad \forall v,w \in V.
\]
A Hermitian form is an \textit{inner product} if moreover we have
\[
\langle v, v \rangle > 0 \quad \forall v\in V\setminus\{0\}.
\]
\end{defn}

\begin{defn}[Unitary representation of Virasoro algebra]\label{def:unitary rep of Vir}
	A map $\rho \colon \Vir\to \End_\complex V$ is called a \textit{representation} if it is a Lie algebra homomorphism. The representation $\rho$ is called \textit{unitary} if there is a positive semidefinite Hermitian form $\langle \cdot, \cdot \rangle : V\times V\to\complex$ such that $\forall v,w\in V$ and $\forall n\in\integ$ we have
	\begin{align*}
	\langle\rho(L_n)v,w\rangle &=\langle v,\rho(L_{-n})w\rangle,\\
	\langle\rho(C)v,w\rangle &=\langle v,\rho(C)w\rangle.
	\end{align*}
\end{defn}

\begin{defn}[Cyclic vector]
	A vector $v\in V$ is called a \textit{cyclic vector} for a representation $\rho\colon \Vir \to \End V$ if the set
	\[
	\bm{\left \{} \rho(X_1)\dots\rho(X_m)v \bm{\mid} X_j\in \Vir \text{ with }  j\in\{1,\dots,m\} \text{ and } m\in\nat\bm{\right\}}
	\]
	spans the vector space V.
\end{defn}

\begin{defn}[Highest weight representation, Virasoro module]
	A representation $\rho\colon\Vir \to \End V$ is called a \textit{highest weight representation} if there are complex numbers $h,c\in\complex$ and a cyclic vector $v_0\in V$ such that
	\begin{align*}
		\rho(C)v_0 &= cv_0,\\
		\rho(L_0)v_0 &= hv_0,\\
		\rho(L_n)v_0 &= 0 \quad\forall n\in \nat.
	\end{align*}
	The vector $v_0$ is then called the \textit{highest weight vector} (or \textit{vacuum vector}) and $V$ is called a \textit{Virasoro module} (via $\rho$) with \textit{highest weight} $(c,h)$ or simply a \textit{Virasoro module for} $(c,h)$.
\end{defn}

\begin{rmk}
The operator $L_0$ is often interpreted as the energy operator which is assumed to be diagonalizable with its spectrum bounded from below. With this assumption and the assumption that $v_0$ is an eigenvector of $\rho(L_0)$ with lowest eigenvalue $h\in\reals$, any representation $\rho$ preserving the energy spectrum property satisfies $\rho(L_n)v_0=0\,\, \forall n\in\nat$. This follows by noting that for $w := \rho(L_n)v_0$ we have
\[
\rho(L_0)w=\rho(L_n)\rho(L_0)v_0-n\rho(L_n)v_0=\rho(L_n)hv_0-nw=(h-n)w.
\]
Since we assumed $h$ to be the lowest eigenvalue of $\rho(L_0)$, $w$ has to vanish for $n>0$.
\end{rmk}

\begin{defn}[Verma module]
	A \textit{Verma module} for $c,h\in\complex$ is  a complex vector space $M(c,h)$ with a highest weight representation
	\[
	\rho\colon \Vir\to \End_\complex M(c,h)
	\]
	and a highest weight vector $v_0\in M(c,h)$ such that
	\begin{equation}\label{eq:M basis}
	\left \{\rho(L_{-n_1})\dots\rho(L_{-n_k})v_0\mid n_1\geq\dots\geq n_k \ge1, \, k\in\nat\right\}\cup\{v_0\}
	\end{equation}
	is a vector space basis of $M(c,h)$.
\end{defn}

Note that by the definition for fixed $c,h\in\complex$ the Verma module $M(c,h)$ is unique up to isomorphism.

\begin{lem}\label{lem:construction of Verma module}
For all $c,h\in\complex$ there exists a Verma module $M(c,h)$.
\end{lem}


\begin{proof}
	Let
	\[
	M(c,h):=\complex v_0 \bigoplus \complex\left\{v_{n_1\dots n_k}:n_1\geq\dots\geq n_k\geq 1\right \}
	\]
	be the complex vector space spanned by $v_0$ and $v_{n_1\dots n_k}$'s for $n_1\geq\dots\geq n_k\ge 1$. We define a map
	\[
	\rho\colon \Vir\to \End_\complex (M(c,h))
	\]
	by
	\begin{align*}
	\rho(C)&:= c \id_{M(c,h)}, \\
	\rho(L_0)v_0&:=h v_0,\\
	\rho(L_0)v_{n_1\dots n_k} &:= \left(\sum_{j=1}^k n_j +h\right)v_{n_1\dots n_k}, \\
	\rho(L_n)v_0&:=0\quad\quad\quad\;\; \forall n\in\nat,\\
	\rho(L_{-n})v_0&:=v_n \quad\quad\quad \forall n \in\nat, \\
	\rho(L_{-n})v_{n_1\dots n_k}&:=v_{n n_1\dots n_k} \;\; \forall n\geq n_1.
	\end{align*}
	For all other $v_{n_1 \dots n_k}$'s with $1\leq n < n_1$, $\rho(L_{-n})v_{n_1\dots n_k}$ can be obtained by permutation using the commutation relations $[L_m,L_n]=(m-n)L_{m+n}$ for $m\neq -n$. E.g. for $n_1>n\geq n_2$:

\begin{align*}
\rho(L_{-n})v_{n_1\dots n_k} &= \rho(L_{-n})\rho(L_{-n_1})v_{n_2\dots n_k} \\
&=\left(\rho(L_{-n_1})\rho(L_{-n})+(-n+n_1)\rho(L_{-(n+n_1)})\right)v_{n_2\dots n_k}\\
&=v_{n_1 n n_2\dots n_k} + (n_1-n)v_{(n_1+n)n_2\dots n_k} .
\end{align*}
So the above calculation guides us to define
\[
	\rho(L_{-n})v_{n_1\dots n_k} :=v_{n_1 n n_2\dots n_k} + (n_1-n)v_{(n_1+n)n_2\dots n_k} .
\]
Similarly we define $\rho(L_n)v_{n_1\dots n_k}\;\forall n\in\nat$ taking into account the commutation relations, e.g.
\[
	\rho(L_n)v_{n_1}:= \left\{\begin{alignedat}{2}
              & 0 							&&n>n_1,\\
              &\left(2nh+\frac{n}{12}(n^2-1)c\right)v_0\quad 	&&n=n_1,\\
              &(n+n_1)v_{n_1-n}			&& 0<n<n_1.
	  \end{alignedat}\right.
\]
Thus $\rho$ is well-defined and $\complex$-linear. It remains to show that $\rho$ respects the commutation relations, so that it is actually a Lie algebra representation, i.e. that
\[
\left [\rho(L_m),\rho(L_n) \right ] = \rho ([L_m,L_n]).
\]
E.g. for $n\ge n_1$ we have
\begin{align*}
[\rho(L_0),\rho(L_{-n})]v_{n_1\dots n_k} 
&= \rho(L_0)v_{n n_1\dots n_k} - \rho(L_{-n})\left(\sum_{j=1}^k n_j + h\right)v_{n_1\dots n_k} \\
&= \left( \sum_{j=1}^k n_j +n +h\right) v_{n n_1\dots n_k} - \left(\sum_{j=1} ^k n_j +h\right) v_{n n_1\dots n_k} \\
&= n v_{n n_1\dots n_k} = n \rho(L_{-n}) v_{n_1\dots n_k} \\
&= \rho \left( [L_0, L_{-n}] \right) v_{n_1\dots n_k}
\end{align*}
and for $n> m > n_1$
\begin{align*}
[\rho(L_{-m}),\rho(L_{-n})]v_{n_1\dots n_k} &= \rho(L_{-m})v_{n n_1\dots n_k} - v_{n m n_1\dots n_k} \\
&= v_{n\, m\, n_1 \dots n_k} + (n-m) v_{(n+m)\, n_1\dots n_k} - v_{n\, m\, n_1\dots n_k} \\
&= (n-m) v_{(n+m)\, n_1\dots n_k} = (n-m) \rho (L_{-(m+n)})v_{n_1\dots n_k} \\
&= \rho([L_{-m},L_{-n}])v_{n_1\dots n_k}.
\end{align*}
Other identities follow similarly. Hence $\rho$ is a highest weight representation. Thus, by construction $M(c,h)$ is a Verma module.
\end{proof}

\begin{cor}
Any Virasoro module $V$ of highest weight $(c,h)$ is isomorphic to a quotient of the corresponding Verma module $M(c,h)$. In particular, $(c,h)$ determines $M(c,h)$ uniquely.
\end{cor}

\begin{proof}
There exists a surjective homomorphism from $M(c,h)$ to $V$ which maps the highest weight vector of $M(c,h)$ to the highest weight vector of $V$ and commutes with the action of $\Vir$ since in $M(c,h)$ the set of vectors \eqref{eq:M basis} is linearly independent. Hence quotienting out the kernel of this homomorphism we obtain the desired result.
\end{proof}

\begin{defn}[Submodule of Virasoro module]
A \textit{submodule} $U$ of a Virasoro module $V$ is a $\complex$-linear subspace of $V$ with $\rho(D) U \subset U\;\, \forall D\in\Vir$, i.e. it is an invariant linear subspace of $V$.
\end{defn}

\begin{rmk}\label{graded decomposition}
Let $V$ be a Virasoro module for $c,h\in\complex$. Then there exists a direct sum decomposition $V=\bigoplus_{N\in\nato} V_N$ where $V_0:=\complex v_0$ and
\[
V_N:=\vspan\left(\left\{ \rho(L_{-n_1})\dots \rho(L_{-n_k})v_0 \middle | n_1\geq\dots\geq n_k\geq 1, \sum_{j=1}^k n_j=N, k\in\nat \right\} \right).
\]
The $V_N$'s are eigenspaces of $\rho(L_0)$ with the eigenvalue $N+h$, i.e.
\[
\rho(L_0)|_{V_N} = (N+h)\id_{V_N}.
\]
This follows from the definition of a Virasoro module and from the commutation relations.
\end{rmk}

\begin{lem}\label{lemma: submodule decomposition}
Let V be a Virasoro module for $c,h\in\complex$ and U a submodule of V. Then
\[
U= \bigoplus_{N\in\nato}(V_N\cap U).
\]
\end{lem}
\begin{proof}
	Let $w=w_0\oplus\dots \oplus w_s\in U$ with $w_j\in V_j$ for $j\in \{0,\dots, s\}$. Then
	\[
	\begin{alignedat}{10}
	w	        &= w_0 +\dots+ w_s, \\
	\rho(L_0)w &= hw_0 +\dots+ (s+h)w_s, \\
	&\;\,\vdots \\
	\rho(L_0)^{s}w &= h^{s}w_0 +\dots+ (s+h)^{s}w_s.
	\end{alignedat} 
	\]
This is a system of linear equations for $w_0,\dots, w_s$ with a regular coefficient matrix. Hence, the $w_0, \dots, w_s$ are linear combinations of the $w,\dots,\rho(L_0)^s w\in U$. Thus $w_j\in V_j\cap U\;\;\forall j\in\{0,\dots,s\}$.
\end{proof}
\vspace{0.2cm}

We will mostly need unitary representations of the Virasoro algebra. To define a suitable Hermitian form, we need the notion of an expectation value first.

\begin{defn}[Expectation value]
Let $V=\bigoplus_{N\in\nato}V_N$ be a Virasoro module and $w\in V$. Then according to \cref{graded decomposition}, $w$ has a unique component $w_0\in V_0$ with respect to the decomposition $\bigoplus_{N\in\nato}V_N$. The \textit{expectation value} of $w$, denoted $<w>$,  is the coefficient of $w_0 \in V_0$ with respect to the basis $v_0$, i.e. $w_0=<w>v_0$. 
\end{defn}

In what follows we will often abuse our notation and simply write $L_n$ for $\rho (L_n)$.

\begin{defn}[Hermitian form on Verma module]\label{def:Hermitian form on Verma module}
Let $M=M(c,h)$ with $c,h\in\reals$ be the Verma module with a highest weight representation $\rho\colon \Vir\to \End_\complex (M(c,h))$ and let $v_0$ be the respective highest weight vector. A \textit{Hermitian form on M} with respect to the basis ${\{v_{n_1\dots n_k}\}\cup\{v_0\}}$ is defined as
\[
\langle v_{n_1\dots n_k}, v_{m_1\dots m_j} \rangle := < L_{n_k}\dots L_{n_1} v_{m_1\dots m_j} > = <L_{n_k}\dots L_{n_1} L_{-m_1} \dots L_{-m_j} v_0 >.
\]
\end{defn}

Note that from the above definition it follows that
\[
\langle v_0, v_0 \rangle = 1 \quad \text{and} \quad \langle v_0, v_{n_1 \dots n_k} \rangle = 0 = \langle v_{n_1 \dots n_k},v_0 \rangle.
\]
The condition $c,h\in\reals$ implies that $\langle v, v' \rangle = \langle v', v\rangle$ for all basis vectors
\[
v,v'\in B:= \{v_{n_1\dots n_k} \mid n_1\geq\dots\geq n_k\geq1 \}\cup \{v_0\}.
\]
The proof of the above consists of repeated use of the commutation relations of $L_n$'s. 

The map $\langle \cdot, \cdot \rangle \colon B\times B\to \reals$ has an $\reals$-bilinear continuation to $M\times M$ which is $\complex$-antilinear in the first and $\complex$-linear in the second variable: for $w, w'\in M$ with unique representations $w=\sum \lambda_j w_j$, $w'=\sum \mu_k w_k'$ with respect to basis vectors $w_j, w_k'\in B$, one defines
\[
\langle w, w'\rangle := \sum \sum \overline{\lambda}_j \mu_k \langle w_j, w_k'\rangle.
\]
By the above discussion, the map $\langle \cdot, \cdot \rangle \colon M\times M\to \complex$ is a Hermitian form. However, it is not positive definite or positive semidefinite in general. To check this, the Kac determinant is used. Before defining it, we need some more results about the Hermitian form.

\begin{thm}\label{thm:Verma unitary props}
Let $c,h\in\reals$ and $M=M(c,h)$. Then
 	\begin{enumerate}
 		\item $\langle \cdot, \cdot \rangle \colon M \times M \to \complex$ is the unique Hermitian form satisfying $\langle v_0, v_0 \rangle = 1$, $\langle L_n v,w \rangle = \langle v, L_{-n} w \rangle$ and $\langle Cv,w\rangle = \langle v, Cw\rangle\quad \forall v,w\in M, \,\forall n\in\integ$.\label{item:Verma1}
 		\item The eigenspaces of $L_0$ are pairwise orthogonal, i.e. if $M\neq N$, then ${\langle v,w\rangle = 0}\quad {\forall v\in V_M}, \forall w\in V_N$.\label{item:Verma2}
 		\item The maximum proper submodule of $M$ is $\ker\, \langle \cdot, \cdot \rangle$.\label{item:Verma3}
 	\end{enumerate}
\end{thm}

\begin{proof}

	\ref{item:Verma1} That the identity
	\[
	\langle L_n v,w\rangle = \langle v, L_{-n} w \rangle
	\]
	holds can be seen using commutation relations. The uniqueness of such a form follows from
	\[
	\langle v_{{n_1}\dots{n_k}},v_{{m_1}\dots{m_j}}\rangle = \langle v_0,L_{n_k}\dots L_{n_1} v_{{m_1}\dots{m_j}}\rangle .
	\]
	
	\ref{item:Verma2} Assume that $N>M$. Then any $\langle v, w\rangle $ with $v\in V_N$ and $w\in V_M$ can be written as a sum of elements of the form $<L_{n_k}\dots L_{n_1} L_{-m_1}\dots L_{-m_j}v_0>$ with $n_1+\dots n_k =N$ and $m_1+\dots m_j=M$. However using the commutation relations we can move $L_n$'s to front and get a sum of expectation values where $L_s, s\in\nat$, acts directly on $v_0$. Thus, $<L_{n_k}\dots L_{n_1} L_{-m_1}\dots L_{-m_j}v_0>=0$ and hence  $\langle v, w\rangle=0$.
	
	\ref{item:Verma3} If $v\in \ker\,\langle \cdot, \cdot \rangle := \{ u\in M \mid \langle w, u \rangle =0\; \forall w\in M \}$, then $L_n v \in  \ker\,\langle \cdot, \cdot \rangle\,$ $\forall n\in\integ$ because $\langle w, L_n v \rangle =  \langle L_{-n}w, v \rangle = 0$. Moreover, $v_0\not\in M$ since $\langle v_0, v_0 \rangle=1$. Hence, $\ker\,\langle \cdot, \cdot \rangle$ is a proper submodule of $M$.

To prove maximality, let $U\subset M$ be an arbitrary proper submodule and let $u \in U$. For $n_1\ge\dots\ge n_k \ge 1$ one has $\langle v_{n_1 \dots n_k},u \rangle = \langle v_0, L_{n_k}\dots L_{n_1} u\rangle$. If $\langle v_{n_1 \dots n_k},u \rangle \ne 0$, then $<L_{n_k}\dots L_{n_1}u>\ne 0$. By \cref{lemma: submodule decomposition} and part (b) of the current theorem we see that in this case $v_0\in U$ because $L_{n_k}\dots L_{n_1}u\in U$, and that $v_{n_1\dots n_k}\in U$. Since $v_{n_1\dots n_k}$ is an arbitrary basis vector of $M$, this implies that $M=U$ contradicting properness of $U\subset M$. Thus,  $\langle v_{n_1 \dots n_k},u \rangle=0$. Similarly, $\langle v_0, u \rangle =0$, so ${u\in \ker \langle\cdot,\cdot\rangle}$.
\end{proof}

\begin{rmk}
$M(c,h)/ \ker\, \langle\cdot,\cdot\rangle$ is a Virasoro module with a nondegenerate Hermitian form $\langle\cdot,\cdot\rangle$. However, in general $\langle\cdot,\cdot\rangle$ is not definite.
\end{rmk}

\begin{cor}
If $\langle\cdot,\cdot\rangle$ is positive semidefinite, then $c\ge0$ and $h\ge0$.
\end{cor}

\begin{proof} We have
\[
\langle v_n, v_n\rangle = \langle v_0, L_n L_{-n} v_0\rangle = \langle v_0, [L_n, L_{-n}] v_0\rangle = 2nh+\frac{c}{12}(n^3-n)\quad \forall n\in\nat.
\]
Now $\langle v_1,v_1\rangle\ge0 \iff  h\ge0$. Moreover, $\langle v_n,v_n\rangle \ge 0 \iff 2nh+\frac{c}{12}(n^3-n)\ge0$. Therefore, $\langle v_n,v_n\rangle \ge 0$ is valid for all $n\in\nat$ if and only if $c\ge0$ and $h\ge0$.
\end{proof}

We need some general results before continuing with unitarity.
\begin{defn}[(In)decomposable representation]
A representation $M$ is \textit{indecomposable} if there are no invariant proper subspaces $U,V$ of $M$ such that $M=U\oplus V$. Otherwise $M$ is \textit{decomposable}.
\end{defn}

\begin{defn}[(Ir)reducible representation]
A representation $M$ is called \textit{irreducible} if there is no invariant proper subspace $V$ of $M$. Otherwise $M$ is called \textit{reducible}.
\end{defn}

\begin{thm}\label{thm:Verma module indecomposable etc}
For each $(c,h)$ we have
\begin{enumerate}
	\item The Verma module $M(c,h)$ is indecomposable.\label{item:VermaModule1}
	\item There exists a unique maximal subrepresentation $J(c,h)$ of $M(c,h)$ and
	\begin{equation}\label{eq:L maximal submodule of Verma}
		L(c,h) := M(c,h)/J(c,h)
	\end{equation}
	is the unique irreducible highest weight representation with highest weight $(c,h)$.
	\label{item:VermaModule2}
\end{enumerate}
\end{thm}

\begin{proof}
\ref{item:VermaModule1} Let $V,W$ be invariant subspaces of $M=M(c,h)$ and $M=V\oplus W$. By \cref{lemma: submodule decomposition} there exist direct sum decompositions
\[
V=\bigoplus(M_j \cap V)\quad \text{and}\quad W=\bigoplus(M_j \cap W).
\]
Since $\dim M_0 =1$, this implies that $M_0\cap V=0$ or $M_0 \cap W=0$. So the highest weight vector $v_0$ is contained in $V$ or in $W$. But if $v_0$ belongs to a subrepresentation, then this subrepresentation must coincide with $M$.

\ref{item:VermaModule2}  By \cref{lemma: submodule decomposition} all proper subrepresentations are graded. Thus, their sum is graded too. The sum is also a proper subrepresentation since it does not contain the vacuum vector $v_0$. The maximal subrepresentation $J(c,h)$ is thus the sum of all proper subrepresentations. Hence the proof.

\end{proof}


\begin{rmk}\label{rmk:unique L}
Combining \cref{thm:Verma module indecomposable etc}\,\ref{item:VermaModule2} with \cref{thm:Verma unitary props}\,\ref{item:Verma3} we see that $J(c,h)= \ker\langle\cdot,\cdot\rangle$ and hence $L(c,h)$ is the unique unitary positive definite highest weight representation of $\Vir$, provided that $M(c,h)$ is unitary and positive semidefinite. Indeed, if $\rho:\Vir\to End_{\complex} (V)$ is a positive definite unitary highest weight representation with vacuum vector $v_0'\in V$ and Hermitian form $\langle\cdot,\cdot\rangle '$ we can define a surjective linear homomorphism $\vphi : M(c,h) \to V$
\[
v_0\mapsto v_0',\quad v_{n_1 \dots n_k} \mapsto \rho(L_{-n_1\dots -n_k})v_0',
\]
which also respects the Hermitian forms:
\[
\langle \vphi(v), \vphi(w)\rangle ' = \langle v,w\rangle.
\]
Therefore, $\langle\cdot,\cdot\rangle$ is positive semidefinite and $\vphi$ factorizes over $L(c,h)$ leading to an isomorphism $\bar{\vphi}:L(c,h)\to V$.
\end{rmk}

\begin{defn}\label{def: Kac det}
Let $P(N):=dim_{\complex}V_N$ and $\{b_1, \dots, b_{P(N)}\}$ be a basis of $V_N$. We define matrices $A^N$ by $A_{ij}^N:=\langle b_i, b_j\rangle$ for $i,j\in\{1,\dots, P(N)\}$.
\end{defn}

Clearly, $\langle\cdot,\cdot\rangle$ is positive semidefinite if all matrices $A^N$ are positive semidefinite. For $N=0$ and $N=1$ we get $A^0= (1)$ and $A^1=(2h)$ with respect to the bases $\{v_0\}$ and $\{v_1\}$. For example, to get $A^2$ we calculate
\begin{align*}
\langle v_2,v_2 \rangle &= <L_2 L_{-2} v_0> = <4L_0 v_0 + \frac{c}{2} v_0> = 4h +\frac{c}{2},\\
\langle v_{1,1},v_{1,1}\rangle&=8h^2+4h,\\
\langle v_2, v_{1,1}\rangle&=6h.
\end{align*}
Thus, relative to the basis $\{v_2,v_{1,1}\}$
\[
A^2=
  \begin{pmatrix}
    4h +c/2 & 6h \\
    6h & 8h^2+4h
  \end{pmatrix}.
\]
Therefore, $A^2$ is (for $c\ge0$ and $h\ge0$) positive semidefinite if and only if
\[
\det A^2=2h(16h^2-10h+2hc+c)\ge0.
\]

\begin{thm}[Kac determinant formula]\label{thm:Kac det}
The Kac determinant $\det A^N$ depends on $(c,h)$ as follows
\[
\det A^N(c,h) = K_N \prod_{\substack{p,q\in\nat\\ pq\le N}}(h-h_{p,q}(c))^{P(N-pq)},
\]
where $K_N\ge0$ is a constant, $P(N)$ is as in \cref{def: Kac det} and
\[
h_{p,q}(c):=\frac{1}{48}((13-c)(p^2+q^2)+\sqrt{(c-1)(c-25)}(p^2-q^2)-24pq-2+2c).
\]
\end{thm} 
For a proof check \cite[Chap. 8]{Kac1987} or \cite[Chap. 4]{Iohara2011}.

\begin{thm}\label{thm:unitarity of Verma module}
Let $c,h\in\reals$.
 	\begin{enumerate}[(i)]
 		\item $M(c,h)$ is unitary positive definite for $c>1, h>0$ and positive semidefinite for $c\ge1, h\ge0$.\label{item:Kac det1}
 		\item\label{item:Kac det2} $M(c,h)$ is unitary in the region $0\le c<1,\; h>0$ if and only if $(c,h)=( c(m), h_{p,q}(m) )$ where
 			\end{enumerate}
 		\begin{align}
 			 \quad c(m)&= 1-\frac{6}{(m+2)(m+3)},\quad m\in\nato,\label{eq:c(m) definition}\\
 			 \quad h_{p,q}(m)&=\frac{\left((m+3)p-(m+2)q\right)^2-1}{4(m+2)(m+3)},\quad p,q\in\nat \text{ and } 1\le p\le q \le m+1.\nonumber
 		\end{align} 
\end{thm}

For a proof of \ref{item:Kac det2} see \cite{Friedan1986} where the authors have shown that the Hermitian form $\langle\cdot,\cdot\rangle$ can only be unitary in the region $0\le c<1$ for $\left(c(m),h_{p,q}(m)\right)$ and \cite{Goddard1986} where the authors have proved that $M(c,h)$ actually gives a unitary representation in all these cases.

To prove part \ref{item:Kac det1} we first need an example of a Virasoro algebra representation.
\subsection{Fock Space Representation of Virasoro Algebra}\label{sec: fock rep}

\begin{defn}[Heisenberg algebra]
Let $\Heisenberg$ be the \textit{Heisenberg algebra}, the complex Lie algebra with a basis $\{a_n,\, \hbar \mid n\in\integ\}$ subject to the commutation relations
\begin{equation}\label{eq: Heisenberg commutation}
	[a_m,a_n]=m \delta_{m+n,0} \hbar, \quad\quad [\hbar, a_n] = 0\quad\quad \forall m,n\in\integ.
\end{equation}
Define the Fock space $\Fock := \complex [x_1, x_2, \dots]$; this is the space of polynomials in infinitely many variables $x_1,x_2, \dots\;$.

Given $\mu,\hbar\in\reals$, define the following representation $\rho$ of $\Heisenberg$ on $\Fock\; \forall n\in\nat$:
\begin{equation}\label{eq: Fock rep of Heisenberg}
\begin{aligned}
	\rho(a_n) &:=\frac{\partial}{\partial x_n},\\
	\rho(a_{-n}) &:=nx_n,\\
	\rho(a_0) &:=\mu \id_{\Fock},\\
	\rho(\hbar) &:=\hbar \id_{\Fock}.
\end{aligned}
\end{equation}
\end{defn}

Clearly the commutation relations \eqref{eq: Heisenberg commutation} hold in Fock representation \eqref{eq: Fock rep of Heisenberg}. Moreover, the Fock representation is irreducible and unitary.
\begin{lem}
If $\hbar \ne 0$, then the representation \eqref{eq: Fock rep of Heisenberg} is irreducible.
\end{lem}
\begin{proof}
Any polynomial in $\Fock$ can be reduced to a multiple of 1 by successive application of $a_n$'s with $n>0$. Then the successive application of $a_{-n}$ with $n>0$ can give any other polynomial in $\Fock$ provided that $\hbar \ne 0$.
\end{proof}

\begin{lem}
For each $\mu\in\reals$ there exists a unique positive definite Hermitian form $\langle \cdot,\cdot \rangle$ on $\Fock$ such that
\[
\langle 1, 1 \rangle = 1\quad \text{and}\quad \langle \rho(a_n)f,g\rangle = \langle f, \rho (a_{-n})g\rangle\quad \forall f,g\in\Fock,\; \forall n\in\integ.
\]
Here and in what follows 1 is the vacuum vector.
\end{lem}
\begin{proof}
First, we need to prove that the Hermitian form of two distinct monomials is zero. So let $f,g\in\Fock$ be two distinct monomials. Then there exists an index $n\in\nat$ and exponents $k\ne l,\; k,l\ge0$, such that $f=x^k_n f_1$ and $g=x_n^l g_1$ for suitable monomials $f_1,g_1$ independent of $x_n$. Without loss of generality assume that $k<l$. We now calculate $\langle f, g\rangle n^{k+1}$ in two different ways:
\[
\left\langle (\rho(a_n))^{k+1}f,x_n^{l-k-1}g_1\right\rangle=\left\langle \left(\frac{\partial}{\partial x_n}\right)^{k+1} x_n^k f_1, x_n^{l-k-1}g_1\right\rangle=\left\langle 0, x_n^{l-k-1}g_1\right\rangle = 0
\]
and
\[
\langle(\rho(a_n))^{k+1}f,x_n^{l-k-1}g_1\rangle = \langle f,(\rho(a_{-n}))^{k+1}x_n^{l-k-1}g_1\rangle=\langle f, n^{k+1} x_n^l g_1\rangle = \langle f, g\rangle n^{k+1}.
\]
Thus, $\langle f, g\rangle=0$. Moreover,
\[
\langle f, f\rangle = \langle f, n^{-k}(\rho(a_{-n}))^k f_1\rangle = n^{-k} \langle \rho(a_n)^k x_n^k f_1,f_1 \rangle=\frac{k!}{n^k}\langle f_1, f_1 \rangle.
\]
By definition $\langle 1,1\rangle =1$. Thus it follows that for monomials $f=x_{n_1}^{k_1}x_{n_2}^{k_2}\dots x_{n_r}^{k_r}$ with $n_1<n_2<\dots<n_r$
\begin{equation}\label{eq: def of H}
\langle f, f\rangle = \frac{k_1!\,k_2!\,\dots k_r!}{n_1^{k_1} n_2^{k_2}\dots n_r^{k_r}}.
\end{equation}
Since the monomials constitute a (Hamel) basis of $\Fock$, $\langle\cdot,\cdot\rangle$ is uniquely determined as a positive definite Hermitian form by \eqref{eq: def of H} and the orthogonality condition. Reversing the arguments, by using \eqref{eq: def of H} and the orthogonality condition $\langle f,g \rangle=0$ for distinct monomials $f,g\in\Fock$ as a definition of $\langle\cdot,\cdot\rangle$, we obtain a Hermitian form on $\Fock$ with the desired properties.
\end{proof}
Note that $\rho(a_n)^* = \rho(a_{-n})$ and for each $n>0$ the operator $\rho(a_n)$ is an annihilation operator whereas $\rho(a_n)^*$ is a creation operator. This justifies another common name of the Heisenberg algebra---the \textit{oscillator algebra}.

Set $\hbar = 1$ and let
\[
\rho(L_n):=\frac{1}{2}\sum_{k\in\integ} \normord{\rho(a_{n-k})\rho(a_k)} \quad n\in\integ,
\]
where the colons indicate \textit{normal ordering} defined by
\begin{equation*}
	\normord{\rho(a_i)\rho(a_j)} = 
		\begin{cases}
		\rho(a_i)\rho(a_j) &\text{if } i\le j\\
		\rho(a_j)\rho(a_i) &\text{if } i>j.
		\end{cases}
\end{equation*}
Due to normal ordering, when an operator $\rho(L_n)$ is applied to any vector of $\Fock$ only a finite number of terms in the sum are non-zero. Hence, $\rho(L_n):\Fock\to\Fock$ is a well-defined map. From now on, we abuse our notation and write $L_n$ for $\rho(L_n)$ and similarly for $\rho(a_n)$.

\begin{prop}\label{prop: Fock space commutation}
The $L_n$'s satisfy the commutation relations
\begin{equation}\label{eq: Fock space commutation}
[L_m, L_n] = (m-n)L_{m+n} + \frac{1}{12}(m^3-m)\delta_{m+n,0}.
\end{equation}
Thus the map $\rho:\Heisenberg\to \End_{\complex} \Fock$ is a representation of the Virasoro algebra in the Fock space $\Fock$ for $c=1$. 
\end{prop}

\begin{proof}
We define a cutoff function $\psi$ on $\reals$ by:
	\[
	\psi(x)=
	\begin{cases}
	1 	&\text{if}\quad |x| \le 1,\\
	0	& \text{if}\quad |x|>1.
	\end{cases}
	\]
Let
	\[
	L_n(\veps)=\frac{1}{2}\sum_{j\in\integ} \normord{a_{n-j}a_j} \psi(\veps j)\,.
	\]
	Notice that $L_n(\veps)$ is a finite sum if $\veps \ne 0$ and that $L_n(\veps)\to L_n$ as $\veps\to 0$. In particular, the latter statement means that given $v\in\Fock$, $L_n(\veps)(v)=L_n(v)$ for $\veps$ sufficiently small. Furthermore, note that $L_n(\veps)$ differs from $1/2\sum_{j\in\integ} a_{n-j} a_j\, \psi (\veps j)$ by a finite sum of scalars. These terms drop out of the commutator $[a_k, L_n(\veps)]$. Hence
	\begin{align*}
	[a_k, L_n(\veps)]=&\frac{1}{2}\sum_j [a_k, a_{n-j}a_{j}] \psi (\veps j)\\
	=&\frac{1}{2}\sum_j [a_k,a_{n-j}]a_j \psi(\veps j) +\frac{1}{2} \sum_j a_{n-j}[a_k,a_j]\psi(\veps j)\\
	=& \frac{1}{2} k a_{k+n} \psi(\veps (k+n)) + \frac{1}{2} k a_{n+k} \psi(-\veps k)
	\end{align*}
	where for the last equality we have used the Heisenberg commutation relations \eqref{eq: Heisenberg commutation}. Letting $\veps\to 0$ gives us
	\[
	[a_k, L_n] = k a_{k+n}\quad \forall k,n\in\integ.
	\]
Using this result we calculate
	\begin{align}\label{eq: commutation of L_n}
	[L_m(\veps), L_n] &= \frac{1}{2}\sum_j [a_{m-j} a_{j}, L_n] \psi(\veps j)\nonumber\\
		&=\frac{1}{2}\sum_j j a_{m-j}a_{j+n} \psi(\veps j) + \frac{1}{2} \sum_j (m-j) a_{m-j+n} a_j \psi(\veps j).
	\end{align}
	We note that
\begin{equation}\label{eq: comm1}
\begin{aligned}
	\bigsum_j a_{m-j}a_{j+n}\psi(\veps j) &= \bigsum_{\frac{m-n}{2}\le j} \normord{a_{m-j} a_{j+n}}\psi(\veps j) + \bigsum_{j<\frac{m-n}{2}} a_{m-j} a_{j+n}\psi(\veps j) \\
	&=\bigsum_j\normord{a_{m-j}a_{j+n}}\psi(\veps j) +\delta_{m+n,0} \bigsum_{j<\frac{m-n}{2}} (m-j) \psi(\veps j)\\
	&=\sum_j\normord{a_{m-j}a_{j+n}}\psi(\veps j) +\delta_{m+n,0} \sum_{j<m} (m-j) \psi(\veps j).
\end{aligned}
\end{equation}
Similarly 
	\begin{equation}\label{eq: comm2}
	\bigsum_j a_{m-j+n} a_j\psi(\veps j) =
		 \bigsum_j \normord{a_{m-j+n}a_j}\psi(\veps j) -  \delta_{m+n,0}\bigsum_{j<0} j\,\psi(\veps j).
	\end{equation}
	Plugging Equations \eqref{eq: comm1} and \eqref{eq: comm2} into Equation \eqref{eq: commutation of L_n} we get
	\begin{align*}
	[L_m(\veps), &L_n] =\\
	 \frac{1}{2}&\sum_j j \normord{a_{m-j}a_{j+n}} \psi(\veps j)+\frac{1}{2}\sum_j (m-j)\normord{a_{m-j+n}a_j}\\
	  +&\delta_{m+n,0}\left (\frac{1}{2}\sum_{j=0}^{m-1}(m-j)j\,\chi_{[1,\infty)}(m)   -\frac{1}{2}\sum_{j=m}^{-1}(m-j)j\,\chi_{(-\infty,-1]}(m)\right)\psi(\veps j).
	\end{align*}
	Here $\chi_A(x)$ is the characteristic function. Both of the sums under the bracket sum up to $1/12(m^3-m)$. Making a variable transformation $j\mapsto j-n$ in the first sum and taking the limit $\veps\to 0$ we get the desired result \eqref{eq: Fock space commutation}.
\end{proof}

\begin{rmk}
One can also prove Proposition \ref{prop: Fock space commutation} without using a cutoff function. However, this method  requires more calculations to treat all the different cases separately. See, e.g., \cite[Chap. 7]{Schottenloher2008}.
\end{rmk}

\begin{cor}
The representation of Proposition \ref{prop: Fock space commutation} yields a positive definite unitary highest weight representation of the Virasoro algebra with the higest weight $c=1,\; h=1/2\, \mu^2$, where $\mu\in\reals$ is such that $\rho(a_0) :=\mu \id_{\Fock}$.
\end{cor}

\begin{proof}
For the highest weight vector $v_0:=1$ let
\[
V:= \Span_{\complex} \left\{L_nv_0 \mid n \in\integ \right\}.
\]
The restrictions of $\rho(L_n)$ to the subspace $V\subset \Fock$ of the Fock space $\Fock$ define a highest weight representation of $\Vir$ with the highest weight $(1, 1/2 \,\mu^2)$ and Virasoro module $V$.
\end{proof}

\begin{rmk}
In most cases $\Fock =V$. But it does not hold, e.g., if $\mu=0$.
\end{rmk}

More unitary highest weight representations can be constructed by taking tensor products:
\[
(\rho \otimes \rho)(L_n)(f_1\otimes f_2) := (\rho(L_n)f_1)\otimes f_2 + f_1\otimes(\rho(L_n)f_2)\quad \forall \;(f_1\otimes f_2)\in V\otimes V.
\]
The Hermitian form on $V\otimes V$ is defined by
\[
\langle f_1\otimes f_2 ,g_1\otimes g_2 \rangle = \langle f_1, g_1 \rangle \langle f_2, g_2\rangle. 
\]
These observations lead to the following.

\begin{prop}\label{prop: unitary reps}
The representation $\rho\otimes\rho : \Vir\to \End_{\complex} (V\otimes V)$ is unitary positive definite with highest weight $c=2,\; h=\mu^2\ne0$. Iterating we obtain positive semidefinite unitary highest weight representations $\forall (c,h)\in \nat \times \reals$ which are positive definite if $c\ge2$ and $h>0$.
\end{prop}

Now we can finally prove Theorem \ref{thm:unitarity of Verma module}\,\ref{item:Kac det1}.

\begin{proof}[{Proof of Theorem  \ref{thm:unitarity of Verma module}\,\ref{item:Kac det1}}]
Let
\[
\vphi_{p,q}=
\begin{cases}
	 h-h_{q,q}(c) & \text{if}\quad p=q,\\
	(h-h_{p,q}(c))(h-h_{q,p} (c)) & \text{if}\quad p\neq q.
\end{cases}
\]
Then by Theorem \ref{thm:Kac det}
\[
\det A^N (c,h) = K_N \prod_{\substack{p,q\in\nat \\ pq\le N, q\le p}} \vphi_{p,q}^{P(N-pq)}\,.
\]
For $1\le p,q\le N$ and $c>1,\; h>0$ we have
\begin{align*}
	\vphi_{q,q}&=h+\frac{1}{24}(c-1)(q^2-1)>0,\\
	\vphi_{p,q}&=\left(h-\left(\frac{p-q}{2}\right)^2\right)^2+\frac{h}{24}(p^2+q^2-2)(c-1)+\\
	&+\frac{1}{576}(p^2-1)(q^2-1)(c-1)^2+\frac{1}{48}(c-1)(p-q)^2(pq+1)>0.
\end{align*}
Hence, $\det A^N(c,h)>0$ for all $c>1,\, h>0$. This implies that the Hermitian form $\langle \cdot,\cdot\rangle$ is positive definite in the entire region $c>1, \; h>0$ if there is just one example $M(c,h)$ with $c>1,\;h>0$ such that $\langle\cdot,\cdot\rangle$ is positive definite. Proposition \ref{prop: unitary reps} shows that we have positive semidefinite representations for $c\in\nat,\, h\ge 0$, and positive definite for $c=2,3,\dots$, $h>0$ thereby proving \cref{thm:unitarity of Verma module}\,\ref{item:Kac det1}.
\end{proof}






	\chapter{Vertex Algebras}\label{chap:vertex}

Borcherds introduced vertex algebras in \cite{Borcherds1986} to understand Frenkel's work on the Lie algebra whose Dynkin diagram is the Leech lattice. Then Frenkel, Lepowsky and Meurman modified the definition and added some natural assumptions to vertex algebras which led to vertex operator algebras. This allowed them to construct the moonshine module \cite{Frenkel1988}---a vertex operator algebra with the monster group, the largest sporadic finite simple group, being its symmetry group. Finally in \cite{Borcherds1992} Borcherds proved the Conway--Norton monstrous moonshine conjecture \cite{Conway1979} for the moonshine module. The conjecture relates the monster group and modular functions, so it was rather unexpected. For this and related work, Borcherds was awarded a Fields Medal in 1998. Thus, vertex algebras are definitely of interest to mathematicians.

The interest of physicists stems from the fact that Frenkel, Lepowsky and Meurman were using ideas from conformal field theory and string theory in their work. Thus, it is no surprise that vertex (operator) algebras can be viewed as a mathematical axiomatization of chiral conformal field theory and indeed we will see that, for example, the operator product expansion, a crucial assumption made in 2D CFT, can be rigorously proved in vertex algebras (\cref{thm:OPE for vertex algebras}). A nice, but a little bit outdated, overview of these connections can be found in the introduction of \cite{Frenkel1988}.\\

To understand the current work, no prior knowledge of vertex algebras is assumed. We present full proofs up to \Cref{sec:VOAs and unitarity}. Our particular choice of material is tailored so that we are able to give full details of the proof of Kac's \cref{thm:Wightman to VA} up to the level found in \cite{Kac1998}. In the last section, however, some proofs are skipped, but freely available references are given. As elsewhere in this work, we only consider bosonic theories, but the generalization to vertex superalgebras which also include fermions is rather trivial, see, e.g., \cite{Kac1998}. \\

\begin{flushleft}
Our main references for \crefrange{sec:formal distributions}{sec:Mobius conformal VAs} is \cite{Kac1998} and \cite{Schottenloher2008}. For \Cref{sec:VOAs and unitarity} we have mostly used \cite{Carpi2015}.
\end{flushleft}

\section{Formal Distributions}\label{sec:formal distributions}


Throughout this chapter let $Z=\{z_1,\dots, z_n\}$ be a set of variables and $U$ be a vector space over $\complex$. A \textbf{formal distribution} is a series
\[ 
a(z_1,\dots, z_n) = \sum_{j\in\integ^n}a_j z^j = \sum_{j\in\integ^n} a_{j_1,\dots,j_n}z_1^{j_1}\ldots z_n^{j_n} , 
\]
with coefficients $a_j\in U$. The \textbf{vector space of formal distributions} over $\complex$ will be denoted by $U[[z_1^{\pm},\ldots, z_n^{\pm}]]=U[[z_1,\dots,z_n,z_1^{-1},\dots,z_n^{-1}]]$. It contains the subspace of \textbf{Laurent polynomials}
\[ 
U[z_1^{\pm},\dots, z_n^{\pm}]=\left\{a\in U[[z_1^{\pm},\dots, z_n^{\pm}]] \mid \exists k,l: a_j=0 \text{ except for } k\le j \le l \right\},
\] 
with the partial order on $\integ^n$ defined by $i\le j\iff i_{\mu}\le j_{\mu}$ $\forall \mu\in\{1,2,\ldots,n\}$. The space of \textbf{formal Laurent series} is
\[
	U((z))=\left\{a\in V[[z^{\pm}]] \mid \exists k \in\integ\; \forall j\in\integ: j<k\implies a_j=0\right\}.
\]

A formal distribution can be always multiplied by a Laurent polynomial (provided that the product of coefficients is defined), but two formal distributions cannot be multiplied in general. For each product of two formal distributions, we need to check that it converges in the algebraic sense, i.e. the coefficient of each monomial $z_1^{j_1} \dots z_n^{j_n}$ must be a finite or at least a convergent sum. Here and further by multiplication of formal distributions we mean the usual Cauchy product: for $a(z)=\sum_n a_n z^n$ and $b(z)= \sum_n b_n z^n$, the Cauchy product is
\[
	a(z)b(z) = \sum_n \left(\sum_{i+j=n} a_i b_j\right)z^n.
\]

Given a formal distribution $a(z)=\sum_{n\in\integ} a_n z^n$, the \textbf{residue} is defined as
\[
	\Res_z a(z)= a_{-1}. 
\]
Defining the \textbf{derivative} of a formal distribution $a(z)=\sum_{n\in\integ}a_n z^n$ by 
\[
\partial a(z):=\sum_{n\in\integ} n a_n z^{n-1},
\]
we note that $\Res_z \partial d(z)=0$ for any distribution $d(z)$. Hence, the integration by parts formula holds, provided that $a(z)b(z)$ is defined:
\begin{equation}\label{eq:integration_by_parts_under_Res}
 \Res_z \partial a(z) b(z)= -\Res_z a(z) \partial b(z).
\end{equation}

We will also need the \textbf{formal delta function} $\delta(z-w)$ which is the formal distribution in $z$ and $w$ with values in $\complex$
\begin{equation}\label{eq:formal delta, expansion}
	\delta(z-w) = \sum_{n\in\integ} z^{n-1} w^{-n}.
\end{equation}
Given a rational function $f(z,w)$ with poles only at $z=0$, $w=0$ or $|z|=|w|$, we denote by $\iota_{z,w} f$ $\left(\iota_{w,z}f \right)$ the \textbf{power series expansion} of $f$ in the domain $|z|>|w|$ ($|w|>|z|$).

E.g. for $j\in\nato$ we have
\begin{subequations}\label{eq:z-w expansions}
	\begin{align}
	\iota_{z,w} \frac{1}{(z-w)^{j+1}} &=\sum_{m=0}^{\infty} {m\choose{j}} z^{-m-1}w^{m-j},\\
	\iota_{w,z} \frac{1}{(z-w)^{j+1}} &=-\sum_{m=-1}^{-\infty} {m\choose{j}} z^{-m-1}w^{m-j}
	\end{align}
\end{subequations}
and it follows that
\begin{subequations}
\begin{align}
	D^j_w \delta(z-w) &= \iota_{z,w} \frac{1}{(z-w)^{j+1}}-\iota_{w,z} \frac{1}{(z-w)^{j+1}} \label{eq:delta derivatives}\\
					&=\sum_{m\in\integ} {m\choose{j}} z^{-m-1} w^{m-j},\label{eq:delta derivativesB}
\end{align}
\end{subequations}
with
\[
D^j_w:=\frac{\partial^j_w}{j!}.
\]

The next proposition justifies the name formal delta function.
\begin{prop}\label{prop:properties_of_formal_delta_function}
	 We have for all formal distributions $f(z)\in U[[z,z^{-1}]]$:
	\begin{enumerate}[label=\bfseries \textnormal{ \hspace{1em}(\alph*)}] \itemsep 0.5em
		\item $f(z)\delta(z-w)$ is well-defined,
		\item $f(z)\delta(z-w) = f(w) \delta(z-w)$,
		\item $\Res_z f(z) \delta(z-w) = f(w)$,
		\item $\delta(z-w) = \delta(w-z)$,
		\item $\partial_z \delta(z-w) = -\partial_w \delta(z-w)$,
		\item $(z-w) D^{j+1}_w \delta(z-w) = D^j_w \delta(z-w)$ with $j\in \nato$,
		\item $(z-w)^{j+1} D^j_w \delta(z-w) = 0$ if $j\in\nato$.
	\end{enumerate}
	Here $D_w^0 = 1$ is understood.
\end{prop}
\begin{proof}
	Note that
	\[
		\delta(z-w)=\sum_{k+n+1=0} z^k w^n = \delta(w-z)
	\]
	and
	\[
		\delta(z-w)=\sum_{n,k\in\integ} \delta_{k,-n-1} z^k w^n\in U[[z^{\pm},w^{\pm}]].
	\]
	Thus, the product $f(z)\delta(z-w)$ is well-defined. Moreover, for $f(z)=\sum_{k\in\integ} f_k z^k$ we have
	\[
	f(z)\delta(z-w) = \sum_{n,k\in\integ} f_k z^{k-n-1} w^n = \sum_{k\in\integ} \left(\sum_{n\in\integ} f_{k+n+1} w^n\right) z^k\implies
	\]
	\[
	\implies\Res_z f(z)\delta(z-w) =f(w).
	\]
	Furthermore,
	\[
	f(w)\delta(z-w)= \sum_{n,k\in\integ} f_k w^k z^{-n-1} w^n=\sum_{n,k\in\integ} f_k  z^{k-n-1} w^{n} = f(z)\delta(w-z)
	\]
	by the above.
	This proves parts (a)-(d). Part (e) follows from the definition of the formal delta function \eqref{eq:formal delta, expansion} by direct calculation. To prove (f), we use \cref{eq:delta derivatives}
	\begin{gather*}
	(z-w) D^{j+1}_w \delta(z-w) = (z-w) \sum_{m\in\integ} {m\choose{j+1}} z^{-m-1} w^{m-j-1} = \\
	\sum_{m\in\integ} {m\choose{j+1}} z^{-m} w^{m-j-1} - \sum_{m\in\integ} {m\choose{j+1}} z^{-m-1} w^{m-j} =\\
	 \sum_{m'\in\integ} {{m'+1}\choose{j+1}} z^{-m'-1}w^{m'-j} - \sum_{m\in\integ} {m\choose{j+1}} z^{-m-1}w^{m-j}=\\
	 \sum_{m\in\integ} \left( {{m+1}\choose{j+1}}-{m\choose{j+1}}\right)z^{-m-1}w^{m-j}=\\
	 \sum_{m\in\integ}{{m}\choose{j}}z^{-m-1}w^{m-j}=D^j_w \delta(z-w).
	\end{gather*}
	Part (g) follows by repeated application of (f) and by symmetry property (b):
	\[
	(z-w)^{j+1} D^j_w \delta(z-w) = (z-w) \delta(z-w)= z\delta(z-w)  - w\delta(z-w)=0.
	\]
\end{proof}

The next proposition will be useful for OPEs.
\begin{prop}\label{prop:local distribution expansion}
	If $a(z,w)\in U[[z^{\pm},w^{\pm}]]$ is such that $(z-w)^N a(z,w)=0$ for some $N\in\nat$, then it can be uniquely written as
	\begin{equation}\label{eq:nullspace_element_expansion}
		a(z,w)=\sum\limits_{j=0}^{N-1}c^j(w)D^j_w\delta(z-w),
	\end{equation}
	with
	\begin{equation}\label{eq:nullspace_element_coefficient}
		c^j(w)= \Res_z(z-w)^j a(z,w).
	\end{equation}
\end{prop}

\begin{proof}
	We have 
	\[
	(z-w)^N\sum_{j=0}^{N-1}c^j(w)D^j_w\delta(z-w)=0
	\]
	by \cref{prop:properties_of_formal_delta_function} (g).
	
	We prove the converse by induction. For $N=1$ we have
	\begin{align*}
	0&=(z-w)a(z,w)=\sum\limits_{m,n\in\integ}a_{m,n} z^{m+1} w^n-\sum\limits_{m,n\in\integ}a_{m,n}z^m w^{n+1}=\\
	&=\sum\limits_{m,n\in\integ}(a_{m,n+1}-a_{m+1,n})z^{m+1}w^{n+1}.	
	\end{align*}
	Thus, $a_{m,n+1}=a_{m+1,n}$ $\forall n,m\in\integ$. Hence, $a_{0,n+1}=a_{1,n}=a_{k,n-k+1}$ $\forall m,k\in\integ$. This implies
	\begin{align*}
		a(z,w)&= \sum\limits_{n,k\in\integ}a_{k,n-k+1}z^k w^{n-k+1}=\sum\limits_{n\in\integ}a_{-1,n+2}w^{n+2} \sum\limits_{k\in\integ}z^k w^{-k-1}=\\
		&=c^0(w)\delta(z-w)
	\end{align*}
	with $c^0(w)=\sum_{n\in\integ}a_{-1,n}w^n$ as required.
	
	Now let $a(z,w)$ be such that
	\[
	 0=(z-w)^{N+1}a(z,w)=(z-w)^N\left((z-w)a(z,w)\right).
	\]
	By induction hypothesis
	\[
		(z-w)a(z,w)=\sum\limits_{j=0}^{N-1}d^j(w)D^j_w\delta(z-w)
	\]
	thus applying $\partial_z$ gives
	\begin{align}\label{eq:proving_nullspace_expansion_uniqueness1}
	a(z,w)+(z-w)\partial_z a(z,w)&= \sum\limits_{j=0}^{N-1}d^j (w)D^j_w \partial_z \delta(z-w)=\nonumber\\
	&=-\sum\limits_{j=0}^{N-1}d^j(w)(j+1)D^{j+1}_w\delta(z-w).
	\end{align}
	Here we have used $\partial_z\delta(z-w)=-\partial_w \delta(z-w)$ from \cref{prop:properties_of_formal_delta_function}. Application of the induction hypothesis to
	\[
	0=\partial_z\left((z-w)^{N+1}a(z,w)\right)=(z-w)^N((N+1)a(z,w)+(z-w)\partial_z a(z,w))
	\]
	yields
	\begin{equation}\label{eq:proving_nullspace_expansion_uniqueness2}
		(N+1)a(z,w)+(z-w)\partial_z a(z,w)=\sum\limits_{j=0}^{N-1}e^j(w)D^j_w\delta(z-w).
	\end{equation}
	Subtracting \cref{eq:proving_nullspace_expansion_uniqueness1} from \cref{eq:proving_nullspace_expansion_uniqueness2} gives
	\[
	Na(z,w)=\sum\limits_{j=0}^{N-1}e^j(w)D^j_w\delta(z-w)+\sum\limits_{j=1}^{N}j\,d^{j-1}(w)D^j_w\delta(z-w)
	\]
	which implies that
	\[
	a(z,w)=\sum\limits_{j=0}^{N}c^j(w)D^j_w\delta(z-w)
	\]
	for suitable $c^j(w)\in U[[w^{\pm}]]$ as required.
	
	We now prove the formula for $c^j(w)$ using \cref{prop:properties_of_formal_delta_function}. From part (g) we see that
	\[
	\Res_z((z-w)^n c^j(w)D^j_w\delta(z-w))=0
	\]
	if $j<n$. If $j=n$, then by (f), (b) and (c)
	\[
	\Res_z((z-w)^n c^j(w)D^j_w\delta(z-w))=c^n(w).
	\]
	Finally, if $j>n$, then (e) and integration by parts \eqref{eq:integration_by_parts_under_Res} gives
	\[
	\Res_z((z-w)^n c^j(w)D^j_w\delta(z-w))=\Res_z\left((z-w)^n c^j(w) (-1)^{j} D^{j}_z\delta(z-w)\right)=0.
	\]
	Thus, the coefficient equation \eqref{eq:nullspace_element_coefficient} holds and therefore the expansion \eqref{eq:nullspace_element_expansion} is unique.
\end{proof}

\begin{rmk}\label{rmk:local distribution expansion coeffs}
			Note that \eqref{eq:nullspace_element_expansion} is equivalent to
		\begin{equation}\label{eq:localDistributionExpansionCoeffs}
		 a_{(m,n)}=\sum\limits_{j=0}^{N-1}{m\choose{j}} c^j_{(m+n-j)}\,,
		\end{equation}
		as follows from \eqref{eq:delta derivativesB} by comparing coefficients.
\end{rmk}


\section{Locality and Normal Ordering}\label{sec:formal distribution locality and normal ordering}
Let the vector space $U$ over $\complex$ be also associative. On $U$ one naturally has the \textbf{commutator} $[a,b]=ab-ba$. The most important example for us of $U$ is $\End V$ of a vector space $V$. 

{\begin{defn}[Locality]\label{def:mutually local distributions}
		Two formal distributions $a(z),b(z)\in U[[z^{\pm}]]$ are \textit{(mutually) local} if 
		\[ (z-w)^N[a(z),b(w)] = 0\quad\text{for}\quad N\gg0. \]
		Here $N\gg 0$ means that there exits $n\in\nato$ such that $\forall N\ge n$ the statement holds.
	\end{defn}
	
	\begin{rmk}\label{rmk:locality of derivatives}
		Differentiating $ (z-w)^N[a(z),b(w)]=0 $ and multiplying by $ (z-w) $ gives $ (z-w)^{N+1}[\partial a(z), b(w)]=0 $. Hence, if $a$ and $b$ are mutually local, $\partial a$ and $b$ are mutually local as well.
	\end{rmk}
	
	Our next goal is to formulate some equivalent definitions of locality. However, we need some notation first. Instead of $a(z)=\sum_{m\in\integ}a_m z^m$ we will often write $a(z)=\sum_{n\in\integ} a_{(n)} z^{-n-1}$. This makes it easy to calculate the coefficients:
	\[
		a_{(n)}=a_{-n-1} = \Res_z \left(a(z) z^n\right).
	\]

	We break $a(z)$ into
	\[
		a(z)_-:=\sum_{n\ge0} a_{(n)} z^{-n-1},\quad a(z)_+:=\sum_{n<0} a_{(n)}z^{-n-1}.
	\]
	Note that the above decomposition is the only way to break $a(z)$ into a sum of ``positive'' and ``negative'' parts such that
	\begin{equation}\label{eq:derivative and positive negative parts of a formal distribution}
	\left(\partial a(z)\right)_{\pm} = \partial \left(a(z)_{\pm}\right).
	\end{equation}
	\begin{defn}\label{def:normally ordered product for distributions}
		The \textit{normally ordered product} of two formal distributions $a(z),b(z)\in U[[z^{\pm}]]$ is the distribution
		\[
		\normord{a(z)b(w)} = a(z)_+ b(w)+b(w) a(z)_-.
		\]
	\end{defn}
Note that the definition implies
	\begin{subequations}\label{eq:distributionExpansionInCommutatorAndNormalOrder}
	\begin{alignat}{2}
		a(z)b(w) &= &[a(z)_-,b(w)]+\normord{a(z)b(w)}\; ,\label{eq:distributionExpansionInCommutatorAndNormalOrderA}\\
		b(w)a(z) &= -&[a(z)_+,b(w)]+\normord{a(z)b(w)}.\label{eq:distributionExpansionInCommutatorAndNormalOrderB}
	\end{alignat}
	\end{subequations}

	\begin{thm}[Equivalent definitions of locality]\label{thm:equivalent definitions of locality for distributions}
		Let $a(z),b(z)\in U[[z^{\pm}]]$ and $N\in\nat$. The following are equivalent:
		\begin{enumerate}
			\item $a(z)$ and $b(z)$ are mutually local with $(z-w)^N [a(z),b(w)]=0$,
			\item $[a(z),b(w)]=\sum\limits_{j=0}^{N-1} c^j(w)D^j_w\delta(z-w)$, where $c^j(w)\in U[[w^{\pm}]]$,
			\item $\;\;\,\displaystyle [a(z)_-,b(w)]=\sum_{j=0}^{N-1}\left( \iota_{z,w}\frac{1}{(z-w)^{j+1}}\right) c^j(w)$,\\
					$\displaystyle -[a(z)_+,b(w)]=\sum_{j=0}^{N-1}\left( \iota_{w,z}\frac{1}{(z-w)^{j+1}}\right) c^j(w)$,
			\item $\displaystyle a(z)b(w)=\sum\limits_{j=0}^{N-1}\left(\iota_{z,w} \frac{1}{(z-w)^{j+1}}\right) c^j(w)+ \normord{a(z)b(w)}\,$,\\
			$\displaystyle b(w)a(z)=\sum\limits_{j=0}^{N-1}\left(\iota_{w,z} \frac{1}{(z-w)^{j+1}}\right) c^j(w)+ \normord{a(z)b(w)}\,$,\\
			where $c^j(w)\in U[[w^{\pm}]]$,
			\item $\displaystyle [a_{(m)}, b_{(n)}]=\sum\limits_{j=0}^{N-1}{m\choose{j}} c^j_{(m+n-j)},\quad m,n\in\integ,$
			\item $\displaystyle [a_{(m)}, b(w)] = \sum\limits_{j=0}^{N-1}{m\choose{j}} c^j(w) w^{m-j},\quad m\in\integ$.
		\end{enumerate}
	\end{thm}

	\begin{proof}
		We have
		\[
		\textnormal{(a)}\iff\textnormal{(b)}\iff\textnormal{(c)}\iff\textnormal{(d)}
		\]
		by \cref{prop:local distribution expansion}, taking all terms in (b) with negative (resp. non-negative) powers of $z$ and using \eqref{eq:delta derivatives}, and equations \eqref{eq:distributionExpansionInCommutatorAndNormalOrder} respectively. Finally, (e) and (f) are equivalent to (b) by \cref{rmk:local distribution expansion coeffs}.
	\end{proof}

\begin{rmk}\label{rmk:abuse of notation for OPE}
	Abusing our notation of \cref{thm:equivalent definitions of locality for distributions}\,(d) gives:

\begin{subequations}
\begin{equation}\label{eq:OPE with normal ordering for distributions}
	a(z)b(w)=\sum\limits_{j=0}^{N-1}\frac{c^j(w)}{(z-w)^{j+1}}+\normord{a(z)b(w)}.
\end{equation}
\text{Often we will simplify even more and write just the singular part}
\begin{equation}\label{eq:OPE for distributions}
	a(z)b(w)\sim \sum\limits_{j=0}^{N-1}\frac{c^j(w)}{(z-w)^{j+1}}.
\end{equation}
\text{Such notation is very common in physics. The condition $|z|>|w|$ is implicit.}
\end{subequations}

\Cref{eq:OPE with normal ordering for distributions,eq:OPE for distributions} are called the \textbf{operator product expansion (OPE)}. By \cref{thm:equivalent definitions of locality for distributions} we can calculate all brackets between all coefficients of mutually local formal distributions $a(z)$ and $b(z)$ using only the singular part of the OPE. Hence, the importance of OPE. Moreover, defining the \textbf{n-th product} ($n\in\nato$) on the space of formal distributions to be
\begin{equation}\label{eq:n-th product of formal distributions}
a(w)_{(n)} b(w)=  \Res_z \left([a(z),b(w)](z-w)^n\right)
\end{equation}
and combining this with \cref{prop:local distribution expansion} and \cref{thm:equivalent definitions of locality for distributions}\,(d) for two mutually local distributions gives
\begin{subequations}
\begin{equation}\label{eq:OPE commutator formal distribution explicit coefficients}
[a(z),b(w)]=\sum\limits_{j=0}^{N-1}\left(a(w)_{(j)}b(w)\right)D^j_w\delta(z-w)
\end{equation}
which by \cref{thm:equivalent definitions of locality for distributions} and abuse of notation is equivalent to
\begin{equation}\label{eq:OPE formal distribution explicit coefficients}
a(z)b(w)=\sum\limits_{j=0}^{N-1}\frac{a(w)_{(j)}b(w)}{(z-w)^{j+1}}+\normord{a(z)b(w)}.
\end{equation}
\end{subequations}
So we have equivalent formulations of OPE.
\end{rmk}}

{ 
We now consider other notions inspired by physics.
\begin{defn}[Hamiltonian, conformal weight]\label{def:Hamiltonian on formal distributions}
	A diagonalizable derivation of the associative algebra $U$ will be called \textit{Hamiltonian} and denoted $H$. Its action on the space of formal distributions with values in $U$ will be given coefficient-wise.
	
	We say that a formal distribution $a=a(z,w,\ldots)$ with values in $U$ is an \textit{eigendistribution for} $H$ of \textit{conformal weight} $h\in\complex$ if
	\[
		\left(H-h-z\partial_z-w\partial_w-\ldots\right)a=0.
	\]
\end{defn}

The following proposition can be proved by straightforward computations.
\begin{prop}
	Given to eigendistributions $a$ and $b$ with conformal weights $h$ and $h'$ respectively, we have
	\begin{itemize}
		\item $\partial_z a$ is an eigendistribution of conformal weight $h+1$,
		\item $\normord{a(z)b(w)}$ is an eigendistribution of conformal weight $h+h'$,
		\item the n-th OPE coefficient of $[a(z),b(w)]$ is an eigendistribution of conformal weight $h+h'-n-1$ with $n\in\nat$, 
		\item if $f$ is a homogeneous function of degree $j$, then $fa$ is an eigendistribution of conformal weight $h-j$.
	\end{itemize}
\end{prop}
\begin{cor}
	The summands of an OPE
	\[
		a(z) b(w) \sim \sum\limits_{j=0}^{N-1}\frac{c^j(w)}{(z-w)^{j+1}},
	\]
	where $a(z)$ and $b(z)$ are two mutually local eigendistributions of conformal weights $h$ and $h'$, have the same conformal weight $h+h'$.
\end{cor}

It is convenient to write
\[
	a(z) = \sum\limits_{n\in -h+\integ}a_n z^{-n-h}
\]
for eigendistributions of conformal weight $h$. In this case, the condition for $a(z)$ to be an eigendistribution of conformal weight $h$ is equivalent to
\[
	[H ,a_n] = -n a_n.
\]}
\begin{ex}[Virasoro formal distribution with central charge C]\label{ex:Virasoro formal distribution}
	Let $V$ be a vector space and consider a representation of Virasoro algebra $\Vir$ on it, such that $L_n\in\End V$ and $C= c \id_V$ with $c\in\complex$. Then
	\[
		L(z)=\sum\limits_{n\in\integ}L_n z^{-n-2}
	\]
	is a formal distribution with coefficients in $\End V$. 
	We compute
	\begin{align*}
		[L(z),L(w)]&=\sum\limits_{m,n\in\integ}[L_m,L_n]z^{-m-2}w^{-n-2}\\
		&=\sum\limits_{m,n\in\integ}(m-n)L_{m+n} z^{-m-2}w^{-n-2} + \sum\limits_{m\in\integ} \frac{m}{12}(m^2-1)z^{-m-2}w^{m-2} C.
	\end{align*}
	Substituting $k=m+n$ and then $j=m+1$ gives
	\begin{flalign*}
		&\sum_{m,n}(m-n)L_{m+n}z^{-m-2}w^{-n-2}=\sum_{k,m}(2m-k)L_k z^{-m-2} w^{-k+m-2} =&\\
		&\quad=\sum_{k,j}(2j-k-2)L_kz^{-j-1}w^{-k+j-3}=&\\
		&\quad=2\sum_{k,j}L_kw^{-k-2}jz^{-j-1}w^{j-1}+\sum_{k,j}(-k-2)L_kw^{-k-3}z^{-j-1}w^j=&\\
		&\quad=2L(w)\partial_w\delta(z-w)+\partial_w L(w)\delta(z-w).
	\end{flalign*}
	For the remaining term we get by substituting $m=n-1$
	\begin{flalign*}
		&\sum_{m\in\integ}\frac{m}{12}(m^2-1)z^{-m-2}w^{m-2}C=&\\
		&\quad\quad\quad\quad\quad\quad\quad\quad\quad=\frac{C}{12}\sum_{m\in\integ}n(n-1)(n-2)z^{-n-1}w^{n-3}=\frac{C}{12}\partial^3_w\delta(z-w).&
	\end{flalign*}
	Thus,
	\begin{equation}\label{eq:Virasoro formal distribution commutator}
		\left[L(z),L(w)\right]=\frac{C}{2}D^3_w \delta(z-w)+2L(w)D_w\delta(z-w)+\partial_w L(w)\delta(z-w)
	\end{equation}
	or equivalently using \cref{thm:equivalent definitions of locality for distributions} and \cref{rmk:abuse of notation for OPE}
	\begin{equation}\label{eq:Virasoro formal distribution OPE}
	 L(z)L(w)\sim \frac{C/2}{(z-w)^4}+\frac{2L(w)}{(z-w)^2}+\frac{\partial_wL(w)}{(z-w)}.
	\end{equation}
	
\end{ex}

Note that $L(z)$ is basically the formal distribution version of the energy-momentum tensor which is usually written $T(z)$ in CFT. But $T$ denotes the infinitesimal translation operator in vertex algebras, so that's why we write $L(z)$ instead.

\section{Fields and Dong's Lemma}
Throughout this section, let $V$ be a vector space.
\begin{defn}[Field in formal distributions]\label{def:field in vertex algebra} 
	A formal distribution $a(z)=\sum a_{(n)}z^{-n-1}\in\End V [[z^{\pm}]] $ is called a \textit{field} if $\forall v\in V $
	\[ a_{(n)}(v)=0\quad\text{for}\quad n\gg0. \]
The collection of fields on a vector space $V$ will be denoted $\mathscr{F}(V)$.
\end{defn}
The definition means that $a(z)v$ is a formal Laurent series in $z$ (i.e. $a(z)v\in V[[z]][z^{-1}]\,$).

For fields normal ordering can be extended to coinciding points.
\begin{defn}[Normally ordered product]
	Given two fields $a(z)$ and $b(z)$ we define
	\begin{equation}\label{eq:normal order of fields}
		\normord{a(z)b(z)}=a(z)_+b(z)+b(z)a(z)_-.
	\end{equation}
\end{defn}
From
\[
\normord{a(z)b(z)}_{(n)}= \sum\limits_{j=-1}^{-\infty}a_{(j)}b_{(n-j-1)} + \sum\limits_{j=0}^{\infty}b_{(n-j-1)}a_{(j)}.
\]
it follows that upon application on $v\in V$ each of the two sums gives only a finite number of non-zero summands. Thus, $\normord{a(z)b(z)}$ is a well-defined formal distribution. Note that the assumption that both $a(z)$ and $b(z)$ are fields was necessary. That is why we were only able to define normally ordered product of general formal distributions in two variables in \cref{def:normally ordered product for distributions}. Furthermore, from \eqref{eq:normal order of fields} it is clear that $\normord{a(z)b(z)}$ is a field, since for all $v\in V$ $b(z)v$ is a formal Laurent series in $z$, hence $a(z)_+b(z)v$ is a formal Laurent series in $z$. Similarly for the other summand. Therefore, the space of fields forms an algebra with respect to the normally ordered product (which in general is not associative).

Another useful property is that the derivative $\partial a(z)$ of a field $a(z)$ is a field and due to \eqref{eq:derivative and positive negative parts of a formal distribution} $\partial$ is a derivation of the normally ordered product
\[
	\partial \normord{a(z)b(z)} = \normord{\partial a(z)b(z)}+\normord{a(z)\partial b(z)}.
\]

The existence of normally ordered product allows us to define the $n$-th product between the fields $\forall n\in\integ$.
\begin{defn}[n-th product of fields]\label{def:n-th product of fields for all n}
	We define the \textit{$n$-th product} for $n\in\integ$ as
\[
a(w)_{(n)}b(w)= \begin{cases}
\Res_z \left([a(z),b(w)](z-w)^n\right)\quad\text{if}\quad n\ge0\\
\normord{D^{(-n-1)}a(w)b(w)}\quad\quad\quad\;\quad\text{if}\quad n<0.
\end{cases}
\]
\end{defn}

The $n$-th product of fields can be written in a single formula. 
\begin{lem}
	For all $n$-th products of fields we have
	\begin{equation}\label{eq:n-th product for fields} 
	a(w)_{(n)}b(w)= \Res_z \left( a(z)\, b(w)\, \iota_{z,w} (z-w)^n - b(w)\, a(z)\, \iota_{w,z} (z-w)^n\right),
	\end{equation}
	where $n\in\integ$.
\end{lem}

\begin{proof} 
	For $n\ge 0$, \cref{eq:n-th product for fields} obviously coincides with \eqref{eq:n-th product of formal distributions}. For $n<0$, the lemma follows from the general Cauchy formulas for any formal distribution $a(z)$ and $k\in\nato$
	\begin{subequations}
		\begin{align}
		\Res_z a(z) \iota_{z,w} \frac{1}{(z-w)^{k+1}} &= +D^k a(w)_+,\\
		\Res_z a(z) \iota_{w,z} \frac{1}{(z-w)^{k+1}} &= -D^k a(w)_-.
		\end{align}
	\end{subequations}
	A straightforward use of definitions proves the $k=0$ case. Differentiating the $k=0$ case $k$ times by $w$ gives the required result. 
\end{proof}

\begin{prop}\label{prop:del is a derivation on n-th products}
	For all fields $a(w)$, $b(w)$ and $\forall n\in\integ$ holds
	\begin{subequations}
		\begin{align}
		\partial a (w)_{(n)}b(w) &= - n a(w)_{(n-1)}b(w),\label{eq:da_n b}\\
		a(w)_{(n)}\partial b(w) &=  +n a(w)_{(n-1)} b(w)+ \partial \left(a(w)_{(n)} b(w)\right). \label{eq:a_n db} 
		\end{align}
	\end{subequations} 
	Hence, $\partial$ is a derivation on all $n$-th products.
\end{prop}

\begin{proof}
	We will only prove the $n<0$ case of the formula \eqref{eq:da_n b}. The other proofs are similar.
	
	Given $n<0$, first of all set $n=-j-1$. Then $j\in\nato$, and using equations \eqref{eq:n-th product for fields} and \eqref{eq:z-w expansions} together with the standard properties of binomial coefficients we have
	\begin{flalign*}
	&\partial a(w)_{(n)} b(w) =&\\
	&\Res_z \Biggl( \sum\limits_{k\in\integ}  k a_k z^{k-1}  b(w) \sum\limits_{m=0}^{\infty}{m\choose{j}} z^{-m-1} w^{m-j}+b(w) \sum\limits_{k\in\integ} k a_k z^{k-1} \sum\limits_{m=-1}^{-\infty} {m\choose{j}} \Biggr)=&\\
	&\sum\limits_{m=0}^{\infty} {m\choose{j}} (m+1) a_{m+1} b(w) w^{m-j}+b(w) \sum\limits_{m=-2}^{-\infty}{m\choose{j}}(m+1) a_{m+1} w^{m-j}=&\\
	&\sum\limits_{m=1}^{\infty} {{m-1}\choose{j}} m a_{m} b(w) w^{m-j-1}+b(w) \sum\limits_{m=-1}^{-\infty}{{m-1}\choose{j}}m a_{m} w^{m-j-1}.&
	\end{flalign*}
	
On the other hand,

	\begin{flalign*} 
	&-n a(w)_{(n-1)} b(w) =&\\
	&\quad\quad (j+1) \Res_z \Biggl( \sum_{k\in\integ} a_k z^k b(w) \sum\limits_{m=0}^{\infty}{m\choose{j+1}} z^{-m-1} w^{m-j-1} +&\\
	& \quad\quad\quad b(w) \sum_{k\in\integ} a_k z^k \sum\limits_{m=-1}^{-\infty}{m\choose{j+1}} z^{-m-1} w^{m-j-1} \Biggr)=&\\
	&\quad\quad\quad(j+1) \Biggl( \sum\limits_{m=1}^{\infty} {m\choose{j+1}} a_m b(w) w^{m-j-1} + b(w)\sum\limits_{m=-1}^{-\infty} {m\choose{j+1}} a_m w^{m-j-1} \Biggr).&
	\end{flalign*}
	Thus, $\partial a(w)_{(n)} b(w) = -n a(w)_{(n-1)}b(w) $ as required. 
\end{proof}

{We now prove two technical lemmas which will be used in the next section.
	
\begin{lem}\label{lem:holomorphic product on the vacuum}
	Let $a(z) = \sum_n a_{(n)} z^{-n-1}$ and $b(z) = \sum_n b_{(n)} z^{-n-1}$ be fields with values in $\End V$ and let $\vac\in V$ be a vector such that
	\[
		a_{(n)}\vac = 0\quad \text{and}\quad b_{(n)}\vac = 0,\quad \forall n\in\nato.
	\]
	Then $\left(a(z)_{(n)} b(z) \right) \vac$ is a $V$-valued formal distribution $\forall n\in\integ$ which does not include any negative powers of $z$ and has a constant term $a_{(n)} b_{(-1)}\vac$.
\end{lem}	

\begin{proof}
	
	Let $k\in\nato$. We consider two cases. Firstly, 
	\begin{align*}
		\left(a(z)_{(-k-1)}b(z)\right)\vac &=\, \normordh{D^k a(z)b(z)}\vac = D^k(a(z))_+\, b(z)\vac\\
		 &=(D^k a(z))_+\, b(z)_+ \vac.
	\end{align*}
	Here we have used \eqref{eq:derivative and positive negative parts of a formal distribution}. Secondly,
	\begin{align*}
			\left(a(z)_{(k)} b(z)\right) \vac &= \sum\limits_{j=0}^{k} {k\choose j} (-z)^{k-j} [a_{(j)}, b(z)]\, \vac\\
			&= \sum\limits_{j=0}^{k} {k\choose j} (-z)^{k-j} a_{(j)} b(z)_+ \vac.
	\end{align*}
	This proves the lemma.
\end{proof}

	\begin{lem}[Dong's Lemma]\label{lem:dong's}
		Given pairwise mutually local fields (resp. formal distributions) $a(z)$, $b(z)$ and $c(z)$, we have that $a(z)_{(n)}b(z)$ and $c(z)$ are mutually local fields (resp. formal distributions) for all $n\in\integ$ (resp. $n\in\nat$). In particular, $\normord{a(z) b(z)}$ and $c(z)$ are mutually local fields if the conditions of the lemma are fulfilled.
	\end{lem}
	
	\begin{proof}
		We will show that for $M\gg 0$
		\begin{equation}\label{eq:AB}
		(z_2-z_3)^M A = (z_2-z_3)^M B,
		\end{equation}
		where
		\begin{subequations}
			\begin{align}
			A &= \iota_{z_1,z_2}(z_1-z_2)^n a(z_1)b(z_2)c(z_3)-\iota_{z_2,z_1} (z_1-z_2)^n b(z_2) a(z_1) c(z_3),\\
			B &= \iota_{z_1,z_2}(z_1-z_2)^n c(z_3)a(z_1)b(z_2)-\iota_{z_2,z_1}(z_1-z_2)^n c(z_3) b(z_2) a(z_1).
			\end{align}
		\end{subequations}
		This suffices since applying $\Res_{z_1}$ to both sides of \cref{eq:AB} and setting $z_2=z$, $z_3=w$ proves the lemma due to \cref{eq:n-th product for fields}. 
		
		Since $a(z)$, $b(z)$ and $c(z)$ are pairwise mutually local, we get for $r\gg 0$
		\begin{subequations}
			\begin{align}
			(z_1-z_2)^r a(z_1) b(z_2) &= (z_1-z_2)^r b(z_2) a(z_1),\label{eq:ab commutator} \\
			(z_2-z_3)^r b(z_2) c(z_3) &= (z_2-z_3)^r c(z_3) b(z_2),\label{eq:bc commutator}\\
			(z_1-z_3)^r a(z_1) c(z_3) &= (z_1-z_3)^r c(z_3) a(z_1).\label{eq:ac commutator}
			\end{align} 
		\end{subequations}
		If we take $r$ sufficiently large, then $n\ge -r$. Pick such an $r\in\nat$. Furthermore, take $M=4r$ and use
		\[ (z_2-z_3)^{3r}=\sum_{s=0}^{3r} {{3r}\choose{s}}(z_2-z_1)^{3r-s}(z_1-z_3)^s  \]
		to write down the left-hand side of \cref{eq:AB} as
		\begin{equation}\label{eq:LHS of eq AB} 
		\sum_{s=0}^{3r}  {{3r}\choose{s}} (z_2-z_1)^{3r-s} (z_1-z_3)^s (z_2-z_3)^r A.
		\end{equation}
		If $3r-s+n\ge r$, then $(z_1-z_2)^{3r-s} \iota_{z_1,z_2} (z_1-z_2)^n=(z_1-z_2)^{r'}$  where $r'\ge r$. Thus, using \eqref{eq:ab commutator} we see that the $s$-th summand in \cref{eq:LHS of eq AB} is $0$ for $0\le s\le r$. Hence, the left-hand of \eqref{eq:AB} becomes
		\begin{equation}\label{eq:last equation of A}
		\sum_{s=r+1}^{3r} {{3r}\choose{s}}(z_2-z_1)^{3r-s}(z_1-z_3)^s (z_2-z_3)^r A.
		\end{equation}
		Analogously, the right-hand side of \eqref{eq:AB} equals
		\begin{equation}\label{eq:last equation of B}
		\sum_{s=r+1}^{3r} {{3r}\choose{s}}(z_2-z_1)^{3r-s}(z_1-z_3)^s (z_2-z_3)^r B.
		\end{equation}
		From the locality assumptions \eqref{eq:bc commutator} and \eqref{eq:ac commutator}, it follows that the equations \eqref{eq:last equation of A} and \eqref{eq:last equation of B} are equal thereby proving the lemma.
	\end{proof}}

\section{Vertex Algebras}\label{sec: vertex}


We are now ready to define one of the central definitions of this work.
\begin{defn}[Vertex Algebra]\label{def:Vertex algebra}
A \textit{vertex algebra} is the following data:

\twocolumns{a vector space $V$}{(the space of states),}

\twocolumns{a vector $| 0\rangle \in V$} {(the vacuum vector),}

\twocolumns{a map $T\in \End V$} {(infinitesimal translation operator),}

\twocolumns{a linear map $Y(\cdot, z): V \to\mathscr{F}(V)$}{(the state-field correspondence)}
		\[
		a\mapsto Y(a,z)=\sum_{n\in\integ} a_{(n)} z^{-n-1}, \quad a_{(n)}\in \End V.
		\]
This data is subject to the following axioms $\forall\, a,b\in V$:
\begin{vaxiom}[Translation covariance]\label{axiom:V translation covariance}
\[[T,Y(a,z)]=\partial Y(a,z)\]
\end{vaxiom}

\begin{vaxiom}[Locality]\label{axiom:V locality}
\[(z-w)^N [Y(a,z),Y(b,w)]=0,\]
\centerline{for some $N \in\nat$ depending on $a$ and $b$.}
\end{vaxiom}

\begin{vaxiom}[Vacuum]\label{axiom:V vacuum}
\[ T|0 \rangle =0, \quad Y(|0\rangle, z)= \id_V,\quad Y(a,z)|0\rangle|_{z=0} = a.\]
\end{vaxiom}

%
%



\end{defn}\vspace{0.5em}

Now we want to prove the Existence \cref{thm:construction of vertex algebras} for vertex algebras which states when a vector space is a vertex algebra. Before that, we need some preliminary results first.

\begin{prop}\label{prop:adT is a derivation on n-th products}
For all fields $a(w)$, $b(w)$ and $\forall n\in\integ$ holds
\[ \ad T \left(a(w)_{(n)} b(w) \right)= \left(\ad T a(w) \right)_{(n)} b(w)+ a(w)_{(n)} \ad T(b(w)), \] 
i.e. $\ad T$ is a derivation on all $n$-th products.
\end{prop}

\begin{proof}
We have using \cref{eq:n-th product for fields} and linearity of $\ad T$
\begin{subequations}
\begin{flalign*} 
&\ad T \left(a(w)_{(n)} b(w) \right)=&\\ 
&\quad=\ad T \left( \Res_z \bm{(}a(z) b(w) \iota_{z,w}(z-w)^n - b(w) a(z) \iota_{w,z} (z-w)^n\bm{)}\right)=&\\
&\quad=\Res_z \left(\ad T\left(a(z) b(w) \iota_{z,w}(z-w)^n - b(w) a(z) \iota_{w,z} (z-w)^n\right)\right)=&\\
&\quad=\Res_z \left( [T, a(z) b(w)] \iota_{z,w}(z-w)^n- [T,b(w)a(z)]\iota_{w,z} (z-w)^n\right).& \numberthis\label{eq:proving adT is a derivation 1}
\end{flalign*}
Moreover,
\begin{flalign*}
&(\ad T a(w) )_{(n)} b(w) =&\\ 
&\quad\quad\Res_z \left([T,a(z)]b(w) \iota_{z,w} (z-w)^n - b(w) [T,a(z)] \iota_{w,z}(z-w)^n\right)\numberthis &\label{eq:proving adT is a derivation 2}
\end{flalign*}
and
\begin{flalign*}
& a(w)_{(n)} \ad T(b(w)) =&\\ 
&\quad\quad\Res_z (a(z) [T,b(w)] \iota_{z,w} (z-w)^n- [T,b(w)]a(z) \iota_{w,z} (z-w)^n).\numberthis &\label{eq:proving adT is a derivation 3}
\end{flalign*}
\end{subequations}
Thus, adding equations \eqref{eq:proving adT is a derivation 2} and \eqref{eq:proving adT is a derivation 3} we obtain \cref{eq:proving adT is a derivation 1} as required.
\end{proof}

{ An analogue of Cauchy problem can be stated and solved for formal series.
\begin{lem}\label{lem:vertex algebra differential equation}
		Let $U$ be a vector space and $S\in\End U$. The initial value problem
		\[\frac{d}{dz}f(z)=Sf(z),\quad f(0)=f_0,\]
		with $f(z)\in U[[z]]$, has a unique solution of the form
		\[f(z)=\sum_{n\in\nato} f_n z^n,\quad f_n\in U. \]
		In fact, $f(z)=e^{zS}f_0=\sum 1/n!\,Sf_n z^n$.
	\end{lem}
	
	\begin{proof}
		The differential equation means $\sum (n+1) f_{n+1} z^n =\sum S f_n z^n$ which implies that $(n+1)f_{n+1}=S f_n$ $\forall n\in\nato$. This is equivalent to $f_n=1/n!\, S^n f_0$.
	\end{proof}}


The following proposition is the first step in the proof that a M\"obius conformal vertex algebra has an action of $\pslinear(2,\complex)$ (\cref{prop:exponentiation in vertex algebra}).
\begin{prop}\label{prop:vertex translation exponentiation and associativity}
\textnormal{(a)} Given a vertex algebra $V$ we have $\forall a\in V$
\begin{align}
Y(a,z)\vac&=e^{zT}(a) \label{eq:Y(a,z) on vacuum is exponential}\\
e^{wT} Y(a,z) e^{-wT} &= Y(a,z+w),\label{eq:conjugation of exponentials is Y translation}\\
e^{wT}Y(a,z)_{\pm}e^{-wT}&=Y(a,z+w)_{\pm}\label{eq:conugation of Y plus minus}.
\end{align}
The last 2 equalities are in $\End V[[z^{\pm}]][[w]]$ which means that $(z+w)^n$ is replaced by its expansion ${\iota}_{z,w}(z+w)^n=\sum_{k\ge0} {{n}\choose{k}} z^{n-k} w^k\in\complex[[z^{\pm}]][[w]]$.\vspace{1mm}\\
\textnormal{(b)} It holds $\forall a,b\in V$ and $\forall n\in\integ$ that
\begin{equation}\label{eq:vertex operator of n product of fields} 
Y(a_{(n)}b,z)\vac = \left(Y(a,z)_{(n)}Y(b,z)\right)\vac.
\end{equation}
\begin{proof}

Let $f(z)=Y(a,z)\vac$ which is in $V[[z]]$ because of the vacuum axiom \cref{axiom:V vacuum}. Using the translation covariance \cref{axiom:V translation covariance} and $T\vac=0$ from \cref{axiom:V vacuum}, we obtain the differential equation $\partial f(z)=T f(z)$. Applying \cref{lem:vertex algebra differential equation} to $U=V$ and $S=T$ gives us $f(z)=e^{zT}a$ proving the first equality.

To prove the second equation, we will apply Lemma \ref{lem:vertex algebra differential equation} to $U=\End V[[z^{\pm}]]$ and $S=\ad T$. First, observe that $\partial_w (e^{wT}Y(a,z)e^{-wT})=[T,e^{wT}Y(a,z)e^{-wT}]=\ad T (e^{wT} Y(a,z) e^{-wT})$. Furthermore, $\partial_w Y(a,z+w)=[T,Y(a,z+w)]$ by translation covariance \cref{axiom:V translation covariance}. Both of these differential equations are of the form $ \partial_w f= (\ad T)(f)$ and have the same initial value $f_0=Y(a,z)\in\End V[[z^{\pm}]]$. Therefore, their solutions are the same by Lemma \ref{lem:vertex algebra differential equation}. This proves the second equation.

Equation \eqref{eq:conugation of Y plus minus} follows from the splitting $[T,Y(a,z)_{\pm}]=\partial Y(a,z)_{\pm}$.

To prove \eqref{eq:vertex operator of n product of fields}, first of all note that both $\partial_z$ and $\ad T$ are derivations of all $n$-th products by \cref{prop:adT is a derivation on n-th products,prop:del is a derivation on n-th products}. Moreover, by vacuum (\cref{axiom:V vacuum}) and translation covariance (\cref{axiom:V translation covariance}) axioms, both sides of \eqref{eq:vertex operator of n product of fields} satisfy the differential equation of Lemma \ref{lem:vertex algebra differential equation}. The initial conditions also coincide by the vacuum axiom \cref{axiom:V vacuum} and \cref{lem:holomorphic product on the vacuum}.
\end{proof}
\end{prop}

\begin{thm}[Uniqueness \cite{Goddard1989}]\label{thm:uniqueness of vertex algebras}
Let $V$ be a vertex algebra and let $B(z)\in \End V[[z^{\pm}]]$ be a field which is mutually local with all the fields $Y(a,z)$, $a\in V$. We have that


\[ \text{if}\quad B(z)\vac=e^{z T} b \;\,\text{for some}\;\, b\in V, \quad\text{then}\quad B(z)=Y(b,z).\]
\end{thm}

\begin{proof}
Locality of $B(z)$ means that
\[ (z-w)^N B(z) Y(a,w)\vac= (z-w)^N Y(a,w)B(z)\vac.\]
Applying to the left-hand side formula \eqref{eq:Y(a,z) on vacuum is exponential} and using the second assumption of the theorem for the right-hand side we obtain
\begin{equation}\label{eq:locality for B(z)}
  (z-w)^N B(z)e^{wT}a=(z-w)^N Y(a,w) e^{zT}b.
\end{equation} 
Using formula \eqref{eq:Y(a,z) on vacuum is exponential} once more for $ e^{zT}b$ we can write the right-hand side as
\[ (z-w)^N Y(a,w) Y(b,z)\vac = (z-w)^N Y(b,z) Y(a,w)\vac \]
where last equality holds for sufficiently large $N$ by locality. Applying formula \eqref{eq:Y(a,z) on vacuum is exponential} yet again to the last equation and equating it with the left-hand side of \eqref{eq:locality for B(z)} we get
\[ (z-w)^N B(z) e^{wT}a= (z-w)^N Y(b,z) e^{wT}a. \]
Setting $w=0$ and dividing by $z^N$ we find that $B(z)a=Y(b,z)a$ $\forall a\in V$. Hence, $B(z)= Y(b,z)$.
\end{proof}

\begin{rmk}\label{rmk:sufficient conditions for uniqueness}
	The assumption $ B(z)\vac= e^{zT}b $ of Theorem \ref{thm:uniqueness of vertex algebras} holds if
	\begin{equation}\label{eq:sufficient conditions for the assumption of uniqueness theorem}
		B(z)\vac|_{z=0}=b\quad\text{and}\quad \partial B(z)\vac=T B(z)\vac
	\end{equation}
	do. This follows from Lemma \ref{lem:vertex algebra differential equation}. Note that only the first condition is not sufficient as the example $B(z)=(1+z) Y(b,z)$ shows.
\end{rmk}

Note that Goddard's Uniqueness Theorem is a vertex algebra analogue of the corollary of Reeh--Schlieder's Theorem \ref{cor:reeh-schlieder}.

\begin{prop}\label{prop:n-th product of vertex algebras}
	Let $V$ be a vertex algebra. One has
	\begin{equation}
		Y(a_{(n)}b,z)=Y(a,z)_{(n)}Y(b,z)
	\end{equation}
	$\forall a,b\in V$,\ and $\forall n \in\integ$.
\end{prop}

\begin{proof}
	Let $B(z)=Y(a,z)_{(n)}Y(b,z)$. By \eqref{eq:vertex operator of n product of fields} and \eqref{eq:Y(a,z) on vacuum is exponential} we have
	\[
	B(z)\vac=Y(a_{(n)}b,z)\vac=e^{zT}(a_{(n)}b).
	\]
	Moreover, by Dong's \cref{lem:dong's}, $B(z)$ is local with respect to all vertex operators $Y(c,z)$. Thus, \cref{thm:uniqueness of vertex algebras} gives the required result.
\end{proof}

\begin{cor}
	\begin{enumerate}
		\item For arbitrary collections of vectors $a^1, \ldots, a^n$ of vertex algebra $V$ and arbitrary collections $k_1,\ldots, k_n$ of positive integers it holds
		\[
			\normord{D^{k_1-1} Y(a^1,z)\ldots D^{k_n-1} Y(a^n,z)} = Y\left(a^1_{(-k_1)} \ldots a^n_{(-k_n)} \vac, z\right).
		\]
		\item We have $\forall a,b \in V$ and $\forall n\in\nat$:
		\[
			\normord{D^n Y(a,z) Y(b,z)} = Y(a_{(-n-1)} b, z).
		\]
		\item It holds $\forall a\in V$
		\begin{equation}\label{eq:T can be taken out}
			Y(Ta,z) = \partial Y(a,z).
		\end{equation}
	\end{enumerate}
\end{cor}

\begin{proof}
	Parts (a) and (b) follow from \cref{prop:n-th product of vertex algebras} by \cref{def:n-th product of fields for all n}. The case (c) follows from (a) by setting $n=1$ and $k_1=2$ and noting that $Ta = a_{(-2)}\vac$.
\end{proof}

We are finally ready to prove the existence theorem.
\begin{thm}[Existence]\label{thm:construction of vertex algebras}
Let $ V $ be a vector space with an endomorphism $T\in \End V$ and a vector $ |0\rangle \in V $. Let $ \left(a^{\alpha}(z)\right)_{\alpha\in I} $ ($ I $ an index set) be a collection of fields such that the following conditions are satisfied $ \forall \alpha,\beta\in I $:

\begin{enumerate}[label=\bfseries \textnormal{ \hspace{1em}(\arabic*)}] \itemsep 0.5em
\item $[T,a^{\alpha}(z)]=\partial a^{\alpha}(z)$,
\item $T\vac=0$ and $a^{\alpha}(z)\vac|_{z=0}=a^{\alpha}$,
\item the linear map $ \sum_{\alpha}\complex a^{\alpha}(z)\to\sum_{\alpha}\complex a^{\alpha} $ defined by $ a^{\alpha} (z)\mapsto a^{\alpha}$ is injective,
\item $a^{\alpha}(z)$ and $a^{\beta}(z)$ are mutually local,
\item the vectors $ a^{\alpha_1}_{(j_1)}\ldots a^{\alpha_n}_{(j_n)}\vac $ with $ j_s\in\integ $, $ \alpha_s\in I $ span $ V $.
\end{enumerate}
Then the definition

\begin{equation}\label{eq:define Y(a,z)} 
 Y\left(a^{\alpha_1}_{(j_1)}\ldots a^{\alpha_n}_{(j_n)}\vac,z\right)= a^{\alpha_1}(z)_{(j_1)}(a^{\alpha_2}(z)_{(j_2)}(\ldots(a^{\alpha_n}(z)_{(j_n)}\id{_V}))
\end{equation}
yields a unique structure of a vertex algebra on $ V $ with the vacuum vector $\vac$, the translation operator $T$  and
\begin{equation}\label{eq:Y(a,z)=a(z)}
 Y(a^{\alpha},z)=a^{\alpha}(z)\quad\forall \alpha\in I.
\end{equation}
\end{thm}

\begin{proof} 
Choose a basis for $ V $ using the vectors of the form (5) and define $Y(a,z)$ by formula \eqref{eq:define Y(a,z)}.

The operators $\ad T$ and $\partial$ are derivations of $n$-th products by propositions \ref{prop:del is a derivation on n-th products} and \ref{prop:adT is a derivation on n-th products}. Furthermore, by assumption (1), $\ad T= \partial$ on the fields $a^{\alpha}(z)$. Thus, $\ad T=\partial$ on all the $Y$'s proving the translation covariance axiom \cref{axiom:V translation covariance}.

The locality axiom \cref{axiom:V locality} holds due to (4), \cref{rmk:locality of derivatives} and Dong's \cref{lem:dong's}.

The first two equations of the vacuum axiom \cref{axiom:V vacuum} are trivially satisfied due to our assumption (2) and the defining equation \eqref{eq:define Y(a,z)}. To prove that 
\[
Y(a^{\alpha},z)\vac|_{z=0}=a^{\alpha},
\]
we first note that by (2) we have $ a^{\alpha}=a^{\alpha}(z)\vac|_{z=0}=a^{\alpha}_{(-1)}\vac $. This implies that $ Y(a^{\alpha},z)\vac|_{z=0}=Y(a^{\alpha}_{(-1)}\vac,z)\vac|_{z=0}=a^{\alpha}(z)_{(-1)}\vac|_{z=0}=a^{\alpha} $ using Equation \eqref{eq:define Y(a,z)}. All the equalities are well-defined because of the injectivity assumption (3).

To prove that our vertex algebra is basis-independent and hence well-defined, note that if we chose another basis out of the monomials (5), we would get a structure of another vertex algebra on $V$ which may differ from our original one. But all the fields of the new structure would be mutually local with respect to those of the old structure and would satisfy \cref{eq:sufficient conditions for the assumption of uniqueness theorem}. Thus, by Remark \ref{rmk:sufficient conditions for uniqueness} and the Uniqueness \cref{thm:uniqueness of vertex algebras}, it follows that these vertex algebra structures would coincide. Therefore, Equation \eqref{eq:define Y(a,z)} is well-defined and Equation \eqref{eq:Y(a,z)=a(z)} holds.


\end{proof}

\begin{defn}[Generating set of fields]\label{def:generating set of fields}
	A \textit{generating set of fields} is a collection of fields of a vertex algebra $V$ satisfying condition (5) of \cref{thm:construction of vertex algebras}. If condition (5) holds restricted to $j_s<0$, then such a collection is called a \textit{strongly generating set of fields}.
\end{defn}

The operator product expansion is usually assumed in 2D CFT in physics. The following theorem shows that it can be deduced from the axioms of a vertex algebra.
\begin{thm}[OPE for vertex algebras]\label{thm:OPE for vertex algebras}
	Let $V$ be a vertex algebra and $a,b\in V$. In the domain $|z|>|w|$ one has
	\begin{subequations}
	\begin{equation}\label{eq:OPE for vertex algebras}
		Y(a,z)Y(b,w)=\nsum\limits_{n=0}^{\infty}\frac{Y(a_{(n)}b,w)}{(z-w)^{n+1}}+\normord{Y(a,z)Y(b,w)}.
	\end{equation}
	Equivalently
	\begin{equation}\label{eq:OPE commutator for vertex algebras}
		[Y(a,z),Y(b,w)]=\sum\limits_{n=0}^{\infty}Y(a_{(n)}b,w)D^n_w\delta(z-w).
	\end{equation}
	\end{subequations}
\end{thm}

\begin{proof}
	Fix $a,b\in V$. Then by the axiom of locality \cref{axiom:V locality} we have
	\[
		(z-w)^{N(a,b)}[Y(a,z),Y(b,w)]=0
	\]
	for some $N(a,b)\in\nat$ depending on $a$ and $b$. Thus, $Y(a,z)$ and $Y(b,z)$ are mutually local formal distributions (\cref{def:mutually local distributions}) and satisfy \cref{eq:OPE formal distribution explicit coefficients} 
	\[
		Y(a,z)Y(b,w)= \bigsum\limits_{j=0}^{N(a,b)-1}\frac{Y(a,w)_{(j)}Y(b,w)}{(z-w)^{j+1}}+\normord{Y(a,z)Y(b,w)}.	
	\]
	Here, as usual, the domain $|z|>|w|$ is implicit. By \cref{prop:n-th product of vertex algebras} we have $Y(a,w)_{(j)} Y(b,w)= Y(a_{(j)}b,w)$. Hence, allowing the sum to go to infinity we obtain
	\[
			Y(a,z)Y(b,w)= \nsum\limits_{j=0}^{\infty}\frac{Y(a,w)_{(j)}Y(b,w)}{(z-w)^{j+1}}+\normord{Y(a,z)Y(b,w)}\quad\quad\forall a,b\in V.
	\]
	Now, \eqref{eq:OPE commutator for vertex algebras} is equivalent to \eqref{eq:OPE for vertex algebras} by the same reasoning which gave us the equivalence between \cref{eq:OPE commutator formal distribution explicit coefficients,eq:OPE formal distribution explicit coefficients} in \Cref{sec:formal distribution locality and normal ordering}.
\end{proof}
We also obtain a useful corollary.
\begin{cor}[Borcherds commutator formulas]\label{cor:Borcherds commutator formulas}
	The vertex algebra commutator OPE \eqref{eq:OPE commutator for vertex algebras} is equivalent to each of the following formulas
	\begin{subequations}
		\begin{align}
			[a_{(m)},b_{(n)}]&=\nsum\limits_{j\ge 0}{m\choose{j}}\left(a_{(j)}b\right)_{(m+n-j)}
			\label{eq:Borcherds coefficient commutator}\\
			[a_{(m)},Y(b,w)]&=\nsum\limits_{j\ge 0}{m\choose{j}}Y(a_{(j)}b,w)w^{m-j}.\label{eq:Borcherds coefficient and field commutator}
		\end{align}
	\end{subequations}
\end{cor}

\begin{proof}
	To prove
	\[
	\eqref{eq:Borcherds coefficient commutator} \implies \eqref{eq:Borcherds coefficient and field commutator}\implies \eqref{eq:OPE commutator for vertex algebras}
	\]
	one has to multiply by the indeterminate with respective power and sum over.
	
	The converse,
	\[
	 \eqref{eq:OPE commutator for vertex algebras}\implies \eqref{eq:Borcherds coefficient and field commutator}\implies\eqref{eq:Borcherds coefficient commutator} 
	\]
	is proved by multiplying with $z^m$ and taking the residue $\Res_z$, and then multiplying by $w^n$ and taking $\Res_w$.
\end{proof}

{ \section{M\"obius Conformal and Conformal Vertex Algebras}\label{sec:Mobius conformal VAs}

Now we add some more structure to vertex algebras which will allow us to define quasiprimary fields.

	\begin{defn}\label{def:graded VA}
		A vertex algebra $V$ is called \textit{graded} if there is a diagonalizable operator $H$ on $V$ such that
		\begin{equation}\label{eq:graded vertex algebra}
		[H,Y(a,z)]=z\partial Y(a,z)+Y(Ha,z).
		\end{equation}
	\end{defn}
	
	\begin{prop}\label{prop:equivalent Hamiltonian}
		A field $Y(a,z)$ of a graded vertex algebra $V$ with diagonalizable operator $H$ has conformal weight $h\in\complex$ with respect to the Hamiltonian $\ad H$ if and only if $Ha=ha$.
	\end{prop}
	
	\begin{proof}
		Use \cref{def:Hamiltonian on formal distributions} together with linearity of $Y(a,z)$ in the first argument.
	\end{proof}
	
	Due to the above proposition, we will abuse our terminology and call $H$ a \textbf{Hamiltonian} of a vertex algebra $V$ if \cref{eq:graded vertex algebra} holds. Moreover, a graded vertex algebra whose Hamiltonian is bounded below by zero will be called a \textbf{positive-energy} vertex algebra.
}

\begin{defn}[M\"obius conformal vertex algebra]
	A vertex algebra $V$ graded by $H$ is called \textit{M\"obius conformal} if there exists an operator $T^*$ on $V$ such that $T^*$ decreases the conformal weight by 1 and
	\begin{equation}\label{eq:Tstar full commutator}
		[T^*, Y(a,z)] = z^2 \partial Y(a,z) + 2z Y(Ha,z) + Y(T^* a,z)
	\end{equation}
	for all $a\in V$. We will also call a field $Y(a,z)$ of a M\"obius conformal vertex algebra with weight $h$ \textit{quasiprimary} if
	\begin{equation}\label{eq:quasiprimary definition of fields}
				[T^*, Y(a,z)] = (z^2 \partial +2 h z) Y(a,z).
	\end{equation}
\end{defn}

Note that $Y(a,z)$ is a quasiprimary field of conformal weight $h$ if and only if 
\begin{equation}\label{eq:quasiprimary definition on vectors}
		Ha = h a,\quad\quad T^* a =0.
\end{equation}
Thus, we will call vectors satisfying \eqref{eq:quasiprimary definition on vectors} \textbf{quasiprimary} (of weight $h$).\\

The following proposition will be very important in \Cref{sec:VAs to Wightman}, where we construct a Wightman CFT from vertex algebras.
\begin{prop}\label{prop:exponentiation in vertex algebra}
	We have 
	\begin{equation}\label{eq:vacuum is quasiprimary}
		H\vac = T^*\vac =0,
	\end{equation}
	i.e. the vacuum vector is quasiprimary of weight $0$. Moreover,
	\begin{equation}\label{eq:sl2 Moebius commutator}
	[H,T] = T,\quad [H,T^*] = - T^*,\quad [T^*, T] = 2H,
	\end{equation}
	i.e. $H, T$ and $T^*$ form a representation of $\slLie(2,\complex)$. It also holds that
	\begin{alignat*}{3}
		&\textnormal{(a) }	e^{\lambda T}\, Y(a,z)\, e^{-\lambda T} &&=Y(a,z+\lambda),\quad\quad &|\lambda| <|z|,\\
		&\textnormal{(b) }\lambda^H\, Y(a,z)\, \lambda^{-H}       &&=Y(\lambda^H a, \lambda z),&\\
		&\textnormal{(c) }e^{\lambda T^*} Y(a,z) e^{-\lambda T^*} &&= Y\left( e^{\lambda (1-\lambda z)T^*} (1-\lambda z)^{-2H} a, \frac{z}{1-\lambda z}\right),\quad &|\lambda z| < 1.
	\end{alignat*}
\end{prop}

\begin{proof}
	 Write for a field of conformal weight $h$
	 \[
	 Y(a,z) = \sum_{n\in -h +\integ}a_n z^{-n-h},
	 \]
	 i.e. shift the coefficients so that
	 \[
	 a_{(n)} = a_{n-h+1}
	 \]
	 holds. Then \eqref{eq:graded vertex algebra} is equivalent to
	 \begin{equation}\label{eq:H on components}
	 [H,a_n] = -n a_n.
	 \end{equation}
	 Similarly, $ [T,Y(a,z)] = \partial Y(a,z) $ is equivalent to
	 \[
		 [T,a_n] = (-n -h +1 ) a_{n-1}.
	 \]
	 \cref{eq:graded vertex algebra} implies that $ H\vac =0 $ and hence
	 \[
		 [H,T] =T
	 \]
	 since both sides give the same result applied on $a_n$'s and annihilate $\vac$.
	 
	 The other commutation relations follow similarly by noting that \eqref{eq:Tstar full commutator} with $a = \vac$ gives $T^*\vac =0 $, and that \cref{eq:Tstar full commutator} in component form for $a\in V$ of conformal weight $h$ is
	 \[
		 [T^*, a_n] = -(n-h+1) a_{n+1} + (T^* a)_{n+1}. 
	 \]
	 
	Part (a) has been already proved in \cref{prop:vertex translation exponentiation and associativity} and is restated here for convenience.
	
	Integrating \eqref{eq:H on components} we get 
	\begin{equation}\label{eq:integrated H on components}
	\lambda^H a_n \lambda^{-H} = \lambda^{-n} a_n
	\end{equation}
	which is equivalent to (b).
	
	Now we prove (c). Write
	\[
		e^{\lambda T^*} Y(a,z) e^{-\lambda T^*} = Y\left( A(\lambda) a, \frac{z}{1-\lambda z}\right),
	\]
	with $A(\lambda)$ a formal power series in $\lambda$ with coefficients in $\Hom\left(V, \End V[[z, z^{-1}]]\right)$. Differentiate both sides by $\lambda$ and use \eqref{eq:Tstar full commutator} to get:
	\[
		\frac{d A(\lambda)}{d\lambda} = z^2 \partial_z A(\lambda) + 2z A(\lambda) H +A(\lambda)T^*.
	\]
	By \cref{lem:vertex algebra differential equation}, this equation has a unique solution. To check that $A(\lambda)=e^{\lambda(1-\lambda z)T^*} (1-\lambda z)^{-2H}$ solves this equation, use \eqref{eq:sl2 Moebius commutator} and that $\ad T^*$ is a derivation.
\end{proof}

\begin{rmk}\label{rmk:exponentiation to SL(2,C)}
By \cref{prop:exponentiation in vertex algebra}, we can identify $T,T^*,H$ with the corresponding $\slLie(2,\complex)$ generators by
\[
T= \begin{pmatrix*}
	0 & 1  \\
	0 & 0  
	\end{pmatrix*},\quad
T^*= \begin{pmatrix*}[r]
	0	&	0\\
	-1	&	0
	\end{pmatrix*},\quad
H= \begin{pmatrix*}[c]
	1/2 & 0\\
	 0  & -1/2
	\end{pmatrix*}.
\]
It is a well-know fact that the exponentiation of $\slLie(2,\complex)$ is not surjective with the usual argument being that there are no elements in $\slLie(2,\complex)$ which under the exponentiation are mapped to the elements of $\slinear_2(\complex)$ whose trace is less than or equal to $-2$. However, in $\pslinear(2,\complex)$ we can choose a representative of positive trace for each equivalence class. Thus, 
\[
	\exp: \slLie(2,\complex) \to \pslinear(2,\complex)
\]
is onto. From \cref{prop:exponentiation in vertex algebra} it follows that $\pslinear_2(\complex)$ acts on the variable $z$ by
\[
	z\mapsto \frac{az + b}{cz + d}\,,\quad\quad \begin{pmatrix}
	a & b\\
	c & d
	\end{pmatrix}\in\slinear_2(\complex)
\]
and that
\[
	e^{\lambda T} = \begin{pmatrix}
						1 & \lambda \\
						0 & 1
					\end{pmatrix},\quad
	e^{\lambda T^*} = \begin{pmatrix}
							\quad1 & 0\\
							-\lambda & 1
					   \end{pmatrix},\quad
	e^{\lambda H} = \begin{pmatrix}
						e^{\lambda/2} & 0\\
						0			  & e^{-\lambda/2}
					\end{pmatrix}.
\]
\end{rmk}
\vspace{5mm}
Adding an action of the Virasoro algebra gives a conformal vertex algebra in which primary fields can be defined.
\begin{defn}[Conformal vertex algebra]\label{def:conformal vertex algebra}
	A \textit{Virasoro field with central charge c} is a field $L(z)=\sum_{n\in\integ} L_n z^{-n-2}\in\End U[[z^{\pm}]]$, $U$ some vector space,  with the OPE
	\begin{equation}\label{eq:defining Virasoro field}
	L(z)L(w)\sim \frac{C/2}{(z-w)^4}+\frac{2L(w)}{(z-w)^2}+\frac{\partial_wL(w)}{(z-w)}
	\end{equation}
	such that $C= c\id_U$ with $c\in\complex$.
	
	A \textit{conformal vector} of a vertex algebra $V$ is a vector $\nu$ such that the corresponding field $Y(\nu,z)=\sum_{n\in\integ}\nu_{(n)}z^{-n-1}=\sum_{n\in\integ} L^{\nu}_n z^{-n-2}$ is a Virasoro field with central charge $c$ satisfying
	\begin{enumerate}
		\item $L^{\nu}_{-1}=T$,
		\item $L^{\nu}_0$ is diagonalizable on $V$.
	\end{enumerate}
	The number $c$ is called the \textit{central charge} of $\nu$.
	
	A \textit{conformal vertex algebra} (of rank $c$) is a vertex algebra having a conformal vector $\nu$ (with central charge $c$). Then the field $Y(\nu,z)$ is called an \textit{energy-momentum field} of the vertex algebra $V$.
\end{defn}

Note that each conformal vertex algebra is M\"obius conformal with 
\[
T=L_{-1},\quad H=L_0,\quad T^*=L_1.
\]
 Indeed by \cref{eq:OPE for vertex algebras} we have $\forall a\in V$
	\begin{equation}\label{eq:very general expansion with energy-mom tensor}
		Y(\nu,z) Y(a,w) \sim \bigsum\limits_{n\ge -1} \frac{Y(L_n a,w)}{(z-w)^{n+2}}
	\end{equation}
	which by \cref{cor:Borcherds commutator formulas} is equivalent to
	\[
		[L_m, Y(a,z)] = \bigsum\limits_{j\ge -1} {{m+1}\choose{j+1}} Y(L_ja,z) z^{m-j}.
	\]
	Setting $m=0$ gives \cref{eq:graded vertex algebra} and setting $m=1$ gives \cref{eq:Tstar full commutator}. The calculation
	\[
		L_0 (L_1 a) = [L_0,L_1]a +L_1 L_0 a= -L_1 a + L_1 h a = (h-1) L_1 a
	\]
	shows that $L_1$ decreases the conformal weight by $1$, as required.

If $L_0 = h a$, then \cref{eq:very general expansion with energy-mom tensor} becomes
\[
	Y(\nu,z) Y(a,w) \sim \frac{\partial Y(a,w)}{z-w} + \frac{h Y(a,w)}{(z-w)^2}+\dots
\]
where we have used $Y(Ta,z) = \partial Y(a,z)$ (\cref{eq:T can be taken out}).

\begin{defn}[Primary field]
A field $Y(a,z)$ of a conformal vertex algebra $V$ is \textit{primary} of conformal weight $h$ if
\[
	Y(\nu,z) Y(a,w) \sim \frac{\partial Y(a,w)}{z-w}+ \frac{h Y(a,w)}{(z-w)^2}.
\]
\end{defn}

All equivalent definitions of primary fields used by physicists also hold in conformal vertex algebras.
\begin{prop}\label{prop:equivalent definitions of primaries}
The following are equivalent:
\begin{enumerate}
	\item $Y(a,z)$ is primary of conformal weight $h$\label{item:primary def},
	\item $L_n a = \delta_{n,0}\, h a\quad\forall n\in\nat$\label{item:primary weights},
	\item $[L_m, Y(a,z)] = z^m(z\partial + h(m+1)) Y(a,z)\quad\forall m\in\integ$\label{item:primary comm with field},
	\item $[L_m,a_n] = ((h-1)m - n ) a_{m+n}\quad\forall m,n\in\integ$\label{item:primary with coeffs}.
\end{enumerate}
\end{prop}

\begin{proof}
	\cref{eq:very general expansion with energy-mom tensor} together with the definition, gives equivalence of \ref{item:primary def} and \ref{item:primary weights}. 
	
	By \cref{thm:OPE for vertex algebras}, the OPE of a primary field is equivalent to
	\begin{align*}
		[Y(\nu,z), Y(a,w)] &= \partial Y(a,w) \delta(z-w) + h Y(a,w) \partial_w \delta(z-w)\\
		&=	\sum\limits_{m\in\integ} (-m-1) a_{(m)} w^{-m-2}\sum\limits_{n\in\integ} z^{-n-1} w^n +\\
		&\quad + h \sum\limits_{m\in\integ} a_{(m)} w^{-m-1} \sum_{n\in\integ} n z^{-n-1} w^{n-1}\\
		&=\sum_{m\in\integ} \sum_{n\in\integ} (-m-1+h(n+1))a_{(m)} w^{n-m-1} z^{-n-2}.
	\end{align*}
	But we also have
	\[
		[Y(\nu,z),Y(a,w)]=\sum [L_n, Y(a,w)]z^{-n-2}
	\]
	and so
	\begin{align}
		[L_m,Y(a,z)] &= \sum\limits_{n\in\integ}  (-n-1+h(m+1)) a_{(n)} z^{m-n-1}\label{eq:Lm with Y(a,z)}\\
		&= z^{m+1} \sum\limits_{n\in\integ} (-n-1) a_{(n)} z^{-n-2} + z^m h(m+1) \sum\limits_{n\in\integ} a_{(n)} z^{-n-1}\nonumber\\
		&= z^{m+1} \partial Y(a,z) + z^m h(m+1) Y(a,z).\nonumber
	\end{align}
	Thus, a primary field $Y(a,z)$ satisfies \ref{item:primary comm with field} and the converse is also true since the reasoning above can be reversed.
	
	Now \ref{item:primary with coeffs} is equivalent to \ref{item:primary comm with field} since by definition and \eqref{eq:Lm with Y(a,z)}
	\begin{align*}
	 [L_m, Y(a,z)] &= \sum\limits_{n\in\integ} [L_m, a_{(n)}]z^{-n-1}\\
	 	&= \sum\limits_{n\in\integ}(-m-n-1+h(m+1)) a_{(m+n)} z^{-n-1}.
	\end{align*}
	Thus, remembering that $a_{(n)} = a_{n-h+1}$ and comparing the coefficients gives the required result.
\end{proof}
Due to \cref{prop:equivalent definitions of primaries}, a \textbf{primary vector} is defined as a vector satisfying $L_0 a = ha$ and $L_n a =0$ for $n\ge 1$.

\begin{ex}\label{ex:vacuum is primary nu quasiprimary}
The vacuum vector $\vac$ is primary, since by \cref{axiom:V vacuum} we have $Y(\nu,z)\vac|_{z=0}=\nu$. This also shows that $\nu = L_{-2}\vac$ and hence $\nu$ is quasiprimary of conformal weight $2$, but not primary unless $c=0$ by \cref{eq:defining Virasoro field}.
\end{ex}

\section{Vertex Operator Algebras and Unitarity}\label{sec:VOAs and unitarity}

Even more assumptions on vertex algebras are usually natural in physics. The following assumptions will allow us to obtain a transparent construction of Wightman 2D CFT in \Cref{sec:VAs to Wightman}.
\begin{defn}[Vertex operator algebra]\label{def:VOA}
	A \textit{vertex operator algebra} (VOA) is a conformal vertex algebra such that
	\begin{enumerate}[label=\textnormal{(\roman*)}]
		\item $V=\bigoplus_{n\in\integ} V_n$, where $V_n:=\ker(L_0-n \id_V)$,
		\item $V_n =\{0\}$ for $n$ sufficiently small,\label{axiom:VOA sufficiently small}
		\item $\dim V_n <\infty$.
	\end{enumerate}
	The subspaces $V_n$ providing the grading are called \textit{homogeneous subspaces}. If $V_n=0$ for $n<0$ and $V_0 = \complex \vac$, then the vertex operator algebra is of \textit{CFT type}.
\end{defn}

\begin{rmk}
If we were to replace the conformal vertex algebra with a M\"obius conformal vertex algebra in \cref{def:VOA}, then we would get a \textbf{quasi-vertex operator algebra} (q-VOA) \cite[Sec. 2.8]{Frenkel1993}.
\end{rmk}

Throughout this section, fix $V$ to be a (q-)VOA, and let $F=\integ$ for a VOA and $F=\{-1,0,1\}$ for a q-VOA which is not a VOA.

The following proposition is due to Roitman \cite{Roitman2004}.
\begin{prop}\label{prop:Roitman}
	Let $M = \bigoplus_{i\in\integ} M_i$ be a graded module over the Lie algebra $\slLie_2 = \Bbbk T+ \Bbbk H+ \Bbbk T^*$, where $\deg T = 1$, $\deg T^* = -1$, $T^*$ is locally nilpotent, $H|_{M_d} = d$ and $\Bbbk$ is a field of characteristic $0$. Furthermore, assume that the $\slLie_2$ commutation relations \eqref{eq:sl2 Moebius commutator} hold. Then $M_d = (T^*)^{1-d} M_1$ for all $d<0$.
\end{prop}

\begin{cor}
By \cref{prop:Roitman} we have that for $n<0$, $V_n= L_1^{1-n} V_1 \subset L_1^{-n} V_0$. Hence, if $V_0= \complex\vac$, then condition (ii) is equivalent to the condition $V_n = \{0\}$ for all $n<0$. Thus, if $V_0 = \complex\vac$, then $V$ is of CFT type.
\end{cor}

We say that a map $\phi:V\to V$ is an \textbf{antilinear automorphism} of a q-VOA $V$, if it is an antilinear isomorphism such that $\phi(u_n v) = \phi(u)_n \phi(v)$ $\,\forall u,v\in V$, $\forall n\in\integ$, and $\phi(\vac) = \vac$. If $V$ is a VOA, then we require $\phi$ to also satisfy $\phi(\nu)=\nu$.

Let $(\cdot,\cdot)$ be a bilinear form on $V$. If $(\cdot,\cdot)$ satisfies
\[
	(Y(a,z)b,c) = (b,Y\bm{(}e^{zL_1}(-z^{-2})^{L_0}a,z^{-1}\bm{)}c)\quad \forall a,b,c\in V,
\]
then it will be called an \textbf{invariant bilinear form}. By \cite[Prop. 2.6]{Li1994}, any invariant bilinear form on a (q-)VOA is in fact symmetric.

\begin{rmk}\label{rmk:invariant bilinear form orthogonal weight spaces}
By direct calculation it follows that $(\cdot,\cdot)$ is invariant if and only if
\begin{equation}\label{eq:coefficient product in bilinear form}
	(a_n b,c) = (-1)^{h_a} \sum_{k\in\nato} \frac{1}{k!}(b,(L_1^k a)_{-n}c)
\end{equation}
for all $b,c\in V$ and for all homogeneous $a\in V$. If $V$ is a VOA, then for $L_n$'s this boils down to
\begin{equation}\label{eq:L_n's hopping in invariant bilinear form}
	(L_n a,b)= (a, L_{-n} b)\quad a,b \in V,\; n\in \integ
\end{equation}
and thus the case $n=0$ shows that $(V_i, V_j) = 0$ if $i \ne j$.

Note that for general q-VOAs the implication from \eqref{eq:coefficient product in bilinear form} to \eqref{eq:L_n's hopping in invariant bilinear form} does not work because there is no conformal vector $\nu$. So in the case of a q-VOA which is not a VOA, $(L_1 a, b) = (a, L_{-1} b)$ has to be assumed along with \eqref{eq:coefficient product in bilinear form} for $(V_i, V_j) = 0$ if $i \ne j$ to hold \cite{Roitman2004}.
\end{rmk}

Now let $(\cdot|\cdot)$ be an inner product on $V$, linear in the second variable. We also want it to be \textbf{normalized}, i.e. $(\Omega| \Omega)=1$, and $(\cdot|\cdot)$ \textbf{invariant}, i.e. there exists a VOA antilinear automorphism $\theta$ of $V$ such that $(\theta \cdot |\cdot)$ is an invariant bilinear form on $V$. We call $\theta$ a \textbf{PCT} operator associated with $(\cdot|\cdot)$. By definition $\theta(\nu)=\nu$, so $\theta$ commutes with all $L_n$'s. \Cref{eq:coefficient product in bilinear form} implies that
\begin{equation}\label{eq:coefficient product in inner product with theta}
	(a_n b|c) = \Big(\theta\left((\theta^{-1}a)_n\theta^{-1}b\right)\Big|c\Big) = (b|(\theta^{-1}e^{L_1}(-1)^{L_0}a)_{-n}c)
\end{equation}
$\forall a,b,c\in V$ and $\forall n\in\integ$. If $a$ is quasiprimary, then
\[
	(a_n b|c) = (-1)^{h_a} (b| (\theta^{-1} a)_{-n} c),
\]
$\forall b,c\in V$ and $\forall n\in\integ$. In particular,
\begin{equation}\label{eq:L_n in unitary VOA}
	(L_n a| b) = (a|L_{-n} b)
\end{equation}
$\forall a,b\in V$ and $\forall n\in \integ$. Thus, the corresponding representations of the Virasoro algebra and its $\slLie(2,\complex)$ subalgebra $\complex\{L_{-1}, L_0, L_1 \}$ are unitary and therefore completely reducible. In particular, we have $V_n = 0$ for $n<0$ by \cref{prop:Roitman}. Similarly, for a q-VOA.

{\begin{prop} 
Let $V$ be a q-VOA with a normalized invariant inner product $(\cdot|\cdot)$. Then there exists a unique PCT operator $\theta$ associated with $(\cdot|\cdot)$. Furthermore, $\theta$ is an antiunitary involution.
\end{prop}

\begin{proof}
	Let $\tilde{\theta}$ be another PCT operator associated with $(\cdot|\cdot)$. \Cref{eq:coefficient product in inner product with theta} shows that $(\theta^{-1} e^{L_1} (-1)^{L_0} a)_n = (\tilde{\theta}^{-1}e^{L_1}(-1)^{L_0}a)_n$ 
	and hence $\theta^{-1} e^{L_1}(-1)^{L_0} a= \tilde{\theta}^{-1} e^{L_1} (-1)^{L_0} a$ for all $a\in V$. Now the surjectivity of $e^{L_1} (-1)^{L_0}$ implies that $\theta= \tilde{\theta}$. 
	
	From \Cref{eq:coefficient product in inner product with theta} and symmetry of the bilinear form, it follows that $a= (e^{L_1} (-1)^{L_0})^2 \theta^{-2} a$ for all $a\in V$. \Cref{eq:integrated H on components} shows that \[(-1)^{L_0} e^{L_1} (-1)^{L_0} = e^{-L_1}\] and hence $(e^{L_1} (-1)^{L_0})^2 = 1$. Thus, $\theta^2 =1$, i.e. $\theta$ is an involution. Moreover, we see that $(\theta a|\theta b) = (\theta^2 b|a)= (b|a)$ $\forall a,b\in V$ by the symmetry of the invariant bilinear form $(\theta\cdot|\cdot)$. Therefore, $\theta$ is antiunitary.
\end{proof}

This leads to the following definition.
\begin{defn}\label{def:unitary VOA} 
A \textit{unitary (quasi-)vertex operator algebra} is a pair $(V,(\cdot|\cdot))$ where $V$ is a (quasi-)vertex operator algebra and $(\cdot|\cdot)$ is a normalized invariant inner product on $V$.
\end{defn}

\begin{rmk} 
Note that the requirement that $\dim V_n <\infty$ was not used until now. Thus, if we have a M\"obius conformal vertex algebra with integer grading, whose Hamiltonian is bounded below, we can also define a \textbf{unitary M\"obius conformal vertex algebra} paralleling \cref{def:unitary VOA}.
\end{rmk}

Many of the well-known VOAs have unitary examples: Virasoro (see \Cref{sec:Virasoro to VA}), affine, Heisenberg and lattice VOAs. Moreover, the moonshine VOA $V^{\natural}$ is also unitary. For proofs see \cite{Dong2014}.\\

The definition of unitarity does not seem to have much in common with the notion of unitarity used in QFT. However, \cite[Sec. 5.2]{Carpi2015} shows that these two notions are equivalent for VOAs of CFT type, but the proof works for q-VOAs with $V_0 = \complex\vac$ as well. In particular, we have \cite[Thm. 5.16]{Carpi2015}:

\begin{thm}\label{thm:equivalence of unitarity for VOA}
	Let $V$ be a (q-)VOA with a normalized inner product $(\cdot|\cdot)$ and $V_0 =\complex\vac$. Then the following are equivalent:
	\begin{enumerate}
		\item $(V,(\cdot|\cdot))$ is a unitary (q-)VOA,
		\item $(V,(\cdot|\cdot))$ has a unitary M\"obius symmetry and every vertex operator has a local adjoint.
	\end{enumerate}
\end{thm}
We now give the definitions and some results of the notions used in the above theorem. More details and the proof can be found in \cite[Sec. 5.2]{Carpi2015}.

For a (q-)VOA with a normalized inner product $(\cdot|\cdot)$ to have \textbf{unitary M\"obius symmetry} means that $\forall a,b\in V$
\[
	(L_n a|b) = (a|L_{-n}b),\quad n=-1,0,1.
\]
For an operator $A\in\End V$ to have an adjoint on $V$ (with respect to $(\cdot|\cdot)$) means that $\exists A^+\in\End V$ such that
\[
	(a|Ab)= (A^+ a|b),\quad \forall a,b\in V.
\]
If $A^+$ exists, then it is unique and called the \textbf{adjoint} of $A$ on $V$.
\begin{rmk}\label{rmk:local adjoint if and only if hilbert adjoint}
	Let $\hilb$ be the Hilbert space completion of $(V,(\cdot|\cdot))$. Then an operator $A\in\End V$ can be considered as a densely defined operator on $\hilb$. Thus, $A^+$ exists if and only if the domain of the Hilbert space adjoint $A^*$ of $A$ contains $V$ and if this is the case we have $A^+\subset A^*$, i.e. $A^+ = A^*|_V$.
\end{rmk}

\begin{lem}
Let $(V,(\cdot|\cdot))$ have unitary M\"obius symmetry. Then the adjoint $a_n^+$ of $a_n$ on $V$ exists $\forall a\in V$ and $\forall n\in\integ$. Moreover, we have $a^+_{-n} b = 0$ for $n\gg 0$.
\end{lem}
\begin{proof}
	The finite-dimensional subspaces $V_n = \ker (L_0- n \id_V)$ of $V$ are pairwise orthogonal by unitary M\"obius symmetry. Since $a_n(V_k)\subset V_{k-n}$, the operator $a_n|_{V_k}$ can be regarded as an operator between two finite-dimensional inner product spaces and thus it has an adjoint $(a_n|_{V_k})^*\in \Hom(V_{k-n},V_k)$ which is well-defined. It follows that
	\[
		a_n^+ := \bigoplus\limits_{k\in\integ} (a_n|_{V_k})^*
	\]
	is the adjoint of $a_n$. This shows that $a^+_{-n}(V_k)\subset V_{k-n}$ and hence $a_{-n}^+ b =0$ for $n\gg 0$.
\end{proof}
The lemma implies that for $a\in V$ the formal series
\[
	Y(a,z)^+ := \sum\limits_{n\in\integ} a^+_{(n)} z^{n+1} = \sum\limits_{n\in\integ}a^+_{(-n-2)}z^{-n-1}
\]
is well-defined and is a field on $V$. So we say that a vertex operator $Y(a,z)$, $a\in V$, has a \textbf{local adjoint} if $\forall b\in V$ the fields $Y(a,z)^+$ and $Y(b,z)$ are mutually local, i.e.
\[
	(z-w)^N [Y(a,z)^+, Y(b,w)] = 0,\quad N\gg 0.
\] }


	\chapter{Wightman QFT}\label{chap:Wightman}

Formulated in 1950s, Wightman's axioms of QFT are the first attempt at putting QFT on a rigorous mathematical footing. Even though the axioms are very natural, it turned out to be very difficult to construct examples. To date there are no non-trivial examples of Wightman QFTs in 4D. Nevertheless, the CPT and Spin-Statistics theorems can be proved in Wightman framework. Moreover, the statement that ``knowing all the fields is the same as knowing all the correlation functions" is made explicit by Wightman Reconstruction Theorem \ref{thm:Wightman reconstruction} and its converse \ref{thm:Wightman fields to Wightman distributions}. Last but not least, the current work shows that Wightman axioms are sufficiently general to incorporate many aspects of 2D genus 0 CFTs.\\

\begin{flushleft}
The main reference for this chapter is the book by Wightman and Streater \cite{Wightman1964}, but we have also used \cite{Schottenloher2008}, \cite{Kac1998} and \cite{Bogolubov1989}. We only consider bosonic theories, but a generalization to include fermions is easily obtained. See, e.g., \cite{Wightman1964} or \cite{Bogolubov1989}.
\end{flushleft}

\section{Preliminaries}

Before stating the axioms we need some definitions.
\begin{defn}[Schwartz space, tempered distribution]
Let 
\[
	\schwartz(\reals^n)=\left\{f\in C^{\infty}(\reals^n) \mid \lVert f\rVert_{\alpha,\beta}<\infty\quad \forall\alpha,\beta \right\}
\]
 be the \textit{Schwartz space} of rapidly decreasing smooth functions. Here $\alpha,\beta$ are multi-indices and
 \[
 	\lVert f\rVert_{\alpha,\beta}=\sup_{x\in\reals^n}|x^{\alpha} D^{\beta} f(x)|
 \]
 are seminorms.
 The elements of $\schwartz(\reals^n)$ are called \textit{test functions} and the dual space consists of \textit{(tempered) distributions} which are linear functionals $\schwartz(\reals^n)\to\complex$, continuous with respect to all seminorms.
\end{defn}

\begin{defn}\textbf{Wightman field.}\label{def:Wightman field}
Let $\ops(\hilb)$ be the set of all densely defined operators on a Hilbert space $\hilb$. Denote by $\langle\cdot,\cdot\rangle$ the inner product of $\hilb$. A \textit{(Wightman) field} $\phi$ on a manifold $M$ is a tempered operator-valued distribution, i.e. a map $\phi : \schwartz (M) \to \ops(\hilb)$, such that there exists a dense subspace $\domain\subset \hilb$ satisfying
	\begin{itemize}
		\item $\domain\subset\domain_{\phi(f)}\quad\forall f\in\schwartz(M)$,
		\item the induced map $\schwartz\to \End \domain, \, f\mapsto \phi(f)|_{\domain}$ is linear,
		\item $\forall v\in\domain,\, \forall w\in\hilb$ the assignment $f\mapsto \langle w, \phi(f)(v)\rangle$ is a tempered distribution. 
	\end{itemize}
\end{defn}


\begin{flushleft}\textbf{Minkowski space, Lorentz group, Poincar\'{e} group.} \end{flushleft}

Let $M=(\reals^{1,d-1},g)$ be a $d$-dimensional \textbf{Minkowski space}, i.e. the vector space $\reals^{d}$ with metric
\[
|x-y|^2=(x^0-y^0)^2-\sum_{i=1}^{d-1} (x^i-y^i)^2,
\]
so that $g=\diag(1, \underbrace{-1, -1, \dots, -1}_{d-1})$.

Given $A,B\subset M$, we say that $A$ and $B$ are \textbf{spacelike separated} if $\forall a \in A$ and $\forall b \in B$ we have $|a-b|^2<0$. Let the \textbf{forward (light)cone} $\cone$ be the set $\left\{ x\in M \mid  |x|^2 \ge 0,\, x^0 \ge 0 \right\}$. Define \textbf{causal order on} $M$ by $x\ge y \iff x-y\in \cone$. We will also often write
\[
a^{\mu} b_{\mu} = a\cdot b.
\]

\begin{defn}\textbf{Lorentz group, Poincar\'{e} group.}\\
In $d$ dimensions:
\begin{itemize}
	\item $\mathcal{L}:= \orthgp(1,d-1) = \left\{ \Lambda \in \glinear(d) \mid \Lambda\, g\, \Lambda^T = g \right\}$---\textit{full Lorentz group}, preserves the metric;
	\item $\mathcal{L}_+ := \sorthgp(1,d-1) = \left\{\Lambda \in \orthgp(1,d) \mid \det \Lambda =1 \right\}$---\textit{proper Lorentz group}, preserves orientation;
	\item $\mathcal{L}^{\uparrow} := \left \{ \Lambda \in \orthgp(1, d) \mid e\, \Lambda\, e^T \ge 0 \right\}$---\textit{orthochronous Lorentz group}, preserves the direction of time, here  $e = (1, 0, 0, \dots, 0)\in M$;
	\item $\lorentz := \mathcal{L}^{\uparrow} \cap \mathcal{L}_+= \sorthgp^+(1,d-1)$---\textit{(proper orthochronous) Lorentz group}.	
\end{itemize}
The $d$-dimensional \textbf{(proper orthochronous) Poincar\'{e} group} is defined as 
	\[
	\poinc := \reals^d \rtimes \lorentz.
	\]
It is a set of pairs $(q, \Lambda)\in (\reals^d, \lorentz)$ with multiplication:
\[
(q_1, \Lambda_1)\cdot (q_2,\Lambda_2) := (q_1+ \Lambda_1 q_2, \Lambda_1\, \Lambda_2),
\]
and $\poinc$ acts continuously on the test functions $\schwartz(\reals^d)$ from the left by 
\[
(q, \Lambda) f(x) := f (\Lambda^{-1}(x - q)).
\]
\end{defn}

Note that equivalently one can define the Poincar\'{e} group $\poinc$ as the \textbf{identity component} (maximal connected subset containing the identity) of the group of all transformations of $M$ preserving the metric. Similarly, the Lorentz group $\lorentz$ can be defined as the group of all unimodular linear transformations of $M$ preserving the lightcone $\cone$. Therefore, the Poincar\'{e} group preserves the causal order and thus the spacelike separation.

\section{Wightman Axioms}

Now we define a bosonic Wightman QFT for an at most countable collection of scalar fields. See, e.g., \cite[Sec. 8.2]{Bogolubov1989} for generalizations to arbitrary bosonic and fermionic fields.

\begin{defn}[Wightman QFT]
	A \textit{Wightman quantum field theory} in $d$ dimensions is:
	\begin{itemize}
	\item the projective space $\projective(\hilb)$ of a complex Hilbert space $\hilb$ (the space of states),
	\item the vector $\Omega\in\hilb$ such that $\langle\Omega,\Omega\rangle =1$ (the vacuum vector),
	\item a continuous unitary representation $(q,\Lambda)\mapsto U(q,\Lambda)$ of the Poincar\'{e} group $\poinc$,
	\item a collection of fields $\phi_a$ and their adjoints $\phi_a^*$, $a\in I$ with $I$ an at most countable index set,
		\[
		\phi_a : \schwartz(\reals^d) \to \ops(\hilb).
		\]
	\end{itemize}
One requires this data to satisfy the following axioms:
\begin{waxiom}[Covariance]\label{axiom:W covariance}
	It holds
	\begin{equation}\label{eq:Poincare covariance}
	U(q,\Lambda) \phi_a(f) U(q,\Lambda)^{-1} = \phi_a ((q,\Lambda)f),
	\end{equation}
	$\forall f\in \schwartz(\reals^d),\; \forall (q,\Lambda)\in\poinc$.
\end{waxiom}
Note that by Stone's theorem $U(q,1)=\exp \left(i \sum_{k=0}^{d-1}q^k P_k\right)$ with $P_k$ self-adjoint and commuting operators on $\hilb$.
\begin{waxiom}[Stable vacuum and spectrum condition]\label{axiom:W spectrum}
	We have \[
		U(q,\Lambda)\Omega=\Omega,
	\]
	$\forall (q,\Lambda)\in\poinc$.
	The simultaneous spectrum of all the operators $P_0,\ldots,P_{d-1}$ is contained in the forward light cone $\cone$.
\end{waxiom}
\begin{waxiom}[Cyclicity of the vacuum]\label{axiom:W completeness}
	The vacuum $\Omega$ is in the domain of any polynomial in the $\phi_a(f)$'s and their adjoints. 
	Let $\domain_0\subset\hilb$ be the subspace spanned by such polynomials
	\[
	\phi_{a_1}(f_1)\phi_{a_2}(f_2)\ldots\phi_{a_m}(f_m)\Omega
	\]
	and their adjoints.
	We assume that $\domain_0$ is dense in $\hilb$. Clearly, $\Omega\in\domain_0$.
\end{waxiom}
Sometimes a weaker version of \cref{axiom:W completeness} is used.
\begin{waxiomspeciall}[Dense domain]\label{axiom:W dense domain}
	There exists a linear set $\domain$ dense in $\hilb$ such that the domain of each smeared operator $\phi_a(f)$ contains $\domain$. Same holds for adjoints $\phi_a(f)^*$. Moreover, 
	\[
		\Omega\in\domain,\quad U(q,\Lambda)\domain\subset\domain,\quad \phi_a(f)\domain\subset\domain,\quad\phi_a(f)^*\domain\subset\domain.
	\]
\end{waxiomspeciall}

\begin{waxiom}[Locality]\label{axiom:W locality}
 If the supports of $f,g\in\schwartz(\reals^d)$ are spacelike separated, then on the common dense domain
	\[
	[\phi_a(f), \phi_b(g)] = 0.
	\]
	Similarly,
	\[
		[\phi_a(f), \phi_b(g)^*] = 0.
	\]
\end{waxiom}
\end{defn}

\begin{rmk}\label{rmk:abuse of notation for W fields}
By abuse of notation, we will often write $ \phi(x) $ instead of $ \phi(f) $. With this notational simplification, the equivariance condition \eqref{eq:Poincare covariance} becomes
\begin{equation}\label{eq:Poincare covariance abused}
 U(q,\Lambda)\phi_a(x) U(q,\Lambda)^{-1}=\phi_a(\Lambda x+q)
\end{equation}
and the adjoint is simply $\phi^*(x)$ which acts by $\phi^*(f)=\phi\left(\overbar{f}\right)^*$.
\end{rmk}

\begin{rmk}
	Note that by definition, we have
	\[
		\phi_a(f)\domain_0 \subset\domain_0
	\]
	for all the fields. Moreover, we have by \cref{axiom:W covariance} and \cref{axiom:W spectrum}
	\begin{align*}
		&U(q,\Lambda)\phi_{a_1}(f_1)\phi_{a_2}(f_2)\ldots\phi_{a_m}(f_m)\vac=\\	&\quad U(q,\Lambda)\phi_{a_1}(f_1)U(q,\Lambda)^{-1}U(q,\Lambda)\phi_{a_2}(f_2)\ldots \phi_{a_m}(f_m)U(q,\Lambda)^{-1}\vac=\\
		&\quad \phi_{a_1}((q,\Lambda)f_1)\phi_{a_2}((q,\Lambda)f_2)\ldots\phi_{a_m}((q,\Lambda)f_m)\vac,
	\end{align*}
	i.e.
	\[
		U(q,\Lambda)\domain_0 \subset \domain_0.
	\]
	Moreover, locally the translation covariance (\cref{eq:Poincare covariance abused} with $\Lambda=\id$) is
	\begin{equation}\label{eq:translation covariance locally}
		i[P_k,\phi_a(x)]=\partial_{x_k} \phi_a(x).
	\end{equation}
	Thus, using that $P_k\Omega=0$ by \cref{axiom:W spectrum}, we conclude that
	\begin{equation}\label{eq:translation invariance locally}
		P_k\domain_0 \subset \domain_0.
	\end{equation}
\end{rmk}

\begin{rmk}\label{rmk:general Lorentz covariance}
	For more general fields the covariance axiom \cref{axiom:W covariance} is replaced by
	\begin{equation}\label{eq:non-trivial Lorentz representation}
		U(q,\Lambda) \phi_j (f) U(q,\Lambda)^{-1} = \sum_{k=1}^m R_{jk} (\Lambda^{-1}) \phi_k ((q,\Lambda)\cdot f),
	\end{equation}
	where $R:G \to \mathrm{GL} (V)$ is a finite-dimensional representation of the corresponding Lorentz group $\lorentz$ (or its double cover) on $\reals^m$ or $\complex^m$. If the representation is non-trivial, then a field is a collection of $\phi_i$'s transforming into each other under \eqref{eq:non-trivial Lorentz representation}. See \cite[Sect. 8.2]{Bogolubov1989} or \cite{Wightman1964} for more details.
\end{rmk}

The following corollary of the well-known Reeh and Schlieder's Theorem will be useful later. The proof can be found in \cite[Sec. 8.2 D]{Bogolubov1989}.
\begin{cor}[Reeh--Schlieder]\label{cor:reeh-schlieder}
	Let $X=\phi_{a_1}(x_1)\ldots\phi_{a_n}(x_n)$ be some product of Wightman fields. If $X\Omega = 0$, then $X=0$.
\end{cor}

Later on we will need a technical lemma which can be found in \cite[Sec. 2]{Luscher1975} which itself is based on \cite{Glaser1974}.
\begin{lem}\label{lem:analytic Glaser extension}
	In an $n$-dimensional Wightman QFT satisfying \crefrange{axiom:W covariance}{axiom:W locality}, the vector
	\[
			\Psi(x_1,\ldots,x_n):=\phi_{a_1}(x_1)\ldots\phi_{a_n}(x_n)\Omega,
	\]
	where $\phi_{a_i}$ are some Wightman fields, extends to a vector-valued analytic function
	\[
		\Psi(z_1,\ldots,z_n),\quad z_k:=x_k+i y_k,
	\]
	on a connected domain which includes the Euclidean points with $z_k=(i y^0_k, \vec{x}_k)$ such that $y^0_k>0$ for all $k$ and $z_i^0\ne z_j^0$ if $i\ne j$. Here $\vec{x}_k:=(x_k^1,\ldots,x_k^{n-1})$.
\end{lem}
\begin{proof}
	We will prove the lemma for the simplified case of a self-adjoint scalar field $\phi$. The general case follows similarly.
	
	Let
	\[
	\Psi(x_1,\ldots,x_n):=\phi(x_1)\ldots\phi(x_n)\Omega.
	\] 
	Note that $\Psi\in\hilb$ by \cref{axiom:W completeness}. 
	
	By Poincare covariance \eqref{eq:Poincare covariance} from \cref{axiom:W covariance}, we have translation covariance
	\begin{equation}\label{eq:translation covariance}
	U(q,1) \phi(x)U(q,1)^{-1}=\phi(x+q).
	\end{equation}
	Applying both sides of this equation to the vacuum and using the invariance of the vacuum vector $ U(q,\Lambda)\Omega=\Omega $ from \cref{axiom:W spectrum} as well as $ U(q,1) =  \exp  i (q^{\mu}P_{\mu}) $ we obtain
	\[
	\phi_a (x+q)\Omega = e^{iq^{\mu} P_{\mu}}\phi_a(x)\Omega.
	\]
	Thus,
	\begin{align*}
	\Psi(x_1,\ldots,x_n)&=e^{ix^{\mu}_1 P_{\mu}}\phi(0)e^{-ix^{\mu}_1 P_{\mu}}e^{ix^{\mu}_2 P_{\mu}}\ldots \phi(0)\Omega=\\
	&=e^{ix^{\mu}_1 P_{\mu}}\phi(0)e^{i(x^{\mu}_2-x^{\mu}_1) P_{\mu}}\ldots \phi(0)\Omega=\\
	&=\int \mathrm{d}^n p\; \mathrm{d}^n q_1\ldots \mathrm{d}^n q_{n-1} \tilde{\Psi}(p,q_1,\ldots,q_{n-1})e^{i\left(p_{\mu}x^{\mu}_1+\sum (q_{\mu})_j (x^{\mu}_{j+1} -x^{\mu}_j)\right)}.
	\end{align*}
	By spectrum assumption from \cref{axiom:W spectrum}, $\tilde{\Psi}$ is non-zero only if $p^0\ge0$ and all $q_i^0\ge 0$. Thus, $\Psi$ can be analytically continued to a vector-valued analytic function $\Psi\in\hilb$, i.e. we have
	\begin{align}\label{eq:Hilbert vector extension}
		\Psi(z_1,\ldots,z_n)\;\text{of}\; z_k=x_k+i y_k\quad
		&\text{defined and holomorphic for}\nonumber\\
		&y_1\in V_+ \text{ and } y_j-y_i\in V_+ \text{ if } j>i. 
	\end{align}
				Fix $\pi$ to be any permutation of $(1,\ldots,n)$ and let
				\[
				\Psi^{\pi} (z_1,\ldots,z_n):=\Psi(z_{\pi(1)},\ldots,z_{\pi(n)}),\quad z_k=x_k+i y_k.
				\]
				By the above, $\Psi^{\pi}$ is well-defined and holomorphic in a domain containing the Euclidean points with $0<y^0_{\pi(1)}<\ldots<y^0_{\pi(n)}$. Furthermore, by locality \cref{axiom:W locality}
				\[
				\Psi^{\pi}(x_1,\ldots,x_n)=\Psi(x_1,\ldots,x_n)\text{ for real }x_k\text{ such that } (x_i-x_j)^2<0 \;\;\forall i\ne j,
				\]
				i.e. all $\Psi^{\pi}$'s are equal on a real neighborhood. Now the Edge of the Wedge Theorem (see, e.g., \cite{Wightman1964}) 
				shows that they are analytic continuations of one and the same analytic function. Moreover, the domain of analyticity of this function must contain the domains of analyticity of each $\Psi^{\pi}$.
\end{proof}

\section{Wightman Distributions and Reconstruction}
In this section we will show that there exist tempered distributions which provide an equivalent description of Wightman QFT. Here we will only consider scalar fields for simplicity.

Let $\phi_1,\ldots,\phi_n$ be scalar fields of a $d$-dimensional Wightman QFT. The function
\[
	W_n(f_1,\ldots,f_n):=\langle \Omega,\phi_1(f_1)\ldots\phi_n(f_n)\Omega\rangle
\]
is well-defined by \cref{axiom:W dense domain} for $f_1,\ldots,f_n\in\schwartz(\reals^d)$ and is a separately continuous multilinear functional. By the Schwartz Nuclear Theorem \cite[Thm. 2-1]{Wightman1964} this functional can be uniquely extended to a tempered distribution in $\schwartz'((\reals^d)^n)=\schwartz'(\reals^{d\cdot n})$. This distribution will be also denoted $W_n$. Such distribution is called a \textbf{Wightman distribution}, a \textbf{vacuum expectation value} or a \textbf{correlation function}.

\begin{thm}\label{thm:Wightman fields to Wightman distributions}
	Given a $d$-dimensional Wightman QFT satisfying \cref{axiom:W covariance}, \cref{axiom:W spectrum}, \cref{axiom:W dense domain} and \cref{axiom:W locality}, the Wightman distributions 	$W_n\in\schwartz'(\reals^{d\cdot n})$, $n \in\nat$, associated to it have the following properties: 
	\begin{wdaxiom}[Covariance]\label{axiom:WD covariance}
		We have
		\[
			W_n(f)=W_n((q,\Lambda)f)\quad\forall (q,\Lambda)\in\poinc.
		\]
	\end{wdaxiom}
	\begin{wdaxiom}[Spectrum condition]\label{axiom:WD spectrum}
		There exists a distribution $W'_n\in\schwartz'(\reals^{d(n-1)})$ supported in the product $\bar{V}^{n-1}_+\subset\reals^{d(n-1)}$ of forward cones such that
		\[
			W_n(x_1,\ldots,x_n)=\int_{\reals^{d(n-1)}} W'_n(p) e^{i\sum p_j\cdot(x_{j+1}-x_j)} \mathrm{d}p,
		\]
		where $p=(p_1,\ldots,p_{n-1})\in\reals^{d(n-1)}$ and $\mathrm{d}p=\mathrm{d}p_1\ldots \mathrm{d}p_{n-1}$.
	\end{wdaxiom}
	\begin{wdaxiom}[Hermiticity]\label{axiom:WD hermiticity}
				We have
				\[
					\langle \Omega, \phi_1(x_1)\ldots \phi_n(x_n) \Omega\rangle = \overline{\langle\Omega,\phi_n^*(x_n)\ldots\phi_1^*(x_1)\Omega\rangle}.
				\]
	\end{wdaxiom}
	\begin{wdaxiom}[Locality]\label{axiom:WD locality}
		For all $n\in\nat$ and $1\le j\le n-1$
		\[
		W_n(x_1,\ldots,x_j,x_{j+1},\ldots,x_n)=W_n(x_1,\ldots,x_{j+1},x_j,\ldots,x_n)
		\]
		if $(x_j-x_{j+1})^2<0$.
	\end{wdaxiom}
	\begin{wdaxiom}[Positive definiteness]\label{axiom:WD positivity}
			For any sequence $\{f_j\}$ of test functions, $f_j\in\schwartz(\reals^{d\cdot j})$, with $f_j=0$ except for a finite number of $j$'s, it holds that
		\begin{align}\label{eq:Wightman positivity}
			\sum\limits_{j,k=0}^{\infty}\int \overbar{f}_j(x_1,\ldots, x_j) W_{jk} &(x_j,\ldots x_1,y_1,\ldots,y_k)\times\\
			&\times f_k(y_1,\ldots,y_k)\,\mathrm{d}x_1\ldots\mathrm{d}x_j\,\mathrm{d}y_1\ldots\mathrm{d}y_k\ge0.\nonumber
		\end{align}
		Here by $W_{jk}$ we mean
		\[
			\langle\Omega, \phi^*_{jj}(x_j)\ldots\phi^*_{j1}(x_1)\phi_{k1}(y_1)\ldots\phi_{kk}(y_k)\Omega\rangle
		\]
		and $\phi_{jk}$ can be any field of our theory. Furthermore, if \eqref{eq:Wightman positivity} is zero for some $\{f_j\}$, then \eqref{eq:Wightman positivity} is zero for any sequence $\{g_j\}$
		\begin{equation}\label{eq:Wightman positivity 0}
			g_0=0,\quad g_1= g(x_1)f_0,\quad g_2 = g(x_1)f_1(x_2),\quad g_3=g(x_1)f_2(x_2,x_3),\ldots
		\end{equation}
		with $g\in\schwartz(\reals^d)$ arbitrary.
		\end{wdaxiom}
\end{thm}

\begin{proof}
	\Cref{axiom:WD covariance} follows from \cref{axiom:W covariance} and \cref{axiom:WD locality} from \cref{axiom:W locality}.
	
	Hermiticity \cref{axiom:WD hermiticity} follows from
	\[
		\langle\Omega,\phi_1(f_1)\ldots \phi_n(f_n)\Omega = \overline{\langle\Omega,(\phi_n(f_n))^*\ldots(\phi_1(f_1))^*\Omega\rangle}
	\]
	and the fact that this relation extends from test functions of the form\\ $f_1(x_1)\ldots f_n(x_n)$ to the whole of $\schwartz(\reals^{d\cdot m})$ by the Schwartz Nuclear Theorem.
	
	The inequalities \eqref{eq:Wightman positivity} of \cref{axiom:WD positivity} are equivalent to the fact that the norm of the state
	\[
		\Psi = f_0 \Omega + \phi_{11}(f_1)\Omega +\int \phi_{21}(x_1)\phi_{22}(x_2) f_2(x_1,x_2)\,\mathrm{d} x_1\mathrm{d} x_2\Omega+\ldots
	\]
	is non-negative. If the norm is zero, then $\Psi=0$ and hence $\pr_j (g) \Psi = 0$ for any component $j$ of the test function $g$. Thus, \eqref{eq:Wightman positivity 0} holds.
	
	\Cref{axiom:WD spectrum} will be proved in \cref{prop:WD spectrum}.
\end{proof}

By the covariance of the fields,
\[
	W_n(x_1,\ldots,x_n)=W_n(\Lambda x_1 +q,\ldots, \Lambda x_n+q)\quad \forall (q,\Lambda)\in\poinc.
\]
Here and further we abuse our notation for the correlation functions as we often do for the fields (cf. \cref{rmk:abuse of notation for W fields}). It follows that Wightman distributions are translation invariant
\[
	W_n(x_1,\ldots,x_n)=W_n(x_1+q,\ldots,x_n+q).
\]
Thus, the distributions depend only on the differences
\[
	\xi_i:=x_i-x_{i-1}
\]
and we define
\[
	w_n(\xi_1,\ldots,\xi_{n-1}):=W_n(x_1,\ldots,x_n).
\]
\begin{prop}\label{prop:WD spectrum}
	The Fourier transform $\widehat{w}_n$ has its support in the product $(\cone)^{n-1}$ of the forward cones $\cone\subset\reals^d$. Thus,
	\[
		W_n(x)=(2\pi)^{-d(n-1)} \int_{\reals^{d(n-1)}} \widehat{w}_n(p)e^{-i\sum p_j\cdot(x_j-x_{j+1})}dp.
	\]
\end{prop}
\begin{proof}
	 Since $U(x,1)^{-1}=U(-x,1)=e^{-i x^{\mu} P_{\mu}}$ for $x\in\reals^d$, the spectrum condition \cref{axiom:W spectrum}  implies
		\begin{equation}\label{eq:spectrum condition in terms of Fourier integral}
			\int_{\reals^d} e^{i x^{\mu} p_{\mu}} U(x,1)^{-1}\, v\, dx=0\quad\forall v\in\hilb\quad \text{if}\; p\notin \cone.
		\end{equation}
		Note that
		\[
			w_n(\xi_1,\ldots,\xi_j+x, \xi_{j+1},\ldots, \xi_{n-1})=W_n(x_1,\ldots,x_j,x_{j+1}-x,\ldots,x_n-x).
		\]
		Thus, the Fourier transform of $w_n$ with respect to $x$ gives
		\begin{align*}
			\int_{\reals^d} w_n&(\xi_1,\ldots,\xi_j+x,\xi_{j+1},\ldots,\xi_{n-1})e^{i p_j\cdot x} dx=\\
			&=\left\langle \Omega,\phi_1(x_1)\ldots\phi_j(x_j) \int_{\reals^d}\phi_{j+1}(x_{j+1}-x)\ldots\phi_n(x_n-x)e^{ip_j\cdot x}\Omega dx\right\rangle=\\
			&=\left\langle\Omega, \phi_1(x_1)\ldots\phi_j(x_j)\int_{\reals^d}e^{i p_j\cdot x}U(x,1)^{-1}\phi_{j+1}(x_{j+1})\ldots\phi_n(x_n)\Omega dx\right\rangle=0,
		\end{align*}
		if $p_j\notin\cone$ by \eqref{eq:spectrum condition in terms of Fourier integral} with $v = \phi_{j+1}(x_{j+1})\ldots\phi_n(x_n)\Omega$. Therefore,
		\[
			\widehat{w}_n(p_1,\ldots,p_{n-1})=0
		\]
		if $p_j\notin\cone$ for at least one index $j$.
\end{proof}


We state the cluster decomposition property for completeness, but do not give a proof since we will not use it. This property ensures that the Wightman QFT obtained via the Wightman Reconstruction Theorem from the Wightman distributions has a unique vacuum. For a proof with a mass gap see \cite{Wightman1964} and references therein.
\begin{wdaxiom}[Cluster Decomposition Property]\label{axiom:WD cluster decomposition}
	For a space-like vector $q$
	\begin{align*}
		W_n(x_1,\ldots,x_j,x_{j+1}+\lambda q, x_{j+2} + \lambda q, \ldots, &x_n+\lambda q)\to\\ 
			&W_j(x_1,\ldots, x_j) W_{n-j}(x_{j+1}, \ldots, x_n)
	\end{align*}
	as $\lambda\to \infty$ with convergence in $\schwartz'$.
\end{wdaxiom}

The following proof is based on \cite{Schottenloher2008}. For a more explicit proof, which also uses \cref{axiom:WD cluster decomposition} and hence proves the uniqueness up to a unitary transformation of the resulting Wightman QFT, see \cite{Wightman1964}. For simplicity, we provide a proof only for a single self-adjoint scalar field.
\begin{thm}[Wightman Reconstruction Theorem]\label{thm:Wightman reconstruction}
	For a sequence of tempered distributions $(W_n)$, $W_n\in\schwartz'(\reals^{d\cdot n})$, satisfying \crefrange{axiom:WD covariance}{axiom:WD positivity}, there exists a Wightman QFT satisfying \cref{axiom:W covariance}, \cref{axiom:W spectrum}, \cref{axiom:W dense domain} and \cref{axiom:W locality}.
\end{thm}

\begin{proof}
	Let
	\[
		\underline{\schwartz}:=\bigoplus_{n=0}^{\infty}\schwartz(\reals^{d\cdot n})
	\]
	be the vector space of finite sequences $\underline{f}=(f_0,f_1,f_2,\ldots)$, i.e. $f_0\in\complex$, $f_n\in\schwartz(\reals^{d\cdot n})$ and all but finitely many of test functions $f_n$ are zero. We define multiplication on $\underline{\schwartz}$ by
	\begin{align*}
		\underline{f}\times\underline{g}&:=(h_n),\\
		h_n&:= \sum_{i=0}^n f_i(x_1,\ldots,x_i)g_{n-i}(x_{i+1},\ldots,x_n).
	\end{align*}
Note that $\underline{\schwartz}$ forms an associative algebra with unit $\underline{1}=(1,0,0,\ldots)$. We put the direct limit topology on $\underline{\schwartz}$ to make it into a complete separable locally convex space. 
Each continuous linear functional $\mu:\underline{\schwartz}\to\complex$ can be represented by sequences $(\mu_n)$ of tempered distributions $\mu_n\in\schwartz_n': \mu( (f_n) ) = \sum \mu_n (f_n)$. For each functional $\lambda$ of this form which is also positive semi-definite, i.e. $\lambda(\underline{\overline{f}}\times \underline{f}) \ge 0$ for all $\underline{f}\in\underline{\schwartz}$, the subspace
\[
	J= \left\{	\underline{f}\in\underline{\schwartz}:\lambda\left(\underline{\overbar{f}}\times\underline{f}\right)=0	\right\}
\]
is an ideal of the algebra $\underline{\schwartz}$. Then on the quotient $\underline{\schwartz}\,/ J$ the positive semi-definite functional $\lambda$ gives rise to a positive definite Hermitian scalar product by setting $\omega(\underline{f},\underline{g}):=\lambda\left(\underline{\overline{f}}\times\underline{g}\right)$. Thus, completing $\underline{\schwartz}\,/J$ with respect to this scalar product produces a Hilbert space $\hilb$.

Now set $\lambda:=(W_n)$. By \cref{axiom:WD positivity}, the continuous functional $\lambda$ is positive semi-definite and hence provides the Hilbert space $\hilb$ constructed above. For the vacuum vector we set $\Omega:=\iota(\underline{1} )$ where $\iota(\underline{f} )$ denotes an equivalence class from the dense domain $\domain:=\underline{\schwartz}\,/J$. We define the field operator $\phi$ on $\domain$ by
\[ 
	\phi(f)\iota\left(\ubar{g} \right):=\iota\left(\underline{g}\times f \right)
\]
for all $f\in\schwartz$. Here $f$ denotes the sequence $(0,f,0,\ldots)$. For $\underline{g},\underline{h}\in\underline{\schwartz}$ the mapping
\[
f\mapsto \left\langle\iota(\underline{h} ),\phi(f)\iota(\underline{g} )\right\rangle = \lambda\left(\underline{h}\times(\underline{g}\times f)\right)
\]
is a tempered distribution by continuity of $\lambda$. Thus, $\phi$ is indeed a field operator (\cref{def:Wightman field}). Furthermore, $\phi(f) \domain \subset \domain$ and $\Omega\in\domain$.

Now we draw our attention to covariance. First of all, we need to define a unitary representation of the Poincar\'{e} group $\poinc$ on $\hilb$. We start by considering the natural action $\underline{f}\mapsto (q,\Lambda)\underline{f}$ of $\poinc$ on $\underline{\schwartz}$ given term-wise by
\[
	(q,\Lambda)f_k(x_1,\ldots,x_k) := f_k(\Lambda^{-1}(x_1-q),\ldots,\Lambda^{-1}(x_k-q)),
\]
where $(q,\Lambda)\in\reals^d\rtimes\lorentz\cong\poinc$. This leads to a homomorphism $\poinc \to \glinear(\underline{\schwartz})$. By the covariance \ref{axiom:WD covariance}, we have that if $\underline{f}\in J$ and $(q,\Lambda)\in\poinc$, then $(q,\Lambda)\underline{f}\in J$. Thus,
\[
	U(q,\Lambda)\iota(\underline{f} ) := \iota\left((q,\Lambda)\underline{f} \right)
\]
is well-defined on the dense domain $\domain\subset\hilb$ and satisfies
\[
	\left\langle U(q,\Lambda)\iota(\underline{f} ), U(q,\Lambda)\iota(\underline{f} )\right\rangle = \langle \iota(\underline{f} ),\iota(\underline{f} )\rangle.
\]
Therefore, we get a unitary representation of $\poinc$ on $\hilb$ such that $U(q,\Lambda)\Omega = \Omega$ and $U(q,\Lambda)\phi(f)U(q,\Lambda)^{-1} = \phi( (q,\Lambda) f)$ because $U(q,\Lambda)$ respects the multiplication $\times$ of $\underline{\schwartz}$. This proves \cref{axiom:W covariance} and \cref{axiom:W dense domain}.

Now the spectrum axiom \cref{axiom:W spectrum} follows from \cref{axiom:WD spectrum} by noting that
\[
	\left\{(f_n)\,\middle |\,  f_0 = 0, \, \widehat{f}(p_1,\ldots, p_n) =0 \; \text{in a neighborhood of } (\cone)_n\right\} \subset J,  
\] 
where $(\cone)_n = \{ p \mid p_1+\ldots+p_n\in\cone,\; j=1,\ldots, N\}$. Similarly the locality axiom \cref{axiom:W locality} holds by noting that $J$ contains the ideal generated by linear combinations of the form
\[
	f_n(x_1,\ldots,x_n) = g(x_1,\ldots,x_j, x_{j+1}, \ldots, x_n) - g(x_1, \ldots, x_{j+1}, x_j, \ldots, x_n)
\]
with $g(x_1, \ldots, x_n) = 0$ if $(x_{j+1}-x_j)^2\ge 0$.
\end{proof}


\section{Wightman CFT}

To get a conformal Wightman QFT, we extend the symmetry group of our system from the Poincar\'{e} group $\poinc$ to the (restricted) conformal group.

\begin{waxiomspecial}[Conformal covariance]\label[waxiom]{axiom:W conformal covariance}
	The continuous unitary representation of the Poincar\'{e} group extends to a continuous unitary representation of the (restricted) conformal group $ (q,\Lambda,b)\mapsto U(q,\Lambda,b) $ such that
	\begin{equation}\label{eq:conformal invariance of vacuum and D }
		U(q,\Lambda, b)\Omega = \Omega
	\end{equation}
	      $\forall (q,\Lambda, b)\in \confgp(\reals^{1,d-1})$ and conformal covariance holds for some collection of fields of the QFT which we call \textit{quasiprimary}. The other fields are just Poincar\'e covariant.
	      
	      We assume that in 2D a quasiprimary field  $\phi_a$ of \textit{scaling dimension} $\Delta_a$ and \textit{spin} $s_a$ transforms as
	\begin{equation}\label{eq:conformal covariance} 
	U(q,\Lambda, b) \phi_a(f) U(q,\Lambda, b)^{-1} = \vphi_a(b,x)\;\phi_a ((q,\Lambda, b)\cdot f),
	\end{equation}
	with
		\[
			\vphi_a(b,x) = (1+ (b^0+b^1)(x^0-x^1))^{-\Delta_a-s_a}\,(1+(b^0-b^1)(x^0+x^1))^{-\Delta_a+s_a}.
		\]
We also assume that $s,\Delta\in\reals$ for all fields.
\end{waxiomspecial}

Note that \eqref{eq:conformal covariance} is just the transformation law of a scalar field, so we should set $s_a = 0$, but we keep $s_a$ for making the upcoming discussion clearer (cf. \cref{rmk:general Lorentz covariance}). The most general transformation laws can be found in \cite{Mack1969}.

Clearly, stronger covariance of the fields leads to stronger covariance of Wightman distributions and so we call such distributions \textbf{conformally covariant}.

\begin{rmk}
Note that by \eqref{eq:conformal covariance} for a special conformal transformation in 2D it holds
\begin{equation}
U(0,1,b)\phi_a (x) U(0,1,b)^{-1}=\vphi_a(b,x)\,\phi_a(x^b)
\end{equation}
and that from Stone's Theorem it follows that $ U(0,1,b)=\exp i \sum_{n=0}^1 b^n K_n $, where $ K_n $ are self-adjoint and commuting operators on $ \hilb $. Here we let $x^b$ to denote a special conformal transformation with parameter $b$
\[
	x\mapsto \frac{x+|x|^2 b}{1+2 \langle x, b \rangle + |x|^2 |b|^2}. 
\]
 Hence, locally we have
\begin{subequations}\label{eq:special conformal covariance locally}
\begin{align}
	i\left[K_0,\phi_a(x)\right]=\left(|x|^2\partial_{0} - 2 x^0 E - 2\Delta_a x^0+2s_a x^1\right) \phi_a(x)\\
	i\left[K_1,\phi_a(x)\right]=\left(|x|^2\partial_{1} +2 x^1 E + 2\Delta_a x^1-2s_a x^0\right) \phi_a(x)
\end{align}
\end{subequations}
with $E=x^0 \partial_{0}+x^1\partial_1$.
\end{rmk}

Sometimes the axiom \cref{axiom:W conformal covariance} is too strong. To prove the L\"uscher--Mack Theorem, only dilation covariance will suffice. Thus, we state the axiom of dilation covariance here separately.

\begin{waxiomdilation}[Dilation covariance]\label[waxiom]{axiom:W dilation invariance}
	There exists a unitary representation $U'$ of the dilation group such that for $\lambda > 0$ we have
	\[
		U'(\lambda) \Omega = \Omega
	\]
	and
	\[
		U'(\lambda) \phi(x) U'(\lambda)^{-1} = \lambda^{\Delta} \phi(\lambda x)
	\]
	for some fields which we call dilation covariant. Other fields are just Poincar\'e covariant.
	Here $\Delta$ is the scaling dimension of $\phi$.
\end{waxiomdilation}

Another very important axiom usually made in 2D CFT is:

\begin{waxiom}[Existence of energy-momentum tensor]\label{axiom:W energy-momentum tensor}
	In the operator algebra generated by the fields $\{\phi_a\}_{a\in I}$ there is a dilation covariant local field $T_{\mu\nu}(x),\; \mu,\nu\in\{0,1\}$, with the following properties:
	\begin{subequations} 
	\begin{align}
	T_{\mu\nu}= T_{\nu\mu},& \quad T_{\mu\nu}^{*}=T_{\mu\nu},\\ 
	\partial^{\mu} T_{\mu \nu} &= 0,\label{eq:T continuity Wightman} \\
	\Delta(T_{\mu\nu})&=2, \label{eq:T has dim 2}
	\end{align}
	\end{subequations}
	where $\Delta$ is the scaling dimension. Moreover, we assume that the generators $P_{\mu}$ can be expressed in terms of $T_{\mu\nu}$:
	\begin{equation}\label{eq:T generates translations Wightman}
	\int \mathrm{d}x^1 [T_{0 \mu}(x^0,x^1),\phi(y)]=[P_{\mu},\phi(y)]=-i\partial_{\mu}\phi(y).  
	\end{equation}
\end{waxiom}

We are now ready to give one of the central definitions of this work.
\begin{defn}[Wightman (M\"obius) CFT]
A 2D Wightman QFT satisfying \crefrange{axiom:W conformal covariance}{axiom:W locality} is called \textit{Wightman M\"obius CFT}. If Wightman M\"obius CFT contains an energy-momentum tensor, i.e. it also satisfies \cref{axiom:W energy-momentum tensor}, then it is a \textit{Wightman CFT}. 
\end{defn}

	\part{Comparisons}
\chapter{Wightman Axioms and Virasoro Algebra}\label{chap:Wightman and Vir}

The goal of this chapter is to prove that a 2D dilation invariant Wightman QFT with an energy-momentum tensor gives rise to two commuting unitary Virasoro algebras as was first proved by L\"uscher and Mack in \cite{Luescher1976}.\\

\begin{flushleft}
This chapter is based on the original source \cite{Luescher1976}, the talk \cite{Luescher1988} and \cite{Furlan1989}.
\end{flushleft}

\section{L\"uscher--Mack Theorem}\label{subsec:Lusher-Mack}
We will use light-cone coordinates in this section:
\begin{alignat*}{2}
&t = x^0-x^1, 									 &&\partial_t = \frac{1}{2}(\partial_0 - \partial_1),\\
&\bar{t} = x^0 +x^1,\quad\quad\quad 					 &&\partial_{\bar{t}} = \frac{1}{2}(\partial_0 + \partial_1),
\end{alignat*}
so that
\begin{equation}
	\Theta:=T_{tt} = \frac{1}{4}( T_{00} - 2 T_{01} + T_{11} ),\quad\quad \bar{\Theta}:=T_{\bar{t}\bar{t}} = \frac{1}{4}( T_{00} + 2 T_{01} + T_{11} ),\label{eq:T plus plus definition}
\end{equation}
\[
T_{t\bar{t}} = T_{\bar{t}t} = \frac{1}{4}(T_{00} - T_{11}),
\]
where $T_{\mu\nu}$ are components of the energy-momentum tensor defined in \cref{axiom:W energy-momentum tensor}.

\begin{lem}\label{lem:holomorphic T}
	In 2D dilation invariant Wightman QFT with an energy-momentum tensor, i.e. a 2D Wightman QFT satisfying \crefrange{axiom:W dilation invariance}{axiom:W energy-momentum tensor}, it holds:
	\begin{itemize}
		\item  $\tr(T_{\alpha\beta}) = \tensor{T}{^\mu _\mu} = 0$,
		\item $\partial_{\bar{t}}\, \Theta = 0,\quad \partial_t \bar{\Theta} = 0\quad $ and $\quad [\Theta(t), \bar{\Theta}(\overbar{t})] = 0\quad\forall t,\bar{t}\in\reals. $
	\end{itemize}
\end{lem}

\begin{proof}
	Direct calculation implies that
	\begin{equation}\label{eq:derivatives of T zero}
		\partial_t \bar{\Theta} + \partial_{\bar{t}} T_{t\bar{t}} =0\quad\text{and}\quad  \partial_{\bar{t}} \,\Theta+\partial_t T_{\bar{t}t} = 0.
	\end{equation}
	In 2D, Lorentz boosts are just squeeze mappings
	\[
	\tensor{\Lambda}{^\mu_\nu}=
	\begin{pmatrix}
	\cosh\xi & \sinh\xi \\
	\sinh\xi & \cosh\xi
	\end{pmatrix}
	\]
	and the tensor field $T_{\mu\nu}$ transforms under Lorentz transformations as
	\[
		U(\Lambda)\, T_{\mu\nu}(\vec{x})\,U(\Lambda)^{-1} = \tensor{(\Lambda^{-1})}{_\mu ^\alpha }\tensor{(\Lambda^{-1})}{_\nu ^\beta }\, T_{\alpha\beta}(\Lambda\vec{x}),
	\]
	where $\tensor{(\Lambda^{-1})}{_\mu ^\nu } = \tensor{\Lambda}{^\mu _\nu}$. Thus,
	\[
		U(\Lambda)\,\bar{\Theta}(\vec{x})\, U(\Lambda)^{-1} = e^{2\xi}\, \bar{\Theta}\left(e^{\xi} \overbar{t}, e^{-\xi} t\right).
	\]
	Moreover, under dilations
	\[
		U'(\lambda)\, \bar{\Theta}(\vec{x})\, U'(\lambda)^{-1} = \lambda^2\, \bar{\Theta}(\lambda \vec{x}).
	\]
	Combining these transformations with $\lambda = e^{-\xi}$ we obtain
	\[
		U(\Lambda) U'(\lambda)\, \bar{\Theta}(\vec{x})\, U'(\lambda)^{-1}U(\Lambda)^{-1} = \bar{\Theta}(\overbar{t}, \lambda^2 t).	
	\]
	From \cref{thm:Wightman fields to Wightman distributions} it follows that
	\begin{equation}\label{eq:theta 2 point function}
		\left\langle \Omega, \bar{\Theta}(\overbar{t}_1,t_1) \bar{\Theta}(\overbar{t}_2,t_2)\Omega\right\rangle = \frac{A}{(\overbar{t}_1 -\overbar{t}_2 -i\veps)^4},\quad \overbar{t}_1\ne \overbar{t}_2,\quad A\in\complex.
	\end{equation}
	Here $\veps$ means that we take the limit $\veps\to 0$, i.e. our $(x\pm i\veps)^n = (x\pm i0)^n$ with
	\[
		(x\pm i0)^n=\lim\limits_{y\to0^{+}} (x\pm i y)^n.
	\]
	See \cite{Gelfand1964} for more details.
	
	We apply $\partial / \partial t_1$ and $\partial / t_2$ to get
	\[
		\langle\Omega,\, \partial_{t_1} \bar{\Theta}(\overbar{t}_1, t_1)\, \partial_{t_2} \bar{\Theta}(\overbar{t}_2, t_2) \Omega\rangle = 0. 
	\]
	Thus, by analytic continuation this distribution is identically zero throughout. Therefore, we have
	\[
		\partial_{t} \bar{\Theta}\Omega = 0
	\]
	and the Corollary of Reeh--Schlieder Theorem \ref{cor:reeh-schlieder} implies that
	\[
		\partial_{t} \bar{\Theta} = 0,
	\]
	i.e. $\bar{\Theta}$ depends only on $\overbar{t}$. Similarly, $\partial_{\overbar{t}} \Theta = 0$. Hence, from \eqref{eq:derivatives of T zero} it follows that
	\[
		\partial_{\overbar{t}} T_{t\overbar{t}} = \partial_{t} T_{t\overbar{t}} = 0,
	\]
	i.e. that $T_{t\overbar{t}}$ is constant. But
	\[
		U'(\lambda)\, T_{t\overbar{t}}(\overbar{t},t)\, U'(\lambda)^{-1} =\lambda^2\, T_{t\overbar{t}} (\lambda \overbar{t}, \lambda t), 
	\]
	so $T_{t\overbar{t}}=0$. Therefore,
	\[
	\tr\left(T_{\alpha\beta}\right)= \tensor{T}{^{\mu} _{\mu}}= g^{\mu\nu} T_{\mu\nu} = T_{00}- T_{11} = 4T_{t\overbar{t}} = 0,
	\]
	as required. Now 
	\[
		[\Theta(t),\bar{\Theta}(\overbar{t}) ] = 0\quad\quad\forall t,\overbar{t}\in\reals,
	\]
	by locality and the fact that $\Theta$ depends only on $t$ and $\bar{\Theta}$ depends only on $\overbar{t}$. 

\end{proof}

\begin{prop}\label{prop:theta commutator}
	We have
	\begin{align*}
		&[\Theta(t_1),\Theta(t_2)] = \frac{c}{24\pi} i^3 \delta'''(t_1-t_2) + 2 i \delta'(t_1-t_2) \Theta(t_2) - i \delta(t_1-t_2)\partial \Theta(t_2)\label{eq:theta commutator},\\
		&[\bar{\Theta}(\bar{t}_1),\bar{\Theta}(\bar{t}_2)] = \frac{\bar{c}}{24\pi} i^3 \delta'''(\bar{t}_1-\bar{t}_2) + 2 i \delta'(\bar{t}_1-\bar{t}_2) \bar{\Theta}(\bar{t}_2) - i \delta(\bar{t}_1-\bar{t}_2)\bar{\partial} \bar{\Theta}(\bar{t}_2)
	\end{align*}
	with $c,\bar{c}\ge0$. If parity is conserved, then $c = \bar{c}$.
\end{prop}

\begin{proof}
	By locality, $[\Theta(t_1),\Theta(t_2)]=0$ if $t_1 \ne t_2$ with $t_1,t_2\in\reals$. Let
	\[
	O_k(t_1)=\frac{i}{k!}\int t_2^k [\Theta(t_1+t_2),\Theta(t_1)]\, \mathrm{d} t_2, \quad\quad k\in\nat_0.
	\]
	The $O_k$'s are local self-adjoint fields. Therefore, using that $\Delta(\Theta)=2$ by \cref{axiom:W energy-momentum tensor} and the definition of $O_k$'s we get
	\[
	U(\lambda) O_k(t) U(\lambda)^{-1} = \lambda^{3-k} O_k(\lambda t).
	\]
	Moreover, $O_k$'s are covariant under translations. Hence
	\begin{align*} 
	\langle \Omega, O_k(t_1) O_k(t_2) \Omega \rangle &= A_k (t_1-t_2 - i\varepsilon)^{2k-6} \\
	&\eqdef{\mathrm{k\ge3}} (-1)^{k-3} A_k \int \frac{dp}{2\pi}e^{-ip(t_1-t_2)}\delta^{(2k-6)}(p),\quad
	A_k\in\complex.
	\end{align*}
	So by Bochner--Schwartz theorem (see \cite{Reed1975}), these distributions are not positive for $k\ge4$. 
	Hence, $O_k = 0$ for $k\ge 4$. Moreover, $O_3(t)$ is independent of $t$. By locality, $O_3$ commutes with all the fields. It is therefore proportional to the unit operator so set
	\[
	O_3 = -\frac{c}{24\pi}, \quad c\in\complex.
	\]
	Recall that by assumption \eqref{eq:T generates translations Wightman} from \cref{axiom:W energy-momentum tensor}, $\Theta(t)$ generates translations:
	\[ 
	\int \mathrm{d}t_1 \,[ \Theta(t_1), \Theta(t_2) ] = -i \partial \Theta(t_2). 
	\]
	Thus,
	\begin{equation}\label{eq:O_0}
	O_0(t)=\partial_t \Theta(t).
	\end{equation}
	
	Let $| \psi \rangle \in \hilb$ be arbitrary. Then by regularity theorem for tempered distributions \cite[Thm. V.10]{Reed1980}, we can write
	\[
	\langle \psi, [\Theta(t_1+t_2), \Theta(t_1) ] \Omega \rangle = \sum^K_{k=0} \delta^{(k)} (t_2) \psi_k (t_1),
	\]
	where $K\in\nat_0$ and $\psi_k(t_1)$ are some distributions. It follows from \cite[p. 177]{Reed1980} that
	\[
	\psi_k(t)= - i (-1)^k \langle \psi , O_k(t)  \Omega \rangle.
	\]
	In particular, $\psi_k = 0$ for $k\ge4$ and
	\begin{equation}\label{eq:theta commutator expansion}
	[\Theta(t_1+t_2),\Theta(t_1)] = -i \sum^3_{k=0} (-1)^k \delta^{(k)}(t_2) O_k(t_1)
	\end{equation}
	holds on the vacuum and thus as an operator equality by the Reeh--Schlieder Theorem (\cref{cor:reeh-schlieder}).
	
	To determine $O_1(t)$ and $O_2(t)$ we use $[\Theta(t_1),\Theta(t_2)]=-[\Theta(t_2),\Theta(t_1)]$ and \eqref{eq:theta commutator expansion} to obtain
	\begin{equation}\label{eq:commutator symmetry}
	-\sum_{k=0}^3 (-1)^k \delta^{(k)}(t_2) O_k(t_1) = \sum_{k=0}^3 (-1)^k \delta^{(k)} (-t_2) O_k(t_1+t_2).
	\end{equation}
	Note that 
	\[
	\delta^{(k)}(-t_2) O_k(t_1+t_2) = \sum_{l=0}^k (-1)^l {k \choose l} \delta^{(l)} (t_2) \frac{\partial^{k-l}}{\partial t_1^{k-l}} O_k(t_1).
	\]
	Plugging this equation into \eqref{eq:commutator symmetry} and equating the coefficients of $\delta^{(2)}(t_2)$'s we get
	\[ 
	O_2(t)= \sum_{k=2}^3 (-1)^{k+1} \frac{k(k-1)}{2} \frac{\partial^{k-2}}{\partial t^{k-2}} O_k(t) =-O_2(t) \implies O_2 = 0,
	\]
	where we have used that $O_3$ is a constant. Moreover, the terms with $\delta(t_2)$ give
	\begin{align*}
	&O_0(t) = \sum_{k=0}^3 (-1)^{k+1} \frac{\partial^k}{\partial t^k} O_k(t)=-O_0(t) + \frac{\partial}{\partial t} O_1(t)\implies\\
	&\implies\frac{\partial}{\partial t} O_1(t) = 2 O_0(t) \eqdef{\eqref{eq:O_0}} 2 \frac{\partial}{\partial t} \Theta(t) \implies O_1(t) = 2 \Theta(t),
	\end{align*}
	where the last implication is by locality and dilation invariance. Plugging the expressions of $O_0(t)$, $O_1(t)$ and $O_3$ into \eqref{eq:theta commutator expansion} we prove the commutation relation for $\Theta$.
	
	The proof for $\bar{\Theta}$ is analogous.
	
	To prove that $c\ge0$, we first of all note that
	\begin{equation}\label{eq:correlation function of commutator theta}
	\langle \Omega, [\Theta(t_1),\Theta(t_2)]\, \Omega \rangle =  -i\frac{c}{24\pi}\delta'''(t_1 -t_2)
	\end{equation}
	since $\langle \Omega, \Theta(t)\, \Omega \rangle = 0$ by translation and dilation invariance. Moreover, the ``unbarred'' version of \eqref{eq:theta 2 point function} gives
	\[	\left\langle \Omega, \Theta(t_1) \Theta(t_2)\Omega\right\rangle = \frac{A}{(t_1 -t_2 -i\veps)^4},\quad t_1\ne t_2,\quad A\in\complex.
	\]
	Hence, using
	\[
		\delta'''(t) = -\frac{6}{2\pi i}\left( (t-i\veps)^{-4} - (t+i\veps)^{-4} \right)
	\]
	we see that $A = c/8 \pi^2$. Then the Fourier transform
	\begin{align*}
		(2\pi)^2 \langle\Omega, \Theta(t_1)\Theta(t_2)\Omega\rangle &= \frac{c}{2(t_1-t_2-i\veps)^4}=-\frac{\mathrm{d}^3}{\mathrm{d}t_{12}^3}\frac{c}{12(t_{12}-i\veps)}\\
		&=-i\frac{c}{12}\frac{\mathrm{d}^3}{\mathrm{d} t_{12}^3}\int^{\infty}_0 e^{-ip t_{12}} \mathrm{d}p = \frac{c}{12}\int^{\infty}_0 p^3 e^{-i p t_{12}} \mathrm{d}p
	\end{align*}
	implies that we must have $c\ge 0$ to ensure the positivity of the correlation function.
	
	To show that $c=\bar{c}$ if parity is conserved, we use parity invariance ($t = \overbar{t}$) to get that $\Theta = \bar{\Theta}$ and so by \eqref{eq:correlation function of commutator theta}
	\[
	-i\frac{c}{24\pi}\delta'''(t_1 -t_2) = -i\frac{\bar{c}}{24\pi}\delta'''(t_1 -t_2) \implies c=\bar{c}.
	\]
%
\end{proof}

\begin{rmk}\label{rmk:theta commutator}
	Note that as an operator equation we have
	\begin{align}
		[\Theta(t_1),\Theta(t_2)] &= \frac{c}{24\pi} i^3 \delta'''(t_1-t_2) + 2 i \delta'(t_1-t_2) \Theta(t_2) - i \delta(t_1-t_2)\partial \Theta(t_2) \nonumber\\
		&= \frac{c}{24\pi} i^3 \delta'''(t_1-t_2) + i\delta'(t_1-t_2) \left\{\Theta(t_1) + \Theta(t_2)\right\}
	\end{align}
	and similarly for $\bar{\Theta}$.
\end{rmk}

\begin{figure}
	\centerline{\includegraphics{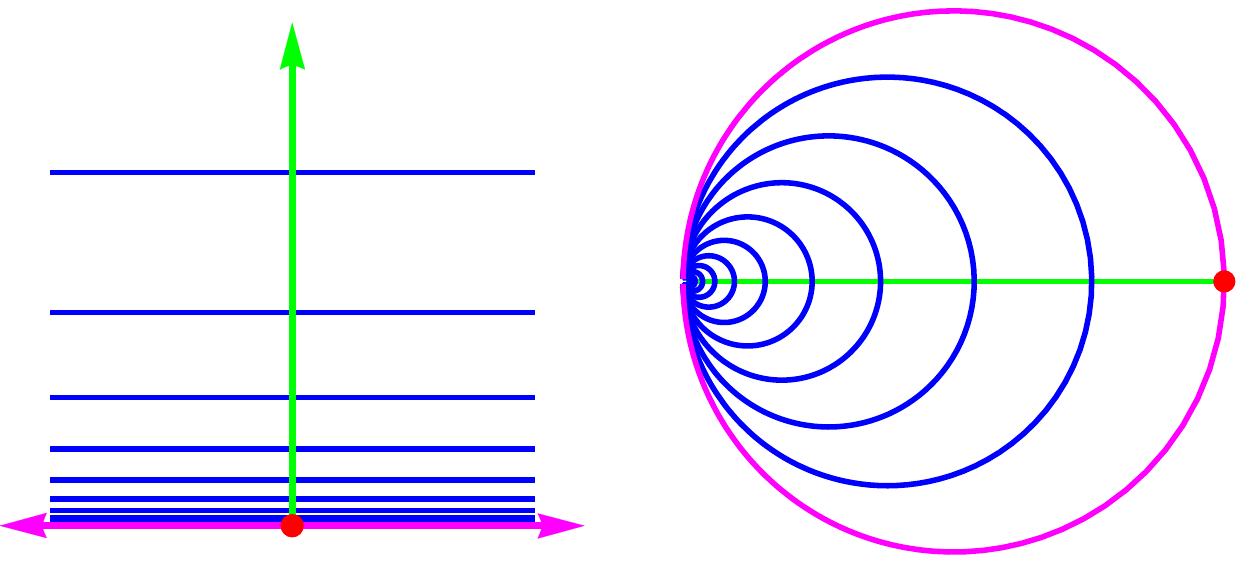}}
	\caption{Cayley transform is a biholomorphic map from the open upper half-plane to the open unit disk}\label{fig:Cayley transform}
\end{figure}

By \cref{lem:analytic Glaser extension} each vector
$\Psi(x_1,\ldots,x_n):=\phi_{a_1}(x_1)\ldots\phi_{a_n}(x_n)\vac$
of $\hilb$ extends analytically to the domain containing
\[
\{\im t_1,\im \overbar{t}_1>0 \}\times\dots\times\{\im t_n,\im \overbar{t}_n>0\}
\;\;\text{such that}\;\; t_i+ \overbar{t}_i\ne t_j+\overbar{t}_j\text{ if } i\ne j.
\]
Since vectors of the form $\Psi$ are dense in the Hilbert space $\hilb$ by \cref{axiom:W completeness}, each field's $\phi_a$ domain of definition extends to Schwartz functions on $\complex^2$ supported in $\im t,\im \overbar{t}\ge0$. This allows us to compactify the Minkowski space
\[ 
z=\frac{1+it/2}{1-it/2},\quad \bar{z}=\frac{1+i\overbar{t}/2}{1-i \overbar{t}/2}. 
\]
Under these transformations, the domain $ \im t>0$, $\im \overbar{t}>0 $ is mapped to the domain $ |z|<1$, $|\bar{z}|<1 $. Define the holomorphic energy-momentum tensor in the domain $ |z|<1 $ by
\begin{equation}\label{eq:def of T} 
T(z):=2\pi \left(\frac{2}{1+z}\right)^{4}\Theta(t),
\end{equation}
with $ t=2i(1-z)(1+z)^{-1} $ and similarly for the antiholomorphic tensor $\overbar{T}(\bar{z})$. 


\begin{thm}\label{thm:Luscher-Mack final}
	Every 2D dilation invariant Wightman QFT with an energy-momentum tensor, i.e. a 2D Wightman QFT satisfying \crefrange{axiom:W dilation invariance}{axiom:W energy-momentum tensor}, gives rise to two commuting unitary representations of the Virasoro algebra with central charges $c$ and $\bar{c}$ with generators defined by
	\[ 
		L_n := \oint\limits_{S^1} z^{n+1} T(z) \frac{\mathrm{d}z}{2 \pi i}\,,\quad\quad \bar{L}_n := \oint\limits_{S^1} \bar{z}^{n+1} \overbar{T}(\bar{z}) \frac{\mathrm{d}\bar{z}}{2 \pi i}\,.
	\]
%
\end{thm}

\begin{proof}
	We define the circular delta function by
	\[ 
		\oint\limits_{S^1} \delta_c(z-z_0) f(z) \frac{\mathrm{d}z}{2\pi i} = f(z_0),
	\]
	so that $\delta_c (z-z_0)\, \mathrm{d} z = 1/(z-z_0)\, \mathrm{d} z$ by Cauchy's integral formula. Using the definition of $T(z)$,  \cref{prop:theta commutator} and \cref{rmk:theta commutator} we get
	\[
		[T(z),T(w)] = \frac{c}{12}\delta_c'''(z-w)+\delta_c'(z-w) \left(T(z)+T(w)\right).
	\]
	Defining
	\[
		L_n := \oint\limits_{S^1} z^{n+1} T(z) \frac{\mathrm{d}z}{2 \pi i},
	\]
	we get by direct calculation
	\begin{equation}\label{eq:Vir in proof of Luescher-Mack}
		[L_m, L_n] = (m-n) L_{m+n} + \frac{c}{12}(m^3-m) \delta_{m+n,0},
	\end{equation}
	as required.
	
	To prove unitarity, i.e. that $L_n^* = L_{-n}$, we note that
	\[ 
	L_0 = \oint\limits_{S^1} z\, T(z) \frac{\mathrm{d}z}{2\pi i} = \int\limits_{-\pi}^{\pi} \Theta\left(2\tan\frac{\alpha}{2}\right)\frac{\mathrm{d}\alpha}{\cos^4(\alpha/2)} = \int\limits_{-\infty}^{\infty}\Theta(t) \left(1+\frac{t^2}{4}\right)\mathrm{d}t
	\]
	is a self-adjoint operator. From \eqref{eq:Vir in proof of Luescher-Mack} it follows that
	\[
	[L_n,L_0] = n L_n \implies [L_0, L_n^*] = n L_n^*\implies L_n^* = L_{-n}.
	\]
\end{proof}

	\chapter{Virasoro Algebra and Vertex Algebras}\label{chap:Vir and vertex}

In this chapter we study the relationship between vertex algebras and the Virasoro algebra. As an attentive reader could have already guessed, the relationship is rather trivial.

\section{From Virasoro Algebra to Virasoro Vertex Algebra}\label{sec:Virasoro to VA}

By Poincar\'e--Birkhoff--Witt Theorem, an equivalent definition of the Verma module $M(c,h)$ is obtained by setting
\[
	M(c,h) = U(\Vir)\otimes_{U(\mathfrak{b})}\complex,
\]
where $\mathfrak{b} = \left ( \oplus_{n\ge1} \complex L_n\right)\oplus \left(\complex L_0\oplus \complex C\right)$ is a subalgebra of $\Vir$, $U(\Vir)$ is the universal enveloping algebra of $\Vir$ and $\complex$ denotes a 1-dimensional $\mathfrak{b}$-module
\begin{align*}
	L_n\vac &= 0,\quad n\ge1,\\
	L_0\vac &= h\vac,\\
	C\vac &= c\vac.
\end{align*}
This should be compared with our explicit construction of \cref{lem:construction of Verma module}. 

Frenkel and Zhu have shown in \cite{Frenkel1992} that
\[
	\overline{M(c,0)} = M(c,0)/ \left(U(\Vir)L_{-1} \vac\otimes\vac\right)
\]
has a vertex operator algebra structure with the conformal vector $\nu = L_{-2}\vac$. We present here an explicit construction as given in \cite[p. 193]{Schottenloher2008} in subsection ``Virasoro Vertex Algebra".

\begin{prop}\label{thm:when Vir unitary}
	The quotient $\overline{M(c,0)}$ gives rise to a vertex operator algebra of CFT type.
\end{prop}

\begin{proof}
We give a construction similar to that of a Verma module $M(c,0)$ in \cref{lem:construction of Verma module}.

Let $\overline{M(c,0)}$ be a vector space with a basis
\[
	\{v_{n_1\ldots n_k} \mid n_1\ge\ldots \ge n_k\ge 2,\; n_j\in\nat,\; k\in\nat\}\cup \{\vac\}
\]
together with the following action of $\Vir$ on $\overline{M(c,0)}$ for all $n,n_j\in\integ$ such that $n_1\ge\ldots\ge n_k\ge2$, $k\in\nat$:
\begin{alignat*}{2}
	C&:= c\id,\\
	L_n\vac&:=0, &&n\ge -1,\\
	L_0 v_{n_1\ldots n_k}&:=\left(\sum_{j=1}^{k}n_j\right) v_{n_1\ldots n_k},\\
	L_{-n}\vac&:=v_n, &&n\ge2,\\
	L_{-n}v_{n_1\ldots n_k}&:= v_{n n_1\ldots n_k}, &&n\ge n_1.
\end{alignat*}
Other actions of $L_n$'s on general $v\in \overline{M(c,0)}$ follow from the commutation relations of the Virasoro algebra. Defining $L(z):= \sum_{n\in\integ} L_n z^{-n-2}$ we see that $L(z)$ is a field, as follows by generalizing
\[
	L_m v_n = L_m L_{-n}\vac= L_{-n} L_m \vac +(m+n)L_{m-n} \vac=0\;\quad m\gg 0,
\]
to arbitrary $v\in \overline{M(c,0)}$.
Moreover, it is a Virasoro field as follows from \cref{def:conformal vertex algebra} and \cref{ex:Virasoro formal distribution}. Hence, $L(z)$ is also local with respect to itself. For the asymptotic state using $L_n\vac=0$ $\forall n\ge-1$ and $L_{-n}\vac=v_n$ $\forall n\ge2$ we get
\[
	L(z)\vac|_{z=0} = \sum_{n\le -2} L_n z^{-n-2}\vac |_{z=0}= L_{-2}\vac = v_2.
\]
Moreover, note that 
\begin{align*}
	[L_{-1},L(z)]&=\sum_{n\in\integ}[L_{-1},L_n ]z^{-n-2}=\sum_{n\in\integ}(-1-n)L_{-1+n}z^{-n-2}=\\
	&=\sum_{n\in\integ}(-n-2)L_{n} z^{-n-3}=\partial L(z).
\end{align*}
Thus, setting $T:=L_{-1}$ for clarity, we apply \cref{thm:construction of vertex algebras} to get a vertex algebra with a single strongly generating field $L(z)$ (\cref{def:generating set of fields}). Clearly, it is a vertex operator algebra of CFT type with conformal vector $v_2$ (\cref{def:VOA}).
\end{proof}
 In the same paper \cite{Frenkel1992}, Frenkel and Zhu have also shown that $L(c,0)=M(c,0)/J(c,0)$, where $J(c,0)$ is the maximal invariant subspace such that $L(c,0)$ is an irreducible highest weight representation of the Virasoro algebra (\cref{thm:Verma module indecomposable etc}), is the unique irreducible quotient VOA of $\overline{M(c,0)}$. The VOA $L(c,0)$ is called the \textbf{Virasoro VOA} with central charge $c$.

Now we prove that $L(c,0)$ is a unitary VOA for $c\in\reals$. The proof is due to Dong and Lin \cite{Dong2014}.

For $c\in\reals$ define an antilinear map $\overline{\vphi}: \overline{M(c,0)} \to \overline{M(c,0)}$ by
\[
	L_{-n_1}\ldots L_{-n_k}\vac\mapsto L_{-n_1}\ldots L_{-n_k} \vac, \quad n_1\ge \ldots \ge n_k \ge 2.
\]
\begin{lem}
	The map $\overline{\vphi}$ is an antilinear involution of the VOA $\overline{M(c,0)}$ $\;\forall c\in\reals$. Moreover, $\vphi$ induces an antilinear involution $\vphi$ of $L(c,0)$.
\end{lem}
\begin{proof}
	Since $\overline{\vphi}^2 =\id$, it suffices to prove that $\overline{\vphi}$ is an antilinear automorphism. Let $U$ be a subspace of $\overline{M(c,0)}$ defined by
	\[
		U=\{ u \in \overline{M(c,0)} \mid \vphi( u_n v) = \vphi (u)_n \vphi(v)\quad\forall v\in \overline{M(c,0)}\, ,\; \forall n\in\integ\}.
	\]
	By associativity of $\End \overline{M(c,0)}\,$, it follows that if $a,b\in U$, then $a_m b\in U$ $\,\forall m\in\integ$. Moreover, $\vac\in U$ and $\nu = L_{-2}\vac\in U$. Hence, $U= \overline{M(c,0)}$ since $\overline{M(c,0)}$ is generated by $\nu$, as required.
	
	Let $\overline{J(c,0)}$ be the maximal proper $L$-submodule of $\overline{M(c,0)}$. Then $\overline{\vphi}(\overline{J(c,0)})$ is a proper $L$-submodule of $\overline{M(c,0)}$. Thus, $\overline{\vphi}(\overline{J(c,0)})\subset \overline{J(c,0)}$ and so $\overline{\vphi}$ induces an antilinear involution $\vphi$ of $L(c,0)$.
\end{proof}

All in all, rather unsurprisingly, we get a result equivalent to the representation theory of the Virasoro algebra (cf. \cref{thm:unitarity of Verma module}).
\begin{thm}
	Let $c\in\reals$ and $\vphi$ be the antilinear involution of $L(c,0)$ defined above. Then $(L(c,0),\vphi)$ is a unitary VOA if and only if $c\ge 1$ or $c = c(m)$ for $m\in\nato$, where $c(m)$ is defined in \cref{eq:c(m) definition}.
\end{thm}
\begin{proof}
	If $c\ge 1$ or $c= c(m)$ for some $m\in\nato$, then there exists a Hermitian form $(\cdot,\cdot)$ on $L(c,h)$ (\cref{def:Hermitian form on Verma module}) and if it also satisfies
	\[
	( L_n v, w) = ( v, L_{-n} w),\quad ( Cv,w) = ( v, Cw),\quad ( \vac, \vac ) =1,
	\]
	then it is positive definite on $L(c,h)$ by \cref{rmk:unique L} and \cref{thm:unitarity of Verma module}. So we just need to prove the invariance property. By \cite{Frenkel1992}, the vertex operators $L(n)$ of $L(c,h)$ and Virasoro generators $L_n$ coincide on $L(c,0)$, i.e. we have $L(n)u=L_n u$ for all $u\in L(c,0)$. We continue writing $L_n$ for vertex operators as elsewhere in this work. This implies that $( L_n u, v) = ( u, L_{-n} v)$ for all $u,v\in L(c,0)$. Hence,
	\begin{align*}
		( u, Y\boldsymbol{(}e^{zL_1}(&-z^{-2})^{L_0}\nu,z^{-1}\boldsymbol{)}v)=z^{-4} (u, Y(\nu,z^{-1})v)=\\
		&=\sum_{n\in\integ} (u, \nu_{(n+1)} v)z^{n-2} = \sum_{n\in\integ} (u, L_n v) z^{n-2} =\sum_{n\in\integ} (L_{-n} u, v) z^{n-2}\\
		&= \sum_{n\in\integ} (\nu_{(-n+1)}u,v)z^{n-2} = (Y(\nu,z)u, v)=(Y(\vphi(\nu),z)u,v).
	\end{align*}
	Since $L(c,0)$ is generated by $\nu$, $(L(c,0), \vphi)$ is a unitary VOA by \cite[Prop. 2.11]{Dong2014} which states that a VOA is unitary if the invariant property holds on its generators.
	
	Conversely, if $(L(c,0), \vphi)$ is a unitary VOA, then $L(c,0)$ is a unitary module of the Virasoro algebra by \cite[Lem. 2.5]{Dong2014}. Therefore, $c\ge1$ or $c= c(m)$, as required.
\end{proof}

\section{From Vertex Algebra to Virasoro Algebra}

Now the converse is a tautology---every conformal vertex algebra has at least one representation of the Virasoro algebra encoded in itself. Moreover, a unitary vertex algebra contains a unitary representation of the Virasoro algebra. 

	\chapter{Wightman QFT and Vertex Algebras}\label{chap:Wightman and vertex}

In this chapter we will show that a Wightman (M\"obius) CFT gives rise to two commuting (M\"obius) conformal vertex algebras and conversely that two unitary (quasi)-vertex operator algebras can be combined to give a Wightman (M\"obius) CFT. The only other reference providing the converse proof of which we are aware of is \cite{Nikolov2004}. However, in \cite{Nikolov2004}, Nikolov studies higher dimensional vertex algebras and gets a one-to-one correspondence between them and Wightman QFTs with global conformal invariance \cite{Nikolov2001}. The relationship between these higher dimensional vertex algebras and the vertex algebras used elsewhere in our work is not explicitly discussed in \cite{Nikolov2004}. Moreover, the proof itself is different from ours. Therefore, we hope that our proof will still be useful.

\section{From Wightman CFT to Vertex Algebras}\label{sec:From Wightman CFT to vertex}

The following theorem is due to Kac \cite[Sec. 1.2]{Kac1998}. We have also used \cite{Furlan1989} for clarifications and minor changes in normalization.
\begin{thm}[\cite{Kac1998}]\label{thm:Wightman to VA}
	Every Wightman M\"obius CFT gives rise to two commuting strongly-generated positive-energy M\"obius conformal vertex algebras. Moreover, if conformal weights are integers and the number of the generating fields of each conformal weight is finite, then these algebras are also unitary quasi-vertex operator algebras.
\end{thm}

\begin{proof}
Introduce the light cone coordinates $ t:=x^0-x^1 $ and $ \overbar{t}:=x^0+x^1 $ so that $ |x|^2=t \overbar{t} $. Define
\[ 
P:=\frac{1}{2} (P_0 - P_1)\quad\text{and}\quad \bar{P}:=\frac{1}{2}(P_0+P_1).
\]
Furthermore, let
\[
\vac:=\Omega.
\]
%
%
%
In the light-cone coordinates special conformal transformations decouple
\[
t^b=\frac{t}{1+b^+ t},\quad\quad \overbar{t}^b = \frac{\overbar{t}}{1+b^- \overbar{t}},
\]
where $ b^{\pm}=b^0\pm b^1$. Since translations and special conformal transformations generate the whole of $\pslinear(2,\reals)$ by \cref{prop:generators of PSL2}, the restricted conformal group acts as
\[ \gamma(t,\overbar{t})= \left(\frac{at+b}{ct+d}\,,\frac{\bar{a}\overbar{t}+\bar{b}}{\bar{c}\overbar{t}+\bar{d}}\right),            \]
where 
\[
\begin{pmatrix} 
a & b \\ 
c & d 
\end{pmatrix}, 
\begin{pmatrix} 
\bar{a} & \bar{b} \\ 
\bar{c} & \bar{d} 
\end{pmatrix} \in SL_2(\reals).
\] 
The transformation law for quasiprimary fields \eqref{eq:conformal covariance} becomes
\begin{equation}\label{eq:lightcone conformal covariance}
 U(\gamma) \phi_a(t,\overbar{t})U(\gamma)^{-1}=(ct+d)^{-2h_a}\left(\bar{c}\overbar{t}+\overbar{d}\right)^{-2\bar{h}_a}\phi_a(\gamma(t,\overbar{t}))
\end{equation}
with $h=(\Delta+s)/2$ and $\bar{h}=(\Delta-s)/2$.

Define 
\[ K:=-\frac{1}{2} (K_0+K_1)\quad\text{and}\quad \bar{K}:=\frac{1}{2}(K_1-K_0).\]

We now focus on the $t$ coordinate, but the same holds for $\bar{t}$.

Let $D$ be the generator of dilations in $t$, i.e.
\begin{equation}\label{eq:D generates dilations}
	e^{i\lambda D} \phi(t,\bar{t}) e^{-i\lambda D} = e^{\lambda h} \phi(e^{\lambda} t,\bar{t}),\quad \lambda>0.
\end{equation}
Then from \eqref{eq:translation covariance locally}, \eqref{eq:special conformal covariance locally} and \eqref{eq:D generates dilations} it follows that in light-cone coordinates
\begin{subequations}\label{eqs:P and K local commutators}
\begin{alignat}{2}
& i[P,\phi_a(t,\overbar{t})]&&= \partial_t \phi_a(t,\overbar{t}),\\
& i[D,\phi(t,\overbar{t})]&&= (t\partial_t+h_a) \phi(t,\overbar{t})\\
& i[K,\phi_a(t,\overbar{t})]&&=(t^2\partial_t+2h_a t)\phi_a(t,\overbar{t})\label{eq:K commutator},
\end{alignat}
\end{subequations}
with $\partial_t:=1/2(\partial_0-\partial_1)$. Note that to prove \eqref{eq:K commutator} it might be easier to start from \eqref{eq:lightcone conformal covariance} with $\gamma$ being special conformal transformation, see \cite{Anshuman2015}. We also have
\[
	[P,K]\phi_a(t,\bar{t})\vac=\left[[P,K],\phi_a(t,\bar{t})\right]\vac
\]
and similarly for the others, so that equations \eqref{eqs:P and K local commutators} imply
\[
	[D,P]=-iP,\quad [D,K]=i K,\quad [P,K]=2i D
\]
on $\domain_0$, i.e. $D,P,K$ form a representation of $\slLie(2,\complex)$. In particular, if we set $P= -i A$, $D = -i B$ and $K = - i C$ with
\[
A = \begin{pmatrix*}
	0 & 1  \\
	0 & 0  
	\end{pmatrix*},\quad
B = \begin{pmatrix*}[c]
	1/2 & 0\\
	 0  & -1/2
	\end{pmatrix*},\quad
C = \begin{pmatrix*}[r]
	 0	&	0\\
	-1	&	0
	\end{pmatrix*},
\]
then
\begin{equation}\label{eq:P and K are adjoint}
	w A w^{-1} = C\quad\text{with}\quad w =\begin{pmatrix*}[r]
												0	& 1\\
												-1	& 0
										    \end{pmatrix*}\in \slLie(2,\complex).
\end{equation}

By \cref{lem:analytic Glaser extension} each vector
$\Psi(x_1,\ldots,x_n):=\phi_{a_1}(x_1)\ldots\phi_{a_n}(x_n)\vac$
of $\hilb$ extends analytically to the domain containing
\[
\{\im t_1,\im \overbar{t}_1>0 \}\dots\times\{\im t_n,\im \overbar{t}_n>0\}\quad\text{such that}\quad t_i+\overbar{t}_i\ne t_j+\overbar{t}_j\text{ if } i\ne j.
\]
Since vectors of the form $\Psi$ are dense in the Hilbert space $\hilb$ by \cref{axiom:W completeness}, each field's $\phi_a$ domain of definition extends to Schwartz functions on $\complex^2$ supported in $\im t,\im \overbar{t}\ge0$. This allows us to make conformal transformations defined everywhere by compactifying the Minkowski space using the Cayley transform (Figure \ref{fig:Cayley transform})
\[ 
z=\frac{1+it/2}{1-it/2},\quad \bar{z}=\frac{1+i\overbar{t}/2}{1-i\overbar{t}/2}. 
\]
Under these transformations, the domain $ \im t>0$, $\im \overbar{t}>0 $ is mapped to the domain $ |z|<1$, $|\bar{z}|<1 $. Define the new fields in the domain $ |z|<1 $, $ |\bar{z}|<1 $ by
\[
 Y(a,z,\bar{z}):=2\pi\left(\frac{2}{1+z}\right)^{2 h_a}\left(\frac{2}{1+\bar{z}}\right)^{2\bar{h}_a}\phi_a(t,\overbar{t}),
  \]
with $ t=2i(1-z)(1+z)^{-1} $ and $ \overbar{t}=2i(1-\bar{z})(1+\bar{z})^{-1}$ (cf. \eqref{eq:def of T}). By the analytic extension, 
\begin{equation}\label{eq:asymptotic states in vertex algebra Wightman}
a:=Y(a,z,\bar{z})|0\rangle|_{z,\bar{z}=0}
\end{equation}
is a well-defined vector in $ \domain_0 $. Furthermore, $ Y(a,z,\bar{z})\mapsto a $ is a linear injective map.

Define
\begin{align*}
T:=P-\frac{1}{4}K-&iD, \quad\quad T^*:=P-\frac{1}{4}K+iD,\\
H&:=P+\frac{1}{4} K.
\end{align*}
By direct calculation from \eqref{eqs:P and K local commutators} it follows that

\begin{subequations}\label{eq:T, H commutators Wightman}
	\begin{align}
	[T,Y(a,z,\bar{z})]  &=\partial_z Y(a,z,\bar{z}),        \label{eq:T Y commutator Wightman}\\
	[H, Y(a,z,\bar{z})] &= (z\partial_z + h_a) Y(a,z,\bar{z})\label{eq:H Y commutator Wightman},\\
	[T^*, Y(a,z,\bar{z})] &= (z^2\partial_z + 2h_a z) Y(a,z,\bar{z}).\label{eq:Tstar Y commutator Wightman}
	\end{align}
\end{subequations}
Moreover, operators $T$, $H$ and $T^*$ annihilate the vacuum since $P$ and $K$ do. Thus, we have
\[
	[H,T]=T,\quad [H,T^*]=-T^*,\quad [T^*,T]=2H
\]
on $\domain_0$ (cf. \cite{Anshuman2015second}).

Applying \eqref{eq:H Y commutator Wightman} to the vacuum and letting $z=\bar{z}=0$ we obtain
\begin{equation}\label{eq:Hamiltonian Wightman}
Ha = h_a a. 
\end{equation}
The operator $P$ is self-adjoint and semi-definite on $\hilb$ by \cref{axiom:W spectrum}. Same holds for $K$ due to \eqref{eq:P and K are adjoint}. Hence, by definition, $H$ is also self-adjoint semi-definite. Therefore, conformal weights are non-negative real numbers. 

The locality axiom \cref{axiom:W locality} in light-cone coordinates is
\begin{equation}\label{eq:lightcone locality}
	\phi_a(t,\overbar{t})\phi_b(t',\overbar{t}')=\phi_b(t',\overbar{t}')\phi_a(t,\overbar{t})\quad\text{if}\quad (t-t')(\bar{t}-\overbar{t}')<0.
\end{equation}
Let us now consider the right chiral Wightman fields---fields satisfying $\partial_{\overbar{t}}\phi_a=0$. Then the locality condition becomes
\[
\phi_a(t)\phi_b(t')=\phi_b(t')\phi_a(t)\quad\text{if}\quad t\ne t'
\]
and since Wightman fields are operator-valued distributions we have
\[
\left[\phi_a(t),\phi_b(t')\right]= \sum_{j\ge0} \delta^{(j)}(t-t') \psi_j(t')
\]
for some fields $\psi_j(t')$. For fields $\psi_j(t')$ the general Wightman axioms \crefrange{axiom:W covariance}{axiom:W locality} hold, but they are not necessarily quasiprimary as defined in \cref{axiom:W conformal covariance}. 
So let us add such fields to our algebra to obtain:
\[
[Y(a,z), Y(b,z')] = \sum_{j\ge0} \delta^{(j)} (z-z')Y(c_j,z').
\]
The map $Y(c_j,z')\vac|_{z=0}=c_j$ is also well-defined, since we used only the general Wightman axioms \crefrange{axiom:W covariance}{axiom:W locality} to extend the fields in \cref{lem:analytic Glaser extension}.

Now the Wightman field $Y(c_j,z')$ has conformal weight $h_a+h_b-j-1$ as can be seen by applying $[H, \cdot\;]$ to both sides of this equality and using \eqref{eq:H Y commutator Wightman} with \cref{prop:equivalent Hamiltonian}. The positivity of conformal weights implies that the sum on the right-hand side is finite. Thus,
\[
(z-z')^N [ Y(a,z), Y(b,z')] = 0 \quad\text{for}\quad N\gg 0,
\]
by the properties of the delta distribution. 

Now we want to write the Wightman fields in a Fourier series
\begin{equation}\label{eq:vertex algebra field mode expansion Wightman}
Y(a,z)=\sum_n a_{(n)} z^{-n-1},
\end{equation}
with $a_{(n)}\in \End\domain_0$. However, it is not obvious that such an expansion is well-defined. Since it is an operator equality, it suffices to prove the equality on $\domain_0$, i.e. we have to prove that
\begin{align*}
Y(a,z)Y(b_1,w_1)\ldots Y(b_n,w_n)\vac&=\\
\sum_{k,k_1,\ldots,k_n}& a_{(k)}b_{(k_1)}\ldots b_{(k_n)}z^{-k-1} w_1^{-k_1-1}\ldots w_n^{-k_n-1}\vac.
\end{align*}
 Note that in $|z|<1$ with $h\ge0$ the function
 \[
 \frac{1}{(1+z)^{2h}}
 \]
 is holomorphic and hence analytic. Therefore, $Y(a,z)Y(b_1,w_1)\ldots Y(b_n,w_n)\vac$ is analytic, since $\phi_a(t)\phi_{a_1}(t_1)\ldots\phi_{a_n}(t_n)\vac$ is analytic as proven above.
 
Let $V$ be the subspace of $\domain_0$ spanned by all polynomials in the $a_{(n)}$ applied to the vacuum vector $\vac$. Clearly $V$ is invariant with respect to all $a_{(n)}$'s and with respect to $T$ since by \eqref{eq:T Y commutator Wightman} we have
\[
\sum_n[T,a_{(n)}]z^{-n-1}=\sum_n (-n-1)a_{(n)}z^{-n-2}= \sum_n -na_{(n-1)}z^{-n-1}.
\]
Thus, 
\[
[T,a_{(n)}]=-na_{(n-1)} 
\]
and because $T\vac=0$ by \cref{axiom:W conformal covariance},
\begin{equation}\label{eq:T commutator on the modes Wightman}
Ta_{(n)}\vac=-na_{(n-1)}\vac.\
\end{equation}

Now we prove that $Y(a,z)$'s are fields in vertex algebra sense (\cref{def:field in vertex algebra}). 
Similarly like for $T$ in \eqref{eq:T commutator on the modes Wightman}, \cref{eq:H Y commutator Wightman}
gives
\[
[H,a_{(n)}]=(h_a-n-1)a_{(n)}. 
\]
Given $v=b_{(j)}\vac\in V$ we get
\begin{align*}
(h_a-n-1)a_{(n)}v&=[H,a_{(n)}]v=Ha_{(n)}v-a_{(n)}H b_{(j)}\vac=\\
&=H a_{(n)}v-a_{(n)}[H,b_{(j)}]\vac=H a_{(n)}v-a_{(n)}(h_b-j-1)v.
\end{align*}
Hence,
\[
Ha_{(n)}v=(h_a+h_b-j-n-2)a_{(n)}v.
\]
Thus, the Wightman field $Y(a_{(n)}v,z)$ has conformal weight $h_a+h_b-j-n-2$, since
\[
Ha=h_a a\iff [H, Y(a,z,\bar{z})] = (z\partial_z + h_a) Y(a,z,\bar{z}).
\]
But conformal weights are non-negative real numbers. Hence, $a_{(n)}v=0$ for $n\gg 0$. The above reasoning clearly holds $\forall v\in V$. Therefore, the Wightman fields $Y(a,z)$ for $a\in V$ are also vertex algebra fields and we can use the Existence \cref{thm:construction of vertex algebras} to obtain a vertex algebra.

Combining the expansion \eqref{eq:vertex algebra field mode expansion Wightman} with the definition of $a$ \eqref{eq:asymptotic states in vertex algebra Wightman} we obtain
\begin{equation}\label{eq:a_n on the vacuum give zero for vertex algebras Wightman}
a_{(n)}\vac = 0\quad\forall n\ge 0.
\end{equation}
Moreover, note that given two generators $a_{(m)}$, $m\ge 0$, and $b_{(j)}$, $j<0$, their commutator is $[a_{(m)},b_{(j)}]= \sum_{k<m} c_{(k)}$ for some generators $c_{(k)}$ with $k<m$ by Borcherds commutator formula \eqref{eq:Borcherds coefficient commutator}. Thus, by generalizing the simple calculation
\[ 
a_{(m)}b_{(j)} \vac = b_{(j)}a_{(m)}\vac + [a_{(m)},b_{(j)}]\vac = 0+\sum\limits_{k<0} c_{(k)}\vac
\]
it follows that $V$ is strongly generated by the fields $Y(a,z)$ (\cref{def:generating set of fields}).

Now if we take the left chiral fields, i.e. fields satisfying $\partial_{t}\phi_{i}=0$, and apply the same reasoning as above, we obtain the left vertex algebra $\bar{V}$ with the same vacuum vector $\vac$, the infinitesimal translation operator $\overbar{T}$ and fields $Y(\bar{a},\bar{z})$ with $\bar{a}\in\bar{V}$. From \eqref{eq:lightcone locality} we see that locality in the mixed chiral case boils down to $\phi_a(t)\phi_{\bar{a}}(\overbar{t})=\phi_{\bar{a}}(\overbar{t})\phi_a(t)$ for all $t$ and $\overbar{t}$ hence
\[
[Y(a,z),Y(\bar{a},\bar{z})]=0\quad\quad\forall a\in V,\;\forall\bar{a}\in\bar{V}.
\]
This finishes the first part of the theorem.

To prove unitarity, we first of all have to show that if $h_a=\bar{h}_a=0$, then $a= \lambda\vac$ with $\lambda\in\complex$. If $h_a = 0$, then $T^* a= 0$ since $h\ge0$. Using Wightman inner product and unitarity of the representation of the restricted conformal group together with $[T^*,T] =2H$, we get
\[
	\lVert{T a}\rVert ^2 = 2h_a \lVert a\rVert^2=0.
\]
Thus, $a$ is annihilated by all of the $\slLie(2,\complex)$ generators $T,T^*$ and $H$. Hence, it is invariant under $\pslinear(2,\complex)$ and in particular under the Poincar\'e group. By the uniqueness of the vacuum vector, $a=\lambda\vac$ as required. Therefore, if the extra assumptions of the theorem hold, then unitarity follows by \cref{rmk:local adjoint if and only if hilbert adjoint} and \cref{thm:equivalence of unitarity for VOA}.
\end{proof}


If we also assume the existence of the energy-momentum tensor, we get two conformal vertex algebras.
\begin{cor} 
	A Wightman CFT gives rise to two commuting strongly-generated unitary positive-energy conformal vertex algebras. Moreover, if conformal weights are integers and the number of the generating fields of each conformal weight is finite, then these algebras are also unitary VOAs of CFT type.
\end{cor}

\begin{proof}
	We use the L\"uscher--Mack \cref{thm:Luscher-Mack final} to get an energy-momentum field $T(z)$ in vertex algebra sense. It gives rise to conformal vector $\nu=T(z)\vac|_{z=0}$. Similarly, for the antichiral part. The rest follows by \cref{thm:Wightman to VA}.
\end{proof}

\section{From  Vertex Algebras to Wightman CFT}\label{sec:VAs to Wightman}

We will show in this section that two unitary vertex operator algebras can be combined to give distributions satisfying all axioms of conformal Wightman distributions. Thus, we can use the Wightman Reconstruction Theorem \ref{thm:Wightman reconstruction} to get a Wightman CFT. The uniqueness of the vacuum vector follows if we assume that our VOAs have a single vacuum vector. The idea of the proof is to reverse the arguments of Kac's \cref{thm:Wightman to VA}.

We summarize this discussion in a theorem.

\begin{thm}\label{thm:VOA to Wightman} 
Given two unitary vertex operator algebras $V$ and $\bar{V}$, one can construct a Wightman CFT.
\end{thm}


We will see in the proof that the energy-momentum tensor of a VOA gives the existence of Wightman energy-momentum tensor \cref{axiom:W energy-momentum tensor} and it does not imply anything else. Thus, we get a corollary. 
\begin{cor}
Given two quasi-vertex operator algebras, one can construct a Wightman M\"obius CFT.
\end{cor}


Throughout this section, set the notation in accordance with Wightman framework
\[
	L_{-1}:= T,\quad L_1:=T^*,\quad L_0:= H,\quad \Omega:= \vac.
\]

{Now we introduce vertex algebra correlation functions which are well-known and can be found in  \cite{Frenkel1988} or \cite{Frenkel2004}. We have also used \cite{Carpi2015} for the discussion of the contragradient module.

Let $V$ be a vertex algebra and $V^*$ be the dual of $V$, i.e. the space of linear functions $\vphi: V\to \complex$.  Let $\langle\cdot,\cdot\rangle$ be the natural pairing between $V^*$ and $V$. Then for $a_1,\ldots,a_n,v\in V$ and $\vphi\in V^*$
\[
	\langle \vphi, Y(a_1,z_1)\ldots Y(a_n,z_n)v\rangle
\]
is a formal power series in $\complex[[z_1^{\pm},\ldots, z_m^{\pm}]]$. Such series are called \textbf{correlation functions} (in the sense of vertex algebra). Note that $v= Y(v,z)\Omega|_{z=0}$ by the vacuum axiom \cref{axiom:V vacuum}. Thus, it suffices to consider only the case $v=\Omega$.


\begin{prop}\label{prop:Vertex algebra correlation function expansion}
	Let $V$ be a vertex algebra, $\vphi\in V^*$ and let $a_1,\ldots, a_n\in V$. Then there exists a series
	\[
		M^{\vphi}_{a_1\ldots a_n}(z_1,\ldots,z_n)\in\complex [[z_1,\ldots,z_n]][(z_i-z_j)^{-1}]_{i\ne j}
	\]
	with the following property:

	For arbitrary permutation $\sigma$ of $\{1,\ldots,n\}$, the correlation function
	\[
		\langle \vphi, Y(a_{\sigma(1)}, z_{\sigma(1)})\ldots Y(a_{\sigma(n)}, z_{\sigma(n)}) \Omega\rangle
	\]
	is the expansion in $\complex ((z_{\sigma(1)}))\ldots((z_{\sigma(n)}))$ of $M^{\vphi}_{a_1,\ldots,a_n}(z_1,\ldots,z_n)$.
\end{prop}

\begin{proof}
	By the definition of vertex algebra \ref{def:Vertex algebra}, $Y(a,z)$ is a field and hence $\langle\vphi, Y(a,z)v\rangle\in\complex((z))$ for all $a,v\in V$. Thus, by induction
	\[
		\langle\vphi, Y(a_{\sigma(1)},z_{\sigma(1)})\ldots Y(a_{\sigma(n)},z_{\sigma(n)})\Omega\rangle\in\complex((z_{\sigma(1)}))\ldots((z_{\sigma(n)})).
	\]
	By locality \cref{axiom:V locality}, there exist positive even integers $N_{ij}\in 2\nat$ such that
	\[
		(z_i-z_j)^{N_{ij}} [ Y(a_i,z_i), Y(a_j,z_j)] =0.
	\]
	Thus, the series
	\[
		\prod_{i<j}(z_i-z_j)^{N_{ij}}\langle\vphi,Y(a_{\sigma(1)},z_{\sigma(1)})\ldots Y(a_{\sigma(n)},z_{\sigma(n)})\Omega\rangle
	\]
	is independent of the permutation $\sigma$. Furthermore, by \cref{axiom:V vacuum}, $Y(a,z)\Omega\in V[[z]]$ and combining this with permutation invariance we get that the series contains only non-negative powers of $z_i$, $1\le i\le n$. Therefore,
	\[
		\prod_{i<j} (z_i-z_j)^{N_{ij}} \langle\vphi, Y(a_{\sigma(1)},z_{\sigma(1)})\ldots Y(a_{\sigma(n)},z_{\sigma(n)}\Omega\rangle
	\]
	is the same as
	\[
		\prod_{i<j} (z_i-z_j)^{N_{ij}} \langle\vphi, Y(a_{1},z_{1})\ldots Y(a_{n},z_{n})\Omega\rangle
	\]
	in $\complex[[z_1,\ldots,z_n]]$. Dividing the last series by $\prod_{i<j} (z_i-z_j)^{N_{ij}}$, we obtain the required series $M^{\vphi}_{a_1\ldots a_n}(z_1,\ldots,z_n)\in\complex[[z_1,\ldots,z_n]][(z_i-z_j)^{-1}]_{i\ne j}$.
	
\end{proof}

\begin{rmk}
	For general $v\in V$, $M^{\vphi,v}_{a_1\ldots a_n}(z_1,\ldots,z_n)$ would belong to
	\[ 
	\complex[[z_1,\ldots, z_n]][z^{-1}_1,\ldots,z^{-1}_n,(z_i-z_j)^{-1}]_{i\ne j}.
	\]
\end{rmk}

Now consider a VOA. The grading allows us to define the \textbf{restricted dual} \cite{Frenkel1993} of a vertex operator algebra $V$ as
\[
	V':=\bigoplus_{n\in\integ} V_n^*,
\]
i.e. as the space of linear functionals on $V$ vanishing on all but finitely many $V_n$. Note that for $a,v\in V$ and $v'\in V'$
\[
	\langle v', Y(a,z)v\rangle\in\complex[z,z^{-1}]
\]
or equivalently
\[
\langle v', Y(a,z)\Omega\rangle\in\complex[z]
\]
because $Y(a,z)$ is a field and $v'$ belongs to the restricted dual $V'$. Hence, application of \cref{prop:Vertex algebra correlation function expansion} to a vertex operator algebra gives
\[
	M^{v'}_{a_1\ldots a_n}(z_1,\ldots,z_n)\in\complex [z_1,\ldots,z_n][(z_i-z_j)^{-1}]_{i\ne j}.
\]
Moreover, if we specialize from the case of arbitrary formal variables to the case $z_i\in\complex$, we obtain the following version of \cref{prop:Vertex algebra correlation function expansion}.

\begin{cor}\label{cor:Vertex operator algebra correlations rational}
	Let $V$ be a vertex operator algebra, $a_1,\ldots, a_n, v\in V$ and $v'\in V'$. For arbitrary permutations $\sigma$ of $\{1,\ldots,n\}$, the correlation functions 
	\[
		\langle v', Y(a_{\sigma(1)}, z_{\sigma(1)})\ldots Y(a_{\sigma(n)}, z_{\sigma(n)}) v\rangle
	\]
	with $z_i\in\complex$, $1\le i \le n$, are absolutely convergent to a common rational function $M^{v',v}_{a_1\ldots a_n}(z_1,\ldots,z_n)$ in the domains
	\[
		\left|z_{\sigma(1)}\right|>\ldots > \left|z_{\sigma(n)}\right|>0.
	\]
\end{cor}

{In light of \cref{cor:Vertex operator algebra correlations rational}, we will call the rational functions
\[
M^{v',v}_{a_1\ldots a_n}(z_1,\ldots, z_n)
\]
analytic extensions of VOA \textbf{correlation functions}.


{
The restricted dual $V'$ becomes a $V$-module by setting
\[
	\langle Y'(a,z)b',c\rangle = \langle b', Y(e^{zL_1} (-z^{-2})^{L_0}a,z^{-1})c\rangle\quad \forall a,c\in V,\; \forall b'\in V'.
\]
This formula determines the field $Y'(a,z)$ on $V'$ and implies that the map $a\mapsto Y'(a,z)$ is a $V$-module. See \cite[Sections 4.1 and 5.2]{Frenkel1993} for the definition and a proof. Note that the $V$-module structure on $V'$ depends not only on the vertex algebra structure of $V$, but also on $L_1$. We will call the module $V'$ the \textbf{contragradient module} and the fields $Y'(a,z)$ \textbf{adjoint vertex operators}. However, the endomorphisms $a'_{(n)}\in\End V$ of the formal series $Y'(a,z)=\sum_{n\in\integ} a'_{(n)} z^{-n-1}$ are not the adjoint endomorphisms of $a_{(n)}$. In particular, we have 
\[
	\langle L_n' a', c\rangle = \langle a', L_{-n}b\rangle\quad\quad a'\in V',\; b\in V,\; n\in \integ,
\]
with $L_n' = \nu'_{(n+1)}$. 
This implies that $L_0'=na'$ for $a'\in V_n^*$, i.e. $V'$ is a $\integ$-graded $V$ module. 

If we let $(\cdot, \cdot)$ to be an invariant bilinear form on $V$, then by \cref{rmk:invariant bilinear form orthogonal weight spaces} $(V_i, V_j) = 0$ if $i \ne j$. Hence,
\[
	(a,\cdot)\in V'\quad\forall a\in V
\]
and the map $a\mapsto (a,\cdot)$ is a module homomorphism from $V$ to $V'$. On the other hand, given a module homomorphism $\vphi:V\to V'$, the bilinear form
\begin{equation}\label{eq:contragradient to bilinear}
	(a,b) := \langle \vphi(a), b\rangle
\end{equation}
is invariant. By finite-dimensionality of the homogeneous subspaces and grading-preserving property, each $V$-module homomorphism from $V$ to $V'$ is injective if and only if it is surjective. We have proved a well-known result:

\begin{prop}
	The restricted dual $V'$ is isomorphic to $V$ as a $V$-module if and only if there exists a non-degenerate invariant bilinear form on $V$.
\end{prop}
}

{Now let $\left(V, (\cdot|\cdot)\right)$ be a unitary VOA with PCT operator $\theta$. By definition, we have that $(\cdot,\cdot):=(\theta\cdot|\cdot)$ is an invariant bilinear form on $V$. Fix $\vphi:V\to V'$ to be an isomorphism between $V$ and $V'$ and set
\[
\Omega' := \vphi(\Omega).
\]
Let
\[
M_{a_1\ldots a_n}(z_1,\ldots,z_n):=M^{\Omega',\Omega}_{a_1\ldots a_n}(z_1,\ldots,z_n)
\]
be VOA \textbf{vacuum expectation values (VEVs)}.
For 
\[
A=\begin{pmatrix}
a & b\\
c & d
\end{pmatrix}\in\slinear(2,\complex),
\] 
define
\[
	g_A(z) = \frac{az+b}{cz+d}
\]
to be a M\"obius transformation. In particular, set
\[
	g^{\lambda}_1(z) = \frac{z}{1-\lambda z},\quad g^{\lambda}_0(z) = {e^{\lambda} z},\quad
	g^{\lambda}_{-1}(z) = z+\lambda.
\]
\begin{prop}\label{prop:VOA expectation values covariant}
Let $V$ be a unitary VOA. Then the VOA vacuum expectation values of quasiprimary fields are M\"obius covariant.
\end{prop}
\begin{proof}
Using \cref{prop:exponentiation in vertex algebra}, \cref{axiom:V vacuum}, \Cref{eq:contragradient to bilinear,eq:L_n's hopping in invariant bilinear form} we have
\begin{flalign}
	\prod\limits_{i=1}^{n}&\left(\frac{d}{dz_i}\,g^{\lambda}_m(z_i)\right)^{h_i}\left\langle\Omega', Y(a_1,g^{\lambda}_m(z))\ldots Y(a_n,g^{\lambda}_m(z_n))\Omega\right\rangle=\nonumber\\
	&=\left\langle\Omega', e^{\lambda L_m} Y(a_1,z_1)\ldots Y(a_n,z_n) e^{-\lambda L_m} \Omega\right\rangle\nonumber\\
	&=\left(\Omega, e^{\lambda L_m} Y(a_1,z_1)\ldots Y(a_n, z_n) \Omega\right)\nonumber\\
	&=\left(e^{\lambda L_{-m}}\Omega,Y(a_1,z_1)\ldots Y(a_n, z_n) \Omega\right)= \left(\Omega,Y(a_1,z_1)\ldots Y(a_n, z_n) \Omega\right)\nonumber\\
	&=\left\langle\Omega', Y(a_1,z_1)\ldots Y(a_n,z_n) \Omega\right\rangle.
\end{flalign}
Since the transformations of the form $g^{\lambda}_1(z)$ and $g^{\lambda}_{-1}(z)$ generate $\pslinear(2,\complex)$ by \cref{prop:generators of PSL2}, it follows that
\begin{equation}
	M_{a_1\dots a_n} (z_1,\ldots,z_n) = \prod\limits_{i=1}^{n}\left(\frac{d}{dz_i}\,g_A(z_i)\right)^{h_i} M_{a_1\ldots a_n} \left(g_A(z_1),\ldots, g_A(z_n)\right)
\end{equation}
$\forall A\in\slinear(2,\complex)$, as required.
\end{proof}
}



Now we restrict the variables $z_i$ to the open unit disk in $\complex^n$. We use the inverse Cayley transform
\[
	t=2i\frac{1-z}{1+z}
\]
to map the open unit disk to the open upper half-plane $\im t> 0$ (Figure \ref{fig:Cayley transform}). We define the transformed fields for quasiprimary 
vectors $a\in V_{h}$ on the upper-half plane as
\begin{equation}\label{eq:transformation of fields to lightcone from disk}
	\phi_a(t) :=\frac{1}{2\pi} \left(\frac{2i}{2i+t}\right)^{2h} Y(a,z)
\end{equation}
with $z=(1+it/2)(1-it/2)^{-1}$. We also define the corresponding correlation functions as
\[
	W_{a_1\ldots a_n}(t_1,\ldots,t_n) := \frac{1}{(2\pi)^n}\prod\limits_{j=1}^{n} \left(\frac{2i}{2i+t_j}\right)^{2h_j} M_{a_1\ldots a_n}(z_1,\ldots,z_n)
\]
and call their limit as $\im t_i\to 0\,$ \textbf{lightcone VEVs}.

\begin{prop}\label{prop:Mobius covariance of lightcone VEVs}
The lightcone VEVs are M\"obius covariant tempered distributions.
\end{prop}
\begin{proof}
We prove temperedness first. The correlation functions $W$ are rational by rationality of $M$'s and the fact that the inverse Cayley transform is rational. We let $\im t_i \to 0$ and use
\[
	\lim\limits_{x\to x_0}f(x)g(x) = \lim\limits_{x\to x_0}f(x)\cdot\lim\limits_{x\to x_0}g(x)
\]
with $f$ being the numerator of $W$ and $g$ one over the denominator. The limit of the numerator simply returns a polynomial, whereas one over the denominator gives a tempered distribution containing factors of the form
\[
\frac{1}{(t_i-t_j\pm i\veps)^k}\,, \quad k\in\nat_0.
\]
Since the product of a function of at most polynomial growth with a tempered distribution is a tempered distribution, we get that the lightcone VEVs are tempered distributions.


The inverse Cayley transform can be viewed as a change of basis matrix
\[
\begin{pmatrix*}[c]
	-2i & 2i\\
	  1	& 1
\end{pmatrix*}.
\]
Thus, defining
\begin{equation}\label{eq:P, K, D on the plane}
	P= \frac{1}{4}(2L_0+L_{-1}+L_1),\quad K = 2 L_0 - L_{-1}- L_1,\quad D = \frac{1}{2i}(L_1-L_{-1})
\end{equation}
and using \cref{prop:exponentiation in vertex algebra} we see that
\begin{alignat}{3}
	&e^{i\lambda P}\quad&&\text{maps}\quad &&t\mapsto t+\lambda,\\
	&e^{i\lambda D}\quad&&\text{}\quad &&t\mapsto e^{\lambda}\,t,\\
	&e^{i\lambda K}\quad&&\text{}\quad &&t\mapsto 	\frac{t}{1-\lambda\, t}.
\end{alignat}
Therefore, the operators $P$, $D$ and $K$ are infinitesimal generators of translations, dilations and special conformal transformations, respectively. In a unitary VOA,
$L_0$ is self-adjoint and $L_1$ is the adjoint of $L_{-1}$ and vice versa. Hence, $P$, $D$ and $K$ are self-adjoint. Therefore, by Stone's Theorem $U_q(A):=e^{iqA}$ are strongly continuous one-parameter unitary groups, where $A=P,D$ or $K$, and $q\in\reals$. By \cref{axiom:V vacuum} and \cref{prop:exponentiation in vertex algebra} we have $L_{-1}\Omega = L_0\Omega = L_1\Omega =0$. Hence, 
\begin{equation}
e^{i q_1 P}\Omega =e^{i q_2 D}\Omega=e^{i q_3 K}\Omega=\Omega\quad \forall q_1, q_2,q_3\in\reals,
\end{equation}
i.e. the vacuum is fixed under global conformal transformations. Hence, by \cref{prop:VOA expectation values covariant} lightcone vacuum expectation values of quasiprimary fields $W_{a_1\ldots a_n}$ are M\"obius covariant.
\end{proof}

}

{
We now take a second vertex operator algebra $(\bar{V}, \bar{Y}, \bar{\Omega}, \bar{\nu})$ and mimic the construction of full field algebras \cite{Huang2007}. Let 
\[
		V_{f} := V\otimes \bar{V}
\]
be the full vector space and let the full vertex operators be
\[
	\mathcal{Y}_{a,\bar{a}}\left(z,\bar{z}\right) := Y(a,z)\otimes \bar{Y}(\bar{a},\overbar{z}).
\]
Here we identify $\bar{z}$ with the complex conjugate of $z$. Then the full vertex operators act as
\[
	\mathcal{Y}_{a,\bar{a}}(z, \bar{z})(v\otimes \bar{v}) = Y(a,z) v \otimes Y(\bar{a},\bar{z})\bar{v},
\]
$\forall a,v\in V$, $\forall \bar{a},\bar{v}\in \bar{V}$.
We also define an inner product on $V_f$ by
\[
( a\otimes \bar{a}|b\otimes \bar{b})_f = (a|b)(\bar{a}|\bar{b}).
\]
By \cite[Prop. 2.9]{Dong2014} a tensor product of unitary VOAs is a unitary VOA with conformal vector $\bm{\nu} = \nu\otimes\bar{\Omega} + \Omega\otimes\bar{\nu}$ and the vacuum vector $\mathbf{1} := \Omega\otimes\bar{\Omega}$.

The action of the full vertex operators implies that the full VOA correlation functions are
\[
	\mathcal{M}_{a_1,\bar{a}_1\ldots a_n,\bar{a}_n}(z_1,\bar{z}_1,\ldots,z_n,\bar{z}_n) = M_{a_1\ldots a_n}(z_1,\ldots,z_n) \bar{M}_{a_1\ldots a_n} (\bar{z}_1,\ldots, \bar{z}_n).
\]
In particular, since we have proved in \cref{cor:Vertex operator algebra correlations rational} that $M$'s are symmetric if $z_i\ne z_j$, the full VOA correlation functions are also symmetric if $z_i\ne z_j$, i.e.
\begin{align*}
		\mathcal{M}_{a_1,\bar{a}_1\ldots a_i,\bar{a}_i\,a_{i+1},\bar{a}_{i+1}\ldots a_n,\bar{a}_n}&(\vec{z}_1,\ldots,\vec{z}_i,\vec{z}_{i+1},\ldots,\vec{z}_n) =\\ 
		&\mathcal{M}_{a_1,\bar{a}_1\ldots a_{i+1},\bar{a}_{i+1} \, a_i,\bar{a}_i\ldots a_n,\bar{a}_n}(\vec{z}_1,\ldots,\vec{z}_{i+1},\vec{z}_i,\ldots,\vec{z}_n)
\end{align*}
with $z_i\ne z_{i+1}$ and $\vec{z}=(z,\bar{z})$.

Applying the inverse Cayley transform to the fields \eqref{eq:transformation of fields to lightcone from disk} we get full lightcone fields
\[
	\Phi_{a,\bar{a}}\left(x^0,x^1\right) := \phi_a\left(t\right)\otimes	 \bar{\phi}_{\bar{a}}(\overbar{t}),
\]
where $t=x^0-x^1$, $\bar{t}=x^0+x^1$, with the corresponding \textbf{full Wightman VEVs}
\begin{align}\label{eq:full Wightman distributions definition}
	\fullW_{a_1,\bar{a}_1 \ldots a_n, \bar{a}_n}(x_1,\ldots, x_n) &= 
	W_{a_1\ldots a_n}(t_1,\ldots t_n) \bar{W}_{\bar{a}_1\ldots\bar{a}_n}(\bar{t}_1,\ldots, \bar{t}_n),
\end{align}
where 
\[
\im t_i, \im  \bar{t}_i\to 0
\]
 is implicit. Moreover, we define the operators on the Minkowski plane as
\begin{subequations}\label{eq:full operators on the plane}
\begin{alignat}{2} 
	P_0 &= P\otimes \id + \id\otimes\bar{P},\quad\quad\; P_1 &&= -P\otimes\id+\id\otimes\bar{P},\\
	K_0 &=-K\otimes\id-\id\otimes\bar{K},\quad K_1 &&=-K\otimes\id+ \id\otimes\bar{K}.
\end{alignat}
\end{subequations}

Now we are ready to prove the main theorem of this section.

\begin{reptheorem}{thm:VOA to Wightman}
Given two unitary vertex operator algebras $V$ and $\bar{V}$, one can construct a Wightman CFT.
\end{reptheorem}

\begin{proof}
	Conformal covariance and temperedness obviously hold for full Wightman vacuum expectation values by \cref{prop:Mobius covariance of lightcone VEVs} and \cref{eq:full operators on the plane}.
	
	Conformal covariance also includes translation invariance and hence we can define for $n\ge 2$
	\[
		w(\zeta_1,\ldots, \zeta_{n-1}) = W_{a_1\ldots a_n}(t_1,\ldots, t_n),\quad \zeta_i = t_i- t_{i+1}.
	\]
	We have
	\[
		w(\zeta_1,\ldots, \zeta_n)= \int \hat{w}(p_1,\ldots,p_n) e^{i \sum p_j\,t_j} \mathrm{d} p,\quad \forall \im t_j >0.
	\]
	Thus, $\hat{w}(p_1,\ldots,p_n) = 0 $ if at least one of $p_j<0$. Combining this with the definition of full Wightman distributions \eqref{eq:full Wightman distributions definition}, we see that this is precisely the spectrum property \cref{axiom:WD spectrum}. 
	
	Now write
		\[
		\fullW(x_1,\ldots, x_n)=	\fullW_{a_1,\bar{a}_1\ldots a_n,\bar{a}_n}(x_1,\ldots, x_n).
		\]
		We want to show that
		\begin{equation}\label{eq:full Wightman locality}
			\fullW(x_1,\ldots, x_i, x_{i+1},\ldots, x_n) = \fullW(x_1,\ldots, x_{i+1}, x_i,\ldots, x_n)
		\end{equation}
		if $(x_i-x_{i+1})^2<0$. It holds that $(x_i-x_{i+1})^2= (t_i - t_{i+1})(\bar{t}_i-\bar{t}_{i+1})$ and so it suffices to prove that \eqref{eq:full Wightman locality} holds whenever $(t_i - t_{i+1})<0$ and $(\bar{t}_i-\bar{t}_{i+1})>0$. But this follows from the symmetry of VOA correlation functions, since we have identified $\bar{z}$ with the complex conjugate of $z$.
		
		Wightman positivity \cref{axiom:WD positivity} was used to prove the existence of positive-definite scalar product on the Hilbert space constructed in the Wightman Reconstruction Theorem \ref{thm:Wightman reconstruction}. However, a unitary VOA already has an inner product which can be used for Wightman reconstruction so \cref{axiom:WD positivity} is unnecessary.
		
		Combining all of the above observations, we get a Wightman M\"obius CFT.
		
		To get the energy-momentum tensor, let
			\[
			\Theta(t)  = \frac{1}{2\pi} \left(\frac{2i}{2i+t}\right)^{4} Y(\nu, z),\quad 
			\bar{\Theta}(\bar{t})  = \frac{1}{2\pi} \left(\frac{2i}{2i+\bar{t}}\right)^{4} \bar{Y}(\bar{\nu},\bar{z})
			\]
			be fields acting on $V$ and $\bar{V}$, respectively, where $\nu$ is the conformal vector of $V$ and $\bar{\nu}$ of $\bar{V}$. Set
			\begin{align*}
				T_{00}(x^0,x^1)&=T_{11}(x^0,x^1) = \Theta(t)\otimes\id +\id \otimes\bar{\Theta}(\overbar{t}),\\
				T_{01}(x^0,x^1)&=T_{10}(x^0,x^1)=\id\otimes\bar{\Theta}(\overbar{t})-\Theta(t)\otimes\id.
			\end{align*}
			Then the full Wightman VEVs containing these fields give rise to Wightman fields satisfying all requirements of \cref{axiom:W energy-momentum tensor}. In particular, self-adjointness follows from the fact that after the reconstruction we have a unitary representation of M\"obius group. Thus, $L_0 = L_0^*$ and
			\[
				L_0 =\int\limits_{-\infty}^{+\infty} \left(1+\frac{t^4}{4}\right)\,\Theta_W(t)\, \mathrm{d} t
			\]
			give the required result, where $\Theta_W(t)$ denotes the chiral part of the Wightman energy-momentum tensor in lightcone coordinates. Similarly, for $\bar{\Theta}_W(\overbar{t})$.

\end{proof}

}

	\newpage
	\pagestyle{plain}
	
	\bibliographystyle{alpha} 
	\bibliography{./referencesBib/references}
\end{document}